\documentclass{article}
\usepackage[utf8]{inputenc}
\usepackage{graphicx,algpseudocode}
\usepackage{amssymb, amsmath, amsthm}
\usepackage{mathtools}
\usepackage{thmtools}
\usepackage{bbm}
\usepackage{enumitem}   
\usepackage{wrapfig}
\usepackage{chngpage}
\usepackage{breqn}
\usepackage{color,soul}
\usepackage[title]{appendix}
\usepackage{multirow}
\usepackage{caption}
\usepackage{natbib}

\usepackage[authormarkup=none,final]{changes}
\definechangesauthor[name=DLS, color=purple]{ds}

\definechangesauthor[name=KK, color=blue]{kk}

\newcommand\scalemath[2]{\scalebox{#1}{\mbox{\ensuremath{\displaystyle #2}}}}
\allowdisplaybreaks

\usepackage{todonotes}

\newcommand{\R}{{\mathbb{R}}}

\newcommand{\E}{{\mathbb{E}}}
\newcommand{\V}{{\text{Var}}}

\newcommand{\bvar}[1]{\mathbf{#1}} 
\newcommand\numberthis{\addtocounter{equation}{1}\tag{\theequation}}

\theoremstyle{plain}
\newtheorem{definition}{Definition}
\newtheorem{theorem}{Theorem}
\newtheorem{proposition}{Proposition}
\newtheorem{Lemma}{Lemma}
\newtheorem{Cor}{Corollary}

\theoremstyle{remark}

\newtheorem{ex}{Example}

\newtheorem{case}{Case}
\newtheorem{assumption}{Assumption}

\renewcommand{\thmcontinues}[1]{continued}

\title{Unbiased estimation for additive exposure models}

\author{Kelly Kung and Daniel L. Sussman}

\begin{document}
\maketitle
\begin{abstract}
    Causal inference methods have been applied in various fields where researchers want to estimate treatment effects.
    In traditional causal inference settings, one assumes that the outcome of a unit does not depend on treatments of other units.
    However, as causal inference methods are extended to more applications, there is a greater need for estimators of general causal effects.
    We use an exposure mapping \citep{aronow2017estimating} framework to map the relationship between the treatment allocation and the potential outcomes.
    Under the exposure model, we propose linear unbiased estimators (LUEs) for general causal effects under the assumption that treatment effects are additive.
    Additivity provides statistical advantages, where contrasts in exposures are now equivalent, and so the set of estimators considered grows. 
    We identify a subset of LUEs that forms an affine basis for LUEs, and we characterize optimal LUEs with minimum integrated variance through defining conditions on the support of the estimator. 
    We show, through simulations that our proposed estimators are fairly robust to violations of the additivity assumption, and in general, there is benefit in leveraging information from all exposures.
\end{abstract}

\section{Introduction}
The goal of many researchers, regardless of field, is often to understand the effect of a particular treatment or intervention; hence, the rise of applications of causal inference methods. 
Traditionally, one estimates the direct effect of a single treatment under the Stable Unit Value Treatment Assumption (SUTVA) \citep{rubin1974estimating} of which the assumption of no interference is crucial. 
However, as we extend causal inference methods to different fields, SUTVA may no longer hold, and so the need for estimation of general treatment effects grows.
For example, there can be multiple treatments or the treatment can affect the outcome in different ways.
Therefore, there is a need for estimators that can be used to estimate causal effects in general settings.

As we stray away from the classical settings of causal inference where SUTVA holds, the estimation of causal effects becomes more difficult.
We have to consider not only how the treatment directly affects the outcome but also how the treatment potentially indirectly affects the outcome.
Since there can be nuances in how a treatment allocation affects the potential outcomes, we use an exposure mapping \citep{aronow2017estimating} to map the relationship between the treatment allocation and the potential outcomes.
Given an exposure mapping, we assume that the potential outcomes depend on the treatment allocation only through the exposures.
We then estimate general causal effects under exposure models.

In general, one prefers to make fewer assumptions so that results are generalizable. 
However, we assume that treatment effects are additive, which provides statistical advantages.
Under additivity, contrasts of potential outcomes under different exposures are equivalent, and so there are fewer contrasts to estimate.
Furthermore, when additivity holds the set of unbiased estimators grows as exposures that may not be immediately related to the estimand can be employed.

In this paper, we propose linear unbiased estimators for causal effects under the additive exposure assumption in an experimental setting.
We characterize the set of linear unbiased estimators and define an affine basis for the set of linear unbiased estimators.
We further characterize a set of optimal estimators with minimum integrated variance.
Lastly, we compare the proposed optimal linear unbiased estimators with other linear unbiased estimators through a series of simulations under various settings. 

We first introduce the background and notation in Section~\ref{sec:background}.
In Sections~\ref{sec:exposure_models} and \ref{sec:lue}, we define exposure models and define linear constraints for unbiased estimators under additivity.
We introduce a class of atomic linear unbiased estimators in Section~\ref{sec:alue} and show that with another class of estimators, they form an affine basis for the set of linear unbiased estimators.
In Section~\ref{sec:mivlue}, we characterize a set of optimal estimators, in which linear unbiased estimators have minimum integrated variance.
Lastly, we evaluate the proposed estimators in different simulation settings in Section~\ref{sec:simulations}.

\section{Background}\label{sec:background}
Early work in causal inference has been done by estimating the treatment effect in randomized experiments under the assumption that a unit's outcome is only affected by the treatment received by that unit \citep{neyman1923application}.
\citet{rubin1974estimating} further formalized these ideas with the Stable Unit Treatment Value Assumption (SUTVA), in which (1) \textit{no interference}: units' outcomes did not depend on other units' treatments and (2) \textit{consistency}: there were no multiple versions of a treatment \citep{rubin1980randomization}.

There has been a growing body of work in relaxing the no interference assumption of SUTVA in which we assume that units' outcomes do not depend on other units' treatments. 
This is likely because there are many settings in which \textit{interference} or \textit{spillover} effects \citep{cox1958planning, rubin1980randomization} may be present.
Early work in interference began with the assumption that treatment effects may spill over through time, focusing on \textit{residual} effects that may be present from the preceding time point \citep{grizzle1965two, kershner1981two}.
Later, spatial interference attracted attention, where neighboring units or units within the same block may be dependent \citep{besag1986statistical, david1996designs}.
Since then, interference has been extended to settings of \textit{partial interference}, where units within a cluster may be dependent but units between clusters are assumed to be independent \citep{sobel2006randomized, rosenbaum2007interference, hudgens2008toward, tchetgen2012causal}.

More recently, there has been a growing interest in estimating causal effects in the presence of interference in networks \citep{ugander2013graph, eckles2017design, athey2018exact, aronow2017estimating, sussman2017elements, forastiere2021identification}.
This is because networks can be used to represent relationships between units and interference effects may be passed through the connections of the network.
Furthermore, the rise of social media has enabled researchers to better observe these connections.

Various assumptions and models have been proposed for network interference.
A common variant is to assume that a unit's outcome can be affected by the treatments of units up to $k$ connections away for some $k$ \citep{athey2018exact}.
As a running example, we will focus on the assumption that interference only occurs for neighboring units.

We also focus on experimental settings, where the treatment assignment probabilities are known.
Our work extends \citet{aronow2017estimating}, who proposed unbiased estimators for causal effects under general interference, and \citet{sussman2017elements}, who proposed unbiased estimators for the direct treatment effect with minimum integrated variance under network interference.
Under an exposure model \citep{aronow2017estimating} in which a treatment allocation is assigned to exposures through an \textit{exposure mapping}, \citet{aronow2017estimating} proposed two-term unbiased estimators for estimands of interest using Horvitz-Thompson estimators \citep{horvitz1952generalization}.
We propose linear unbiased estimators for unit-level causal effects, but we deviate from \citet{aronow2017estimating} in that we assume that treatment effects are additive.
The additivity constraint enables flexibility in estimation through the fact that different estimands under different exposures are equivalent.
Assuming additivity, our proposed linear unbiased estimators may place non-zero weights
on exposures that are ``seemingly unrelated'' to the estimand of interest.
Furthermore, we deviate from \citet{sussman2017elements} in that we estimate general treatment effects, which include both direct and indirect treatment effects.
However, like \citet{sussman2017elements}, we further characterize an optimal subset of linear unbiased estimators that have minimum integrated variance.

\subsection{Potential Outcomes Framework}\label{po sec}
Consider a randomized experiment with $n$ units that are together assigned a treatment allocation $\mathbf{z} \in \{0, \dotsc, m\}^n$ where $z_i \in \{0, \dotsc, m\}$ represents the treatment that unit $i$ receives.
The experimental design of a randomized trial is given by the probability of a treatment allocation, denoted by $p: \{0, \dotsc, m\}^n \to [0,1]$. 
Since we focus on a randomized experiment setting, we assume that the design is known.
The treatment allocation $\mathbf{z}$ provides information, such as unit treatment assignments, number of treated units, etc., which can be used to determine a unit's outcome. 

We use the Rubin causal model \citep{rubin1974estimating} or the \textit{potential outcomes framework} to estimate treatment effects.
We denote the potential outcome of patient $i$ under treatment allocation $\mathbf{z}$ as $Y_i(\textbf{z}) \in \R$.
Note, however, we only observe the treatment allocation $\mathbf{z}^{obs}$, and so we only observe one potential outcome for unit $i$, namely $Y_i(\mathbf{z}^{obs})$.
We denote the observed outcome of unit $i$ as $Y_i^{obs} = Y_i(\mathbf{z}^{obs})$.
This is the Fundamental Problem of Causal Inference \citep{holland1986statistics}.
Since only one potential outcome is observed for a unit, estimating treatment effects becomes a missing data problem, where we impute missing potential outcomes to estimate treatment effects.

\section{Exposure Models} \label{sec:exposure_models}
The treatment allocation $\mathbf{z}$ provides information, such as unit treatment assignments, number of treated units, etc., which can be used to determine a unit's outcome. 
While in general $Y_i(\mathbf{z})$ depends on all of $\mathbf{z}$, we often assume that the outcome only depends on specific aspects of the treatment allocation. 
For example, under the stable unit treatment value assumption (SUTVA), the outcome of a unit only depends on its treatment \citep{rubin1974estimating}. 
That is, $Y_i(\mathbf{z}) = Y_i(\mathbf{z}')$ whenever $z_i = z'_i$.
To capture the dependencies of potential outcomes on treatment allocations, \citet{aronow2017estimating} proposed exposure models as an alternative representation of the potential outcomes that can account for these pathologies while still limiting the complexity of the model.

Exposure models are given by exposure mappings, which are used to capture all the information needed from a treatment allocation to determine a unit's potential outcome:
\begin{definition}[Exposure mapping]
Let $\mathcal{E}$ denote the set of exposures. For each unit $i$, an exposure mapping $f(i,\cdot): \{0, \dotsc, m\}^n \to \mathcal{E}$ maps each treatment allocation to an exposure in the set $\mathcal{E}$.
\end{definition}
\noindent  Exposure mappings are flexible and can be defined in various ways. 
However, we assume in this paper that the exposure mapping is known.
The goal of an exposure mapping is to capture all the information needed to determine a unit's outcome while reducing the number of possible potential outcomes for unit $i$ from $(m+1)^n$ to $|\mathcal{E}|$, the cardinality of $\mathcal{E}$.
This motivates the following assumption.

\begin{assumption}[\citet{aronow2017estimating}]\label{assump:exp_model}
We assume that for any pair $\mathbf{z},\mathbf{z}'\in \{0, \dotsc, m\}^n$, $f(i,\mathbf{z}) = f(i,\mathbf{z}') = \vec{e}$ implies that $Y_i(\mathbf{z}) = Y_i(\mathbf{z}')$.
That is, we can write
\begin{equation}\label{contrast}
    Y_i(\mathbf{z}) = Y_i(\Vec{e}).
\end{equation}
\end{assumption}
\noindent Assumption \ref{assump:exp_model} states that the potential outcome of a unit is determined only by its exposure, and so we assume that potential outcomes are dependent on treatment allocations through the exposures.
Note that Assumption \ref{assump:exp_model} holds regardless of whether SUTVA holds.

\begin{ex}[SUTVA]
\label{sutva_ex}
Under SUTVA, a unit's outcome only depends on its own treatment assignment.
Here, $\mathcal{E}= \{0, \dotsc, m\}$, and the exposure mapping is given by $f(i,\mathbf{z}) = z_i \in \{0, \dotsc, m\}$.
Potential outcomes are then given by $Y_i(\mathbf{z}) = Y_i(z_i)$. 
\end{ex}

\begin{ex}[Network Interference]\label{network_interference_ex}
Consider a network amongst the $n$ units, given by the $n \times n$ adjacency matrix $A$. 
Suppose that a unit's outcome can depend on its own treatment assignment, which is binary, and the treatment assignments of other units in the network. 
In particular, suppose that the potential outcome of a unit only depends on the number of neighbors that are treated and not necessarily which ones \citep{sussman2017elements}.
Note that SUTVA no longer holds since network interference is present.
Here, $\mathcal{E} = \{0, \dotsc, n-1\} \times \{0,1\}$, and the exposure mapping is defined as $f(i, \mathbf{z}) = (d_i^{\mathbf{z}}, z_i)$, where $d_i^{\mathbf{z}} = (A^T \mathbf{z})_i$ is the number of treated neighbors or the \textit{treated degree}.
Note that for a unit $i$, the treated degree $d_i^{\mathbf{z}} \in \{0, \dotsc, d_i\}$ where $d_i$ is the degree of unit $i$.
\end{ex}

Using potential outcomes given by the exposure mappings, we define causal effects under the exposure model framework.
In general, a \emph{causal effect} is given by the difference in the potential outcome under one exposure and the potential outcome under another exposure. 
We focus on the unit-level causal effect of exposure $\Vec{e} \in \mathcal{E}$ compared to $\Vec{e}{\,'} \in \mathcal{E}$:
\begin{equation}
    \tau_i(\Vec{e}, \Vec{e}{\,'}) = Y_i(\vec{e})-Y_i(\Vec{e}{\,'}).
\end{equation} 
Since we focus on unit-level effects, we simplify the notation by dropping the subscript $i$ throughout the rest of the paper. 
Following prior work \citep{aronow2017estimating}, we use unit-level causal effects to estimate the average causal effects by averaging unit-level estimates.
Exposures are flexibly defined, but they are often represented with multiple exposure components.
For example, exposures in Example \ref{network_interference_ex} are given by two exposure components: $\mathcal{E} = \{0, \dotsc, n-1\} \times \{0,1\}$.
We use an \emph{exposure vector} to denote exposures with multiple components:

\begin{definition}[Exposure vector]
As the exposure set is finite, without loss of  generality, we assume the exposure set has the form $\mathcal{E} = \{0, \dotsc, m_1\} \times \dotsb \times \{0, \dotsc, m_{K}\}$, where $K \geq 1$ is the number of \emph{exposure components}. 
Exposures, denoted by $\Vec{e} \in \mathcal{E}$, are hence given by exposure vectors: $\Vec{e} = (e_{1}, \dotsc, e_{K}) \in \mathcal{E}$.
\end{definition}
\noindent Since exposure vectors are multi-dimensional vectors in the real space, vector operations can be applied to exposures.
For example, we can take the difference between exposures, which is given by the difference in the exposure components.
We define the vector of all zeros, denoted as $\Vec{e} = \Vec{0}$, as the {\em baseline} exposure.
We interpret the exposure components as different information given by the exposure mapping.
For example, in Example \ref{network_interference_ex}, the first exposure component corresponds to the number of treated neighbors for the unit, and the second exposure component corresponds to the treatment assigned to the unit.

The set of estimands for exposure causal effects is given by the contrasts in exposures.
In general, potential outcomes under an exposure can be decomposed into the baseline, the corresponding direct effects for each exposure component, and interactions between the effects from multiple exposure components.
As the number of exposures, and especially the number of exposure components, increase, the number of interaction effects become large.
Instead, we assume that additivity holds:

\begin{assumption}(Additivity)\label{assump:additivity}
Consider exposure vectors $\Vec{e}, \Vec{e}{\,'} \in \mathcal{E}$. 
Exposures are additive if, whenever $\Vec{e} - \Vec{e}{\,'} > \Vec{0}$,
\begin{equation}
    Y(\Vec{e}) - Y(\Vec{e}{\,'}) = Y(\Vec{e} - \Vec{e}{\,'}) - Y(\Vec{0}).
\end{equation} 
\end{assumption}
\noindent \noindent Under additivity, there are no interaction effects. 
That is, the difference in potential outcomes given two different exposures only depends on the difference in exposure components.
We can then isolate the effect of the $k$th exposure component by removing the effects of all other components.
To do this, we can add and subtract potential outcomes under different exposures so that the net value of all other exposure components besides the $k$th exposure component is zero.
Additivity provides statistical advantages since certain contrasts are now equivalent, such as 
\begin{align*}
    Y(m_1, e_2, \dotsc, e_K) - Y(0, e_2, \dotsc, e_K) = Y(m_1, e'_2, \dotsc, e'_K) - Y(0, e'_2, \dotsc, e'_K)
\end{align*}
for $e_k \neq e'_k$ for $k \in \{2, \dotsc, K\}$.
Since contrasts in potential outcomes under different exposures are equivalent under additivity, the number of contrasts we consider is then reduced to $\sum_{k=1}^K m_k$. 

Under additivity, there are no interaction effects, and so we can write the potential outcome under exposure $\Vec{e} = (e_1, \dotsc, e_K)$ as:
\begin{align}\label{po_written_by_exposure_components}
    Y(e_1, e_2 ,\dotsc, e_K) &= Y(0, \dotsc, 0)\\
    &+\left[Y(e_1, 0, \dotsc, 0) - Y(0, \dotsc, 0)\right] \nonumber \\
    &+ \left[Y(0, e_2, 0, \dotsc, 0) - Y(0,  \dotsc, 0)\right] \nonumber \\
    &\dotsc \nonumber\\
    &+ \left[Y(0, \dotsc, 0, e_K) - Y(0,  \dotsc, 0) \right],
\end{align}
\noindent  where the first summand indicates the baseline and the other summands indicate the various causal effects for the $k$th exposure component at level $e_k \in \{1, \dotsc, m_k\}$.
We denote the unit-level causal effect for the $k$th exposure at level $j_k \in \{1, \dotsc, m_k\}$ as:
\begin{equation}\label{eq:unit_level_eff}
\theta_{k,j_k} = Y(0, \dotsc, 0, j_k, 0, \dotsc, 0) - Y(0, \dotsc, 0).
\end{equation}
Let the parameter set, denoted by $\Theta$, contain the baseline parameter, denoted as $\alpha = Y(0, \dotsc, 0)$, and parameters $\theta_{k,j_k}$ for all $k \in \{1, \dotsc, K\},j_k \in \{1, \dotsc, m_k\}$.
Under additivity, potential outcomes are given as: 
\begin{equation}\label{eq:po_sum_of_param}
    Y(\Vec{e}) = \alpha + \sum_{k=1}^K \sum_{j_k = 1}^{m_k}\theta_{k,j_k} \mathbb{I}\{e_k = j_k\}.
\end{equation}

\begin{ex}[continues=sutva_ex]
Under SUTVA with $m$ levels or variants of treatment, we define $\Vec{e} = z_i \in \{0, \dotsc, m\}$.
The unit-level causal effect for the first (and only) exposure component when the unit has treatment $m$ versus when the unit is not treated is given by $\theta_{1,m} = Y(m) - Y(0)$. 
\end{ex}

\begin{ex}[continues=network_interference_ex]
Under network interference with binary treatment, we define $\Vec{e} = (d_i^{\mathbf{z}}, z_i)$.
The causal effect of the first exposure component when all of unit $i$'s neighbors are treated versus when none of unit $i$'s neighbors are treated is given by $\theta_{1,d_i} = Y(d_i, 0) - Y(0,0)$.
Here, $\theta_{1,d_i}$ corresponds to the \emph{unit-level interference effect}.
Note that we defined $\theta_{1,d_i}$ using an estimand with exposures where $z_i = 0$. 
However, under additivity, contrasts in potential outcomes under different exposures are equivalent, and so $Y(d_i, 1) - Y(0,1)$ is also an estimand for the unit-level interference effect.
\end{ex}

\section{Linear Unbiased Estimators} \label{sec:lue}
In this section, we introduce estimators for the unit-level causal effect. 
Without the loss of generality, for the rest of this paper, we focus on estimating the effect for a single unit when the first exposure component is $m_1$, compared to baseline.
No generality is lost since we can remap the exposures to a new exposure set where the $k$th component is mapped to the first component and the $e_k$th level is mapped to the maximum $m_k$.

Linear estimators of the unit-level causal effect of the first exposure component are of the form
\begin{align*}
    \hat{\theta}_{1,m_1} =  w(\mathbf{z}^{obs})Y(\Vec{e}^{\,obs}), 
\end{align*}
where $w: \{0, \dotsc, m\}^n \to \R$ is a weight function depending on the treatment allocation $\mathbf{z}^{obs}$ and $Y(\Vec{e}^{\,obs})$ is the outcome under observed exposure $\Vec{e}^{\,obs}$.
We further consider linear estimators with weights that depend only on the unit's exposure, i.e. $w: \mathcal{E} \to \R$.
We denote the support of $w$, or equivalently of the estimator $\hat{\theta}_{1,m_1}$, as $\mathrm{supp}(\hat{\theta}_{1,m_1}) = \{\Vec{e} \in \mathcal{E}: w(\Vec{e}) \neq 0\}$.
Hence, the linear estimators we consider are of the form
\begin{align}\label{eq:lue}
    \hat{\theta}_{1,m_1} =  w(\Vec{e}^{\,obs})Y(\Vec{e}^{\,obs}).
\end{align}
Linear estimators include Horvitz-Thompson inverse propensity score weighting estimators \citep{horvitz1952generalization}. 
The Horvitz-Thompson (HT) estimator for a potential outcome is given by
\begin{align}
    HT_{\Vec{e}} = \frac{Y(\Vec{e}^{\,obs})}{p(\Vec{e})}\mathbb{I}\{\Vec{e}^{\,obs} = \Vec{e}\},
\end{align}   
where $p(\Vec{e}) = \mathbb{P}(\Vec{e}^{\,obs} = \Vec{e})$ is the probability of observing the exposure $\Vec{e}$, which is given by the design probabilities. 
Since the experimental design is known, the probabilities of exposures are also known.
Furthermore, $\mathbb{I}\{\Vec{e}^{\,obs} = \Vec{e}\}$ indicates whether the exposure $\Vec{e}$ is observed.
Given exposures $\Vec{0}, \Vec{e}\in \mathcal{E}$, where $\Vec{e} = (m_1, 0, \dotsc, 0)$, \citet{aronow2017estimating} proposed estimators for the causal effect using Horvitz-Thompson inverse propensity score weighting estimators:
\begin{equation}
    \hat{\theta}_{1,m_1} = HT_{\Vec{e}} - HT_{\Vec{0}} = Y(\Vec{e}^{\,obs})\left[\frac{\mathbb{I}\{\Vec{e}^{\,obs} = \Vec{e}\}}{p(\Vec{e})} - \frac{\mathbb{I}\{\Vec{e}^{\,obs} = \Vec{0}\}}{p(\Vec{0})}\right]. 
\end{equation}  

\noindent On the other hand, the naive difference in means estimator has weights:
\begin{equation}\label{est:naive}
    w(\Vec{e}) = \frac{\mathbb{I}\{\Vec{e}^{\,obs} = (m_1, 0, \dotsc, 0)\}}{\sum_{j=1}^n\mathbb{I}\{\Vec{e}_j^{\,obs} = (m_1, 0, \dotsc, 0)\}} - \frac{\mathbb{I}\{\Vec{e}^{\,obs} = \Vec{0}\}}{\sum_{j=1}^n\mathbb{I}\{\Vec{e}_j^{\,obs} = \Vec{0}\}},
\end{equation}
for unit $i$, where $\Vec{e}_i^{\,obs}, \Vec{e}_j^{\,obs}$ are the observed exposure vectors for units $i$ and $j$ for $i,j \in \{1, \dotsc, n\}$, respectively.
Equation~\eqref{est:naive} shows that the denominator of the weight for unit $i$ depends on the exposures of other units.
Hence, the linear estimators we consider preclude naive estimators, except under certain highly symmetric designs (e.g. a Completely Randomized Design).

\subsection{Unbiased Estimators}
As a first step to limit the set of linear estimators considered, we further focus on linear estimators that are unbiased for the unit-level causal effect. 
An estimator is unbiased for the unit-level causal effect of the first exposure component if 
\begin{align}
    \E(\hat{\theta}_{1,m_1}) = \theta_{1,m_1}.
\end{align}
Under additivity, linear unbiased estimators (LUEs) exist under certain constraints, which are given by the following proposition. 
\begin{proposition} \label{constraints}
Assuming additive exposures, a linear estimator $\hat{\theta}_{1,m_1}$ is unbiased for ${\theta}_{1,m_1}$ if and only if the following constraints hold:
\begin{align*}
  &  &\sum_{\Vec{e} \in \mathcal{E}} p(\Vec{e})w(\Vec{e}) = 0 \tag{$\alpha$ constraints} \\
  &	&\sum_{\Vec{e} \in \mathcal{E}} p(\Vec{e})w(\Vec{e})\mathbb{I}\{e_{1} = m_1\} = 1 \tag{$\theta_{1,m_1}$ constraints}\\
  &	\forall m: m \in \{1, \dotsc, m_1-1\} &\sum_{\Vec{e} \in \mathcal{E}} p(\Vec{e})w(\Vec{e})\mathbb{I}\{e_{1} = m\} = 0 \tag{$\theta_{1,m}$ constraints}\\
  &\forall k,j_k: k \in \{2, \dotsc, K\}, j_k \in \{1, \dotsc, m_k\}	&\sum_{\Vec{e} \in \mathcal{E}}  p(\Vec{e})w(\Vec{e})\mathbb{I}\{e_{k} = j_k\} = 0 &\tag{$\theta_{k,j_k}$ constraints}.
\end{align*}
\end{proposition}
\noindent Denote the set of linear unbiased estimators as $\mathcal{U}$.
Given the linear constraints, the size of $\mathcal{U}$, denoted as $|\mathcal{U}|$, is $|\mathcal{U}| = \prod_{k=1}^K (m_k + 1) - \sum_{k=1}^K m_k - 1$.
Here, the product $\prod_{k=1}^K (m_k + 1)$ corresponds to the number of exposures in $\mathcal{E}$ and the summation $\sum_{k=1}^K m_k + 1$ corresponds to the number of linear constraints.
The linear constraints for unbiasedness ensure that when the estimator is averaged across exposures, it leads to a coefficient of 1 in front of the $\theta_{1,m_1}$ term, while the coefficients for the other terms are zero. 
Hence, when we compute the expected value of $\hat{\theta}_{1,m_1}$, we obtain the parameter of interest $\theta_{1,m_1}$ (see Appendix \ref{appendix:lue_constraint}).
Examples of linear unbiased estimators include Horvitz-Thompson inverse probability estimators.
Note that unbiasedness holds given the constraints in Proposition \ref{constraints} only under additivity. 
Without additivity, we will require more constraints, and hence, under additivity, we consider more estimators that would otherwise be biased.  

\begin{ex}[continues=network_interference_ex]
In the network interference example, consider the estimators 
\begin{align}\label{ex_unbiased_est}
    \hat{\theta}_{1,d_i}^{\text{two term}, 0} &= HT_{(d_i, 0)} - HT_{(0,0)} \nonumber \\
    \hat{\theta}_{1,d_i}^{\text{two term}, 1} &= HT_{(d_i, 1)} - HT_{(0,1)} \nonumber \\
    \hat{\theta}_{1,d_i}^{\text{Avg}} &= \frac{1}{2}\left(HT_{(d_i, 0)} - HT_{(0,0)}  + HT_{(d_i, 1)} - HT_{(0,1)}\right) \nonumber \\
    \hat{\theta}_{1,d_i}^{\text{four term}, 2} &= HT_{(d_i, 1)} - HT_{(2,1)}  + HT_{(2, 0)} - HT_{(0,0)} \nonumber.
\end{align}
Under additivity, all the estimators above are linear unbiased estimators for $\theta_{1,d_i}$. 
For example, $\hat{\theta}_{1,d_i}^{\text{two term}, 1}$ introduces the parameter $\theta_{2,1}$ by placing non-zero weight on the $HT_{(d_i, 1)}$ term, but $\theta_{2,1}$ is then canceled by the $HT_{(0,1)}$ term.
Furthermore, the baseline $\alpha$ is canceled, and so the parameter that remains is the parameter of interest $\theta_{1,m_1}$.
This holds for $\hat{\theta}_{1,d_i}^{\text{Avg}}$, which leverages both estimators $\hat{\theta}_{1,d_i}^{\text{two term}, 0}$ and $\hat{\theta}_{1,d_i}^{\text{two term}, 1}$, and $\hat{\theta}_{1,d_i}^{\text{four term}, 2}$, which leverages ``seemingly unrelated'' exposures such as $(2,1)$ and $(2,0)$.
However, if additivity does not hold, then $\hat{\theta}_{1,d_i}^{\text{two term}, 0}$ is the only linear unbiased estimator for $\theta_{1,d_i}$.
For example, when additivity does not hold, $\hat{\theta}_{1,d_i}^{\text{two term}, 1}$ is no longer unbiased for $\theta_{1,d_i}$, and the bias is equal to the interference plus the interaction term. 
\end{ex}

\section{Atomic Linear Unbiased Estimators} \label{sec:alue}
In the previous section, we defined a class of linear unbiased estimators when additivity holds. 
Because of the flexibility of estimators imposed by additivity, the class of LUEs can be quite large. 
However, there are particular subclasses of LUEs that are of importance. 
We first focus on a subclass of linear unbiased estimators---atomic linear unbiased estimators (ALUEs), which are simpler in terms of their supports.

\begin{definition} [Atomic Linear Unbiased Estimators]
The estimator $\hat{\theta}_{1,m_1} \in \mathcal{U}$, given by
$\hat{\theta}_{1,m_1} =  w(\Vec{e})Y(\Vec{e})$, where $\Vec{e} \in \mathcal{E}$, is atomic within $\mathcal{U}$ if for all $u \in \mathcal{U}$, if $\mathrm{supp}(u) \subset \mathrm{supp}(\hat{\theta}_{1,m_1})$, then $\mathrm{supp}(u) = \mathrm{supp}(\hat{\theta}_{1,m_1})$. 
\end{definition}
\noindent We denote the set of ALUEs by $\mathcal{A} \subset \mathcal{U}$.
The restriction of minimal support reduces the class of linear unbiased estimators considered to those whose support cannot be reduced and still be unbiased. 
Examples of ALUEs include the following two-term and four-term estimators.

\begin{ex}[continues=sutva_ex]
The treatment effect when SUTVA holds can be estimated using a two-term ALUE:
\begin{align*}
     \hat{\theta}_{1,m}^{\text{two term}} &= HT_{(m)}- HT_{(0)}.
\end{align*}
\end{ex}

\begin{ex}[Four Exposure Model]
Consider the four exposure model \citep{aronow2017estimating}, where $z_i \in \{0, 1\}$ and $d_i^{\mathbf{z}} \in \{0, \dotsc, d_i\}$ is the treated degree of unit $i$. 
The exposures are defined as $\Vec{e}=(z_i, \mathbb{I}\{d_i^{\mathbf{z}} > 0\}) \in \{0,1\}^2$.
The first exposure component gives the treatment assignment of the unit and the second exposure component indicates whether network interference is present.
We can estimate the direct treatment effect using a two-term ALUE:
\begin{align*}
     \hat{\theta}_{1,1}^{\text{two term}} &= HT_{(1,1)} - HT_{(0,1)}.
\end{align*}

\end{ex}
\noindent Note that under additivity, $HT_{(1,0)} - HT_{(0,0)}$ is also a linear unbiased estimator for the direct treatment effect.
If we do not assume additivity, $HT_{(1,1)} - HT_{(0,1)}$ is no longer unbiased. 
There are no four-term ALUEs for the direct treatment effect.
\begin{ex}[continues=network_interference_ex]
We can estimate the network interference effect using a four-term ALUE:
\begin{align*}
     \hat{\theta}_{1,d_i}^{\text{four term}} &= HT_{(d_i, 1)} - HT_{(d, 1)} + HT_{(d, 0)} -  HT_{(0,0)},
\end{align*}
where $d \in \{1, \dotsc, d_i - 1\}$.
Note there are also two-term ALUEs for the network interference effect.  
\end{ex}
\noindent In general, the number of Horvitz-Thompson terms in ALUEs can be less than more than four, but the number of terms in the ALUEs is restricted to be even.
Generally, the number of terms in ALUEs can be up to $2K$, where $K$ is the number of exposure components. 
This is because for every exposure component not of interest whose effects are added by a Horvitz-Thompson term, we need to subtract its effect with another Horvitz-Thompson term so that the estimator is unbiased for $\theta_{1,m_1}$.

\subsection{Affine Basis for Linear Unbiased Estimators}
Atomic linear unbiased estimators are the simplest LUEs in terms of its support.
However, we want to be able to generalize the properties of ALUEs to the entire class of LUEs. 
To do this, we relate the class of ALUEs to the rest of the LUEs.
We introduce a subclass of ALUEs and show that, with another class of estimators, they form an affine basis for LUEs.

In particular, we focus on a subclass of \textit{monotonic atomic linear unbiased estimators} (MALUEs):
\begin{definition}[Monotonic Atomic Linear Unbiased Estimator]
A linear unbiased estimator $\hat{\theta}_{1,m_1} \in \mathcal{A}$ is \textit{monotonic} if the exposures in its support, $\Vec{e} \in \mathrm{supp}(\hat{\theta}_{1,m_1})$, can be arranged such that there is a component-wise partial ordering.
In particular, each exposure component is simultaneously non-increasing. 
\end{definition}
\noindent Note that all two-term ALUEs are also MALUEs since, by definition, the support only contains exposures $(m_1, e_2, \dotsc, e_K)$ and $(0, e_2, \dotsc, e_K)$, where $m_1 > 0$ and all other exposure components are equal.
However, ALUEs with more than two terms are not necessarily monotonic.
\begin{ex}[continues=network_interference_ex]
Consider the following four-term ALUEs for the network interference effect:
\begin{align}
    \hat{\theta}_{1,d_i}^{\text{four term}, a} = HT_{(d_i, 1)} - HT_{(d,1)} + HT_{(d,0)} - HT_{(0,0)} \\
    \hat{\theta}_{1,d_i}^{\text{four term}, b} = HT_{(d_i, 0)} - HT_{(d,0)} + HT_{(d,1)} - HT_{(0,1)},
\end{align}
where $d \in \{1, \dotsc, d_i-1\}$.
Here, $\hat{\theta}_{1,d_i}^{\text{four term}, a}$ and $\hat{\theta}_{1,d_i}^{\text{four term}, b}$ are both ALUEs, but only $\hat{\theta}_{1,d_i}^{\text{four term}, a}$ is also a MALUE.
In $\mathrm{supp}(\hat{\theta}_{1,d_i}^{\text{four term}, b})$, consider exposures $(d_i, 0)$ and $(d,1)$, where $d_i > d$ in the first exposure component but $0 < 1$ in the second exposure component. 
We cannot arrange exposures in $\mathrm{supp}(\hat{\theta}_{1,d_i}^{\text{four term}, b})$ according to the component-wise partial order where all exposure components are non-increasing.
\end{ex}

We focus on a particular subclass of MALUEs, denoted as $\mathcal{M} \subset \mathcal{A}$, and we show that $\mathcal{M}$ is affine independent. 
\begin{Lemma}[$\mathcal{M}$ is affine independent.]\label{construction}
Consider an ordered set of exposures $\tilde{\mathcal{E}} \subseteq \mathcal{E}$ where 
\begin{align*}
    \tilde{\mathcal{E}} = &\left\{\Vec{e} \in \mathcal{E}: e_1 \in \{1, \dotsc, m_1-1\}, \exists k \in \{2, \dotsc, K\} \text{ s.t. } e_k \neq 0\right\} \\
    &\cup \{\Vec{e} \in \mathcal{E}: e_1 = m_1\}
\end{align*}
such that exposures with $e_1 \in \{1, \dotsc, m_1-1\}$ are first, followed by the exposures with $e_1 = m_1$. 
Within the subsets of exposures with $e_1 \in \{1, \dotsc, m_1-1\}$ and $e_1 = m_1$, the exposures follow a reverse reflected lexicographic order. 
Let $\mathcal{M} \subset \mathcal{A}$ contain the following estimators. 
For each exposure $\Vec{e} \in \tilde{\mathcal{E}}$, where $\Vec{e} = (e_1, \dotsc, e_K)$, consider the following:
    
    \begin{itemize}
        \item If $e_1 \in \{1, \dotsc, m_1-1\}$, add estimator
        \begin{align}
            \hat{\theta}_{1,m_1}^{\text{four term}} =&  HT_{(m_1, e_2, \dotsc,e_K)}- HT_{(e_1, e_2, \dotsc, e_K)} \\
            &+ HT_{(e_1, e'_2, \dotsc, e'_K)}- HT_{(0, e'_2, \dotsc, e'_K)}\notag
        \end{align}
        into $\mathcal{M}$.
        Here, $e'_k = 0$ for the first $k \in \{2, \dotsc, K\}$ such that $e_k \neq 0$ and $e'_{k'} = e_{k'}$ for all other $k' \in \{2, \dotsc, K\}$ where $k \neq k'$. 
        
        \item If $e_1 = m_1$, add estimator
        \begin{align}
            \hat{\theta}_{1,m_1}^{\text{two term}} =  HT_{(m_1, e_2, \dotsc,e_K)} - HT_{(0, e_2, \dotsc, e_K)}
        \end{align}
        into $\mathcal{M}$, where $e_k \in \{0, \dotsc, m_k\}$ for $k \in \{2, \dotsc, K\}$.
        
    \end{itemize}

\noindent The set $\mathcal{M}$ is affine independent.
\end{Lemma}
\noindent By construction, estimators $\hat{\theta}_{1,m_1} \in \mathcal{M}$ have support such that exposures can be ordered such that exposure components are simultaneously non-increasing, and so $\mathcal{M}$ is a subset of MALUEs.
Furthermore, note that each estimator $\hat{\theta}_{1,m_1} \in \mathcal{M}$ is uniquely identifiable by an exposure in $\tilde{\mathcal{E}}$.
Namely, the two-term estimators are uniquely identified by exposures where $e_1 = m_1$, and the four term estimators are unique identified by exposures where $e_1 \in \{1, \dotsc, m_1-1\}$.
To show that the set $\mathcal{M}$ is affine independent, we leverage the fact that the estimators are monotonic and uniquely identifiable (see Appendix \ref{appendix:proof_of_affine_ind_of_m}).
Consider estimator $\hat{\theta}$ and let $\hat{\theta} = \sum_{\tilde{\theta} \in \mathcal{M}} g(\tilde{\theta}) \tilde{\theta}$. 
We show that if $\hat{\theta} \in \mathcal{M}$, then
\begin{align}\label{def_of_g}
    g(\tilde{\theta}) = \begin{cases} 1, & \text{if $\tilde{\theta} = \hat{\theta}$} \\
    0, & \text{otherwise}
    \end{cases}.
\end{align}
Since estimators in $\mathcal{M}$ are uniquely identified by the ordered set of exposures $\tilde{\mathcal{E}}$, there is also a natural ordering of the corresponding estimators.
Using induction, we iterate through the ordered set of estimators and assign weights $g(\tilde{\theta})$ according to Equation~\eqref{def_of_g}.
At the $u$th step, if $\Vec{e}^{\,(u)} \notin \mathrm{supp}(\hat{\theta})$, then $g(\tilde{\theta}^{(u)}) = 0$.
Otherwise, since the estimators are ordered according to the estimator's uniquely identifying exposure $\Vec{e}^{\, (u)} \in \tilde{\mathcal{E}}$, and each estimator is a MALUE, the estimator $\tilde{\theta}^{(u)}$ is the last estimator in $\mathcal{M}$ with $\Vec{e}^{\,(u)}$ in its support.
Hence, if for all $u' < u$, we have $g(\tilde{\theta}^{(u')}) = 0$, $\Vec{e}^{\,(u)} \in \mathrm{supp}(\hat{\theta})$, and $\hat{\theta} \in \mathcal{M}$, then $\tilde{\theta}^{(u)} = \hat{\theta}$, i.e. $g(\tilde{\theta}^{(u)}) = 1$.
If there were at least one $u' < u$ such that $g(\tilde{\theta}^{(u')}) = 1$, then $g(\tilde{\theta}^{(u)}) = 0$ in order for unbiasedness to hold.
Since $g(\tilde{\theta}) = 1$ only if $\tilde{\theta} = \hat{\theta}$, then $\mathcal{M}$ is affine independent.

The size of the set of estimators $\mathcal{M}$, denoted as $|\mathcal{M}|$, is equal to:
\begin{align}
    |\mathcal{M}| = \underbrace{\prod_{k=2}^K (m_k + 1)}_{\text{two term estimators}} + \underbrace{(m_1 - 1)\left[\prod_{k=2}^K(m_k + 1) - 1\right]}_{\text{four term estimators}}.
\end{align}
The first term is equal to the number of two-term estimators, which are uniquely identifiable by the exposures with $e_1 = m_1$.
The second term is equal to the number of four-term estimators, where there are $m_1-1$ possible values for the first exposure component, and there are $\prod_{k=2}^K(m_k + 1)$ possible values for $e_2, \dotsc, e_K$.
We subtract the case when $e_2 = \dotsb = e_K = 0$; hence the minus one.

Although the estimators in $\mathcal{M}$ are affine independent, there are not enough estimators to span $\mathcal{U}$.
We introduce an additional set of estimators, denoted by $\mathcal{Z}$:
\begin{definition}[Zero Estimators]
Consider a set of estimators $\mathcal{Z}$, defined as the following:
\begin{align}
    \mathcal{Z} = \{\hat{\theta}_0: \hat{\theta}_0 = HT_{(0, e_2, \dotsc, e_K)} - HT_{(0, e_2, 0, \dotsc, 0)} - HT_{(0, 0, e_3, \dotsc, e_K)}  + HT_{(0, \dotsc, 0)}\},
\end{align} 
where there are at least two $k, k' \in \{2, \dotsc, K\}$ such that $e_k \neq 0, e_{k'} \neq 0$, and without loss of generality, we assumed that $e_2, e_3 \neq 0$.
\end{definition}
\noindent The size of $\mathcal{Z}$ is:
\begin{align}
    |\mathcal{Z}| = \prod_{k=2}^K(m_k + 1) - 1 - \sum_{k=2}^K m_k.
\end{align}
The first term is equal to the number of exposures where $e_1 = 0$.
Since we require that at least two $k,k'$ are such that $e_k \neq 0, e_{k'} \neq 0$, we subtract the case when $e_1 = \dotsb = e_K = 0$ and when only one of $e_2, \dotsc, e_K$ is non-zero.
Under additivity, $\E(\hat{\theta}_0) = 0$ (hence we call $\hat{\theta}_0$ a \textit{zero estimator}), which is needed to ensure the unbiased estimation of $\theta_{1,m_1}$.
We denote the union of the estimators of $\mathcal{M}$ and the zero estimators as $\hat{\Theta} = \mathcal{M} \cup \mathcal{Z}$.
\begin{theorem}[Affine basis for LUE]\label{affine_basis_thm}
The set $\hat{\Theta}$ forms an affine basis for the set of linear unbiased estimators.
\end{theorem}
\noindent 
\noindent The proof for the affine independence of $\hat{\Theta}$ is very similar to the proof of Lemma \ref{construction} (see Appendix \ref{appendix:proof_affine_basis_for_lue}).
Note now that 
\begin{align*}
    \mathrm{supp}(\hat{\Theta}) &= \{\Vec{e}: \Vec{e} \in \mathcal{E}, e_1 = m_1\} \nonumber\\
    &\cup \{\Vec{e}: \Vec{e} \in \mathcal{E}, e_1 \in \{1, \dotsc, m_1-1\}, \exists k \in \{2, \dotsc, K\} \text{ s.t. } e_k \neq 0\} \nonumber \\
    &\cup \{\Vec{e}: \Vec{e} \in \mathcal{E}, e_1 = 0, \exists k,k' \in \{2, \dotsc, K\} \text{ s.t. } e_k \neq 0, e_{k'} \neq 0\}.
\end{align*}
We order the exposures in the support such that the exposures with first exposure component equal to $m \in \{1, \dotsc, m_1-1\}$ are first, the exposures with first exposure component equal to $m_1$ are next, and the exposures with first exposure component equal to zero are last.
Within each subset of exposures, we order the exposures according to the reverse reflected lexicographic order. 
Similar to the proof of Lemma \ref{construction}, we use induction and rely on the monotonicity and uniquely identifiable estimators to show that $\hat{\Theta}$ is affine independent.
Note that each zero-estimator is uniquely identified by exposure $(0, e_2, \dotsc, e_K)$ corresponding to the first Horvitz-Thompson term in the estimator.
However, note that the zero estimators are not monotonic in the sense that MALUEs are. 
Instead, they are monotonic in the sense that the exposures follow a reverse reflected lexicographic order when we arrange them according to the order of the corresponding Horvitz-Thompson terms.
For example, for a zero estimator where $e_2 \neq 0$, the exposures corresponding to the Horvitz-Thompson terms
\begin{align*}
    HT_{(0, e_2, \dotsc, e_K)} - HT_{(0, e_2, 0, e_4, \dotsc, e_K)} - HT_{(0, 0, e_3, 0, \dotsc, 0)} + HT_{(0, \dotsc, 0)}
\end{align*}
are ordered (in increasing order) according to the reverse reflected lexicographic order.
Since exposures in $\mathrm{supp}(\hat{\Theta})$ are also ordered according to the reverse reflected lexicographic order, then for the $u$th and $u + 1$th step, we have $\Vec{e}^{\,(u)} < \Vec{e}^{\,(u+1)}$.
Hence, the zero estimator $\hat{\theta}^{(u)}_0$ is the last estimator that contains exposure $\Vec{e}^{\,(u)}$ in its support.  
We iterate through $\hat{\Theta}$ using induction and show that if an estimator $\hat{\theta} = \sum_{\tilde{\theta} \in \hat{\Theta}} g(\tilde{\theta}) \tilde{\theta}$ such that $\hat{\theta} \in \hat{\Theta}$, then the weights $g(\tilde{\theta})$ are given by Equation~\eqref{def_of_g}, i.e. $\hat{\Theta}$ is affine independent.
Since $\hat{\Theta}$ is affine independent, and the dimension of $\hat{\Theta}$ minus one (since the sum of weights is restricted to equal one for unbiasedness) is equal to the dimension of $\mathcal{U}$, then $\mathrm{span}(\hat{\Theta}) = \mathcal{U}$.
Hence, $\hat{\Theta}$ forms an affine basis for $\mathcal{U}$, and properties of the simpler estimators in $\hat{\Theta}$ extend to estimators in $\mathcal{U}$.

\section{Optimal Linear Unbiased Estimators} \label{sec:mivlue}
At this point, we have defined a set of estimators $\hat{\Theta} = \mathcal{M} \cup \mathcal{Z}$ that forms an affine basis for the set of LUEs.
Recall that additivity provides flexibility so that there are additional unbiased estimators in $\mathcal{U}$ that would otherwise be biased if additivity did not hold.
Hence, even if we just focus on estimators in $\hat{\Theta}$, the set of estimators considered could still be fairly large.
Additionally, thus far, estimators for the same estimand, such as two-term and four-term ALUEs are equivalent.
Hence, a natural question is \textit{which estimator should we use?}
In this section, we consider an additional property of variance in order to rank different linear unbiased estimators.

\subsection{Minimum Integrated Variance Linear Unbiased Estimators (MIV LUE)}\label{MIVLUE}
We consider a ``good'' estimator as one that is unbiased and has small variance.
Since LUEs depend both on the exposures and the parameters $\Theta = \{\alpha, \theta_{k,j_k}\}$ for $k \in \{1, \dotsc, K\}$ and $j_k \in \{1, \dotsc, m_k\}$ corresponding to the given exposures, we would ideally account for the parameters when we compute the variance of LUEs.
However, in general, we do not know the true set of parameters $\Theta$. 
Instead, we use distributions $\pi$ on $\Theta$ which describe the set of parameters.
We then focus on minimizing the integrated variance (IVAR), where the variance is computed with respect to distributions $\pi$ on $\Theta$, i.e. $\text{IVAR} = \int_{\Theta} \V(\hat{\theta}) \pi(d\theta)$.
Borrowing from Bayesian statistics, one can view the distributions as ``prior'' distributions on the parameters.
However, note that this is not actually Bayesian since we do not have posterior distributions---instead, we use the prior distributions to inform our choices of the weights for LUEs.
These prior distributions act as a weight, where parameters that have a higher likelihood are weighted more when computing the variance of the estimator.
Minimum integrated variance linear unbiased estimators (MIV LUEs) \citep{sussman2017elements} are then given by weights, which depend on the prior distributions, that minimize the integrated variance.
As with linear estimators, MIV LUEs depend only on the prior means and covariances \citep{hoff2009first, bickel2015mathematical, sussman2017elements}.

We seek weights $w(\Vec{e})$ that minimize the integrated variance such that the linear constraints in Proposition~\ref{constraints} hold.
To simplify the optimization problem, we assume that the parameters are uncorrelated across units, but can be correlated within units. 
We also assume that the priors have mean zero. 
However, note that if priors do not have mean zero, then the estimator
\begin{align}\label{eq:shifted_prior_est}
    \hat{\theta}_{1,m_1} = w(\Vec{e}) \left(Y(\Vec{e}) - \mu_{Y(\Vec{e})}\right) + \mu_{\theta_{1,m_1}},
\end{align}
where $\mu_{Y(\Vec{e})}$ and $\mu_{\theta_{1,m_1}}$ denote the prior mean of the potential outcome and the prior mean of $\theta_{1,m_1}$ respectively, is unbiased if $w(\Vec{e})Y(\Vec{e})$ is unbiased for $\theta_{1,m_1}$.
If $w(\Vec{e})$ minimizes the integrated variance of the estimator when priors are mean zero, then $w(\Vec{e})$ also minimizes the integrated variance of the estimator given by Equation \eqref{eq:shifted_prior_est} \citep{hoff2009first}.

Under these assumptions, the optimization problem is solved by minimizing the following Lagrangian over the weights, $w(\Vec{e})$, and lambdas:
\begin{align}\label{mivlue_problem}
    \mathcal{L} =& \frac{1}{2}\int_{\Theta}\sum_{\Vec{e} \in \mathcal{E}}p(\Vec{e})\left(w(\Vec{e})Y(\Vec{e}) - {\theta}_{1,m_1}\right)^2\pi(\theta')d\theta'\\
     &+ \lambda_1\left(1 - \sum_{\Vec{e} \in \mathcal{E}}p(\Vec{e})w(\Vec{e})\mathbbm{I}\{e_{1} = m_1\}\right) \nonumber \\
    &- \sum_{m = 1}^{m_1 - 1}\lambda_{2,m}\left(\sum_{\Vec{e} \in \mathcal{E}}p(\Vec{e})w(\Vec{e})\mathbbm{I}\{e_{1} = m\}\right) - \lambda_3\sum_{\Vec{e} \in \mathcal{E}}p(\Vec{e})w(\Vec{e})\\
    &- \sum_{k = 2}^K\sum_{j_k = 1}^{m_k} \lambda_{4, k, j_k}\sum_{\Vec{e} \in \mathcal{E}}p(\Vec{e})w(\Vec{e})\mathbbm{I}\{e_{k} = j_k\},
\end{align}
where, by taking the derivative of $\mathcal{L}$ with respect to $w(\Vec{e})$ and setting it equal to 0, the MIV LUE weights $w(\Vec{e})$ are defined as:
\begin{align}\label{mivlue_wt}
        &w(\Vec{e}) =\\
        &\frac{\lambda_1 \mathbb{I}\{e_{1} = m_1\} + \sum_{m = 1}^{m_1 - 1}\lambda_{2,m} \mathbb{I}\{e_{1} = m\} + \lambda_3 + \sum_{k = 2}^K\sum_{j_k = 1}^{m_k} \lambda_{4, k, j_k}\mathbb{I}\{e_{k} = j_k\}}{\V(Y(\Vec{e}))}. \notag
\end{align}
Note that we added the $\frac{1}{2}$ in the Lagrangian to simplify computations, but this does not change the optimization problem since it is a positive constant.

We can rewrite the optimization problem into a matrix equation.
We first define the following matrices.
Let $\mathbf{W}$ be a $|\mathcal{E}| \times |\mathcal{E}|$ diagonal matrix where the $j$th diagonal entry for $j \in \{1, \dotsc, |\mathcal{E}|\}$ is
\begin{align*}
    \mathbf{W}_{j,j} = p(\Vec{e}_j)\V(Y(\Vec{e}_j)),
\end{align*}
where $\Vec{e}_j$ is the exposure corresponding to the $j$th row/column of $\mathbf{W}$.
Let $\mathbf{C}$ be a $|\Theta| \times |\mathcal{E}|$ matrix of linear constraints given by Proposition \ref{constraints} where the rows correspond to the parameters in $\Theta$ (i.e. $k \in \{1, \dotsc, |\Theta|$) and the columns correspond to the exposures (i.e. $j \in \{1, \dotsc, |\mathcal{E}|\}$).
That is, the $k,j$th entry of matrix $\mathbf{C}$ is equal to 
\begin{align*}
    \mathbf{C}_{k,j} = p(\Vec{e}_j) \mathbb{I}\{\theta_{k} \in \Vec{e}_j\},
\end{align*}
where we write $\theta_{k} \in \Vec{e}_j$ to mean the $k$th parameter $\theta_{k}$ contributes to the value of the potential outcome, given the $j$th exposure, $Y(\Vec{e}_j)$. 
The solution vector to the optimization problem, denoted by
\begin{align*}
    \mathbf{w} = \begin{pmatrix} w(\Vec{e}_1) & \dotsc & w(\Vec{e}_{|\mathcal{E}|}) & \lambda_1 & \dotsc & \lambda_{4,K,m_K} & \lambda_3 \end{pmatrix}^T,
\end{align*} 
is then the solution to the following matrix equation:
\begin{align}\label{eq:matrix_eq_mivlue}
       \mathbf{P}^{-1}\mathbf{b} = \mathbf{w},
\end{align}
where the matrix $\mathbf{P} = \begin{pmatrix}
        \mathbf{W} & \mathbf{C}^T \\
        \mathbf{C} & \mathbf{0}
       \end{pmatrix}$ and $\mathbf{b}$ is a vector of zeros besides at the element corresponding to $\lambda_1$, at which $\mathbf{b}_{\lambda_1} = 1$.
The matrix $\mathbf{P}$ is full-rank given that the diagonal elements in $\mathbf{W}$ are positive (see Lemma \ref{lemma:p_matrix_full_rank} in Appendix \ref{appendix:characterizing_mivlue}), which holds provided the prior variance for each exposure is positive and the probability of observing each exposure is positive.
Equation~\eqref{eq:matrix_eq_mivlue} shows that the solution $\mathbf{w}$ depends on the prior variances of parameters and the probability of exposures. 
Hence, not all LUEs are also MIV LUEs---whether LUEs are also MIV LUEs depends on the design probabilities and support of the estimators.
We characterize the set of MIV LUEs in the next section through the support of the estimator.

\subsection{Characterization of MIV LUEs}
Before now, we have characterized LUEs through the linear constraints as given in Proposition \ref{constraints}.
However, we can also classify LUEs through their support, denoted by $\mathcal{E'} \subseteq \mathcal{E}$.
The support $\mathcal{E'}$ of an LUE contains exposures such that there exist weights of exposures where, when multiplied with the vector of indicators for exposures, it solves
\begin{align}\label{eq:lue_exposure_condition}
    \mathbf{C} \Vec{u} = \begin{pmatrix} 0 & \dotsc & 0 & 1 & 0 & \dotsc & 0 \end{pmatrix}^T,
\end{align}
where the $j$th element of $\Vec{u} \in \R^{|\mathcal{E}|}$ is $\Vec{u}_{j} = \mathbb{I}\{\Vec{e}_j \in \mathcal{E'}\} w(\Vec{e}_j)$, and the 1 on the right hand side corresponds to $\theta_{1,m_1}$.
Effectively, solving for $\Vec{u}$ such that it satisfies the equation ensures that the linear unbiased constraints are satisfied.
\begin{ex}\label{lue_support_ex}
Consider the network interference example where $\Vec{e} = (d_i^{\mathbf{z}}, z_i)$ for $d_i^{\mathbf{z}} \in \{0, \dotsc, d_i\}$, where $d_i$ is the degree of unit $i$, and $z_i \in \{0,1\}$.
Examples of supports of LUEs include: 
\begin{align*}
    \mathcal{E}^{\text{two term}, z_i} &= \{(d_i, z_i), (0,z_i)\} \\
    \mathcal{E}^{\text{four term}, d} &= \{(d_i, 1), (d, 1), (d,0), (0,0)\} \\
    \mathcal{E}^{\text{six term}, d} &= \{(d_i, 1), (d_i,0), (d,1), (d,0), (0,1), (0,0)\},
\end{align*}
where $d \in \{1, \dotsc, d_i-1\}$.
These sets of exposures satisfy Equation~\eqref{eq:lue_exposure_condition}. 
For example, the weight vectors $(1, -1)$, $(1, -1, 1, -1)$, and $\left(\frac{3}{2}, -\frac{1}{2}, -1, 1, -\frac{1}{2}, -\frac{1}{2}\right)$ lead to LUEs with support $\mathcal{E}^{\text{two term}, z_i}$,  $\mathcal{E}^{\text{four term}, d}$, and $\mathcal{E}^{\text{six term}, d}$, respectively.
\end{ex}

Given a subset of exposures $\mathcal{E'} \subseteq \mathcal{E}$ that is a valid support for LUEs, i.e. it satisfies Equation~\eqref{eq:lue_exposure_condition}, we can divide the set of parameters $\Theta$ into the following subsets. 
Let $\Theta^N \subseteq \Theta$ denote the set of parameters where $\theta^N \in \Theta^N$ are such that $\theta^N \notin \Vec{e}{\,'}$ for all $\Vec{e}{\,'} \in \mathcal{E}'$.
We further divide the parameters in $\Theta^F = \Theta \setminus \Theta^{N}$ as $\Theta^F = \Theta^R \cup \Theta^{NR}$.
Specifically, $\Theta^{NR}$ will be a maximal subset of $\Theta^F$ such that the submatrix of $\bvar{C}$, with rows given by $\Theta^{NR}$ and columns given by $\mathcal{E}'$, has linearly independent rows.
Additionally, we can subdivide matrices $\mathbf{W}$ and $\mathbf{C}$.
Matrix $\mathbf{W}$ is a block diagonal matrix with matrices $\mathbf{N}$ and $\mathbf{F}$ on the diagonal.
Matrix $\mathbf{N}$ is a diagonal matrix corresponding to exposures $\Vec{e} \in \mathcal{E} \setminus \mathcal{E'}$ and $\mathbf{F}$ is a diagonal matrix with rows corresponding to exposures $\Vec{e}{\,'} \in \mathcal{E'}$. 
We denote the constraint submatrices of $\mathbf{C}$ as $\bvar{C}_{e}^{p}$, where the subscript corresponds to the set of exposures $e$ and the superscript corresponds to the set of parameters $p$.
For each $e \in \{N, F\}$ and $p \in \{N, NR, R\}$, we define $\bvar{C}_{e}^{p}$ to contain rows corresponding to constraints of parameters in $\Theta^p$ and columns correspond to the exposures in $\mathcal{E}^e$.
Here, $\mathcal{E}^F = \mathcal{E'}$ and $\mathcal{E}^N = \mathcal{E} \setminus \mathcal{E'}$.

Given the subsets of exposures and parameters defined above, we can then characterize MIV LUEs through their support:
\begin{theorem}\label{mivlue_thm}
Let $\mathcal{E'} \subseteq \mathcal{E}$ such that $\mathrm{span}\left(\{\Vec{v}_{\Vec{e}{\,'}}\}_{\Vec{e}{\,'} \in \mathcal{E'}}\right) \cap \{\Vec{v}_{\Vec{e}}\}_{\Vec{e} \in \mathcal{E}} = \{\Vec{v}_{\Vec{e}{\,'}}\}_{\Vec{e}{\,'} \in \mathcal{E'}}$ where for $\vec{e}\in \mathcal{E}$,  $\Vec{v}_{\Vec{e}} \in \{0, 1\}^{|\Theta|}$ such that $\Vec{v}_{\Vec{e}}^T\Vec{v} = Y(\Vec{e})$ where $\Vec{v}$ is the vector of parameters $\Theta$.
Furthermore, assume that $\mathcal{E'}$ satisfies Equation~\eqref{eq:lue_exposure_condition}, i.e. there exists an unbiased estimator $\hat{\theta}$ where $\mathrm{supp}(\hat{\theta}) = \mathcal{E}'$.
If the design $p$ is such that $p(\Vec{e}) > 0$ for all $\Vec{e} \in \mathcal{E}$,
then there exists a $\hat{\theta}$ with $\mathrm{supp}(\hat{\theta}) \subseteq \mathcal{E'}$ and $\hat{\theta}$ is a limit of MIV LUEs. 
Furthermore, if for every exposure $\Vec{e}{\,'} \in \mathcal{E'}$, we have $\lim_{\eta \to \infty} \sum_{k = 1}^{|\Theta^{NR}|} Adj\left({\mathbf{C}_F^{NR}} \mathbf{F}_{\eta}^{-1} {\mathbf{C}_F^{NR}}^T\right)_{k,1} \mathbb{I}\{\theta_{k} \in \Vec{v}_{\Vec{e}{\,'}}^T\Vec{v}\} \neq 0$, where $Adj$ is the adjugate, then $\mathrm{supp}(\hat{\theta}) = \mathcal{E'}$.
\end{theorem}

\noindent Theorem \ref{mivlue_thm} states that we can find a limit of MIV LUEs $\hat{\theta}$ whose support is a subset of $\mathcal{E'} \subseteq \mathcal{E}$ as long as $\mathcal{E}'$ is a valid support for LUEs and $\mathcal{E}'$ is such that the corresponding set of vectors of indicators for exposures in $\mathcal{E}'$, denoted $\{\Vec{v}_{\Vec{e}{\,'}}\}_{\Vec{e}{\,'} \in \mathcal{E'}}$, contains all vectors in $\mathrm{span}\left(\{\Vec{v}_{\Vec{e}{\,'}}\}_{\Vec{e}{\,'} \in \mathcal{E'}}\right)$ that correspond to valid exposures of interest. %
Since $\mathrm{span}\left(\{\Vec{v}_{\Vec{e}{\,'}}\}_{\Vec{e}{\,'} \in \mathcal{E'}}\right)$ is a linear subspace of $\R^{|\Theta|}$, there exists a positive semi-definite matrix $\mathbf{\Sigma}$ such that $\mathrm{span}\left(\{\Vec{v}_{\Vec{e}{\,'}}\}_{\Vec{e}{\,'} \in \mathcal{E'}}\right) = \mathrm{Null}(\mathbf{\Sigma})$.
For example, $\mathbf{\Sigma} = I - XX^T$, where the columns of $X$ are vectors that form an orthonormal basis for $\mathrm{span}\left(\{\Vec{v}_{\Vec{e}{\,'}}\}_{\Vec{e}{\,'} \in \mathcal{E'}}\right)$. %
Given prior variance-covariance matrix $\mathbf{\Sigma}$, we then solve for $\mathbf{w}$ given by Equation~\eqref{eq:matrix_eq_mivlue}.
However, we require the following lemma:

\begin{Lemma}\label{var_cov_matrix_lemma}
Let $\Theta = \{\alpha, \theta_{1,1}, \dotsc, \theta_{K,m_K}\}$ be the set of parameters, and let $\mathbf{\Sigma} \in \R^{|{\Theta}| \times |{\Theta}|}$ be a variance-covariance matrix for the parameters. 
Let $\Vec{v}_1, \Vec{v}_2 \in \{0,1\}^{|{\Theta}|}$ be vectors such that $\Vec{v}_1^{T} \mathbf{\Sigma} \Vec{v}_1 = 0$ and $\Vec{v}_2^{T} \mathbf{\Sigma} \Vec{v}_2 = a$ where $0<a<\infty$.
There exists a sequence of positive semi-definite matrix $\tilde{\mathbf{\Sigma}}_{\eta} \in \R^{|{\Theta}| \times |{\Theta}|}$ such that $\lim_{\eta \to \infty} \Vec{v}_1^T \tilde{\mathbf{\Sigma}}_{\eta} \Vec{v}_1 < \infty$ and $\lim_{\eta \to \infty} \Vec{v}_2^T \tilde{\mathbf{\Sigma}}_{\eta} \Vec{v}_2 = \infty$.
\end{Lemma}
\noindent Specifically, let $\tilde{\mathbf{\Sigma}} = \eta \mathbf{\Sigma} + B$ for $\eta \in \R$ and $B \in \R^{|\Theta| \times |\Theta|}$ be a positive semi-definite matrix where elements $0< b_{k,j} < \infty$ are small, where $k,j \in \{1, \dotsc, |\Theta|\}$.
From Theorem \ref{mivlue_thm}, since there exists a positive-definite matrix $\mathbf{\Sigma}$ such that $\mathrm{span}\left(\{\Vec{v}_{\Vec{e}{\,'}}\}_{\Vec{e}{\,'} \in \mathcal{E'}}\right) = \mathrm{Null}(\mathbf{\Sigma})$, the variances of the potential outcomes corresponding to exposures $\Vec{e}{\,'} \in \mathcal{E'}$ are zero. 
Lemma \ref{var_cov_matrix_lemma} then says there exists a sequence of variance-covariance matrices $\tilde{\mathbf{\Sigma}}_{\eta}$ such that the potential outcomes corresponding to exposures $\Vec{e}{\,'} \in \mathcal{E}'$ have finite limiting variances.
On the other hand, potential outcomes given by $\Vec{e} \not \in \mathcal{E'}$ have infinite limiting variances under $\tilde{\mathbf{\Sigma}}_{\eta}$ since $\Vec{v}_{\Vec{e}} \notin \mathrm{span}(\{\Vec{v}_{\Vec{e}{\,'}}\}_{\Vec{e}{\,'} \in \mathcal{E'}}) = \mathrm{Null}(\mathbf{\Sigma})$.
Denote $\mathbf{P}_{\eta}$ as the matrix $\mathbf{P}$, where submatrix $\mathbf{W}_{\eta}$ depends on variances given by $\tilde{\mathbf{\Sigma}}_{\eta}$.  
We also denote submatrices of $\mathbf{W}_{\eta}$ with the subscript $\eta$.
Together with Theorem \ref{mivlue_thm}, we then see that potential outcomes with finite limiting variances potentially have non-zero weights, while potential outcomes with infinite limiting variances have weights of zero.
Note that this is supported by Equation~\eqref{mivlue_wt}, where the variance of the potential outcome is inversely related to the MIV LUE weights.
We can interpret this as we put more weight on exposures that we are more confident about, i.e. potential outcomes with smaller prior variances, while we put less weight on exposures that we are not as informed about, i.e. potential outcomes with larger prior variances.

To ensure that weights of the potential outcomes corresponding to exposures in $\mathcal{E'}$ are non-zero, we further require that, for every exposure in $\mathcal{E'}$, the limit of the sum of the entries of the adjugate of ${\mathbf{C}_F^{NR}} \mathbf{F}_{\eta}^{-1} {\mathbf{C}_F^{NR}}^T$ in the column corresponding to parameter $\theta_{1,m_1}$ as $\eta \to \infty$ is non-zero.
Although it is possible for the weights of exposures in $\mathcal{E'}$ to be zero, we show, through an example of a six-term exposure set (see Appendix \ref{appendix:six_term_ex}), that ``typical'' choices of design $p$ will lead to non-zero weights.
Hence, under most designs $p$, we have $\mathrm{supp}(\hat{\theta}) = \mathcal{E'}$.
In general, if limiting prior variances of all parameters are finite, $\hat{\theta}$ is a MIV LUE with non-zero weight on all exposures, and $\hat{\theta}$ is an affine combination of estimators in $\hat{\Theta}$.
Note that formally, $\hat{\theta}$ is a solution to the matrix equation in Equation~\eqref{eq:matrix_eq_mivlue} while taking the limit $\mathbf{P}_{\eta}$ as $\eta \to \infty$.
Since the matrix $\mathbf{P}_{\eta}$ may contain infinite values in the limit, it is not a well-defined problem.
However, for convenience, we say that a limit of MIV LUEs is also MIV LUE.
Hence $\hat{\theta}$ is a MIV LUE.

\begin{ex}[continues=lue_support_ex]
We considered three examples of supports for LUEs in the context of network interference:
\begin{align*}
    \mathcal{E}^{\text{two term}, z_i} &= \{(d_i, z_i), (0,z_i)\} \\
    \mathcal{E}^{\text{four term}, d} &= \{(d_i, 1), (d, 1), (d,0), (0,0)\} \\
    \mathcal{E}^{\text{six term}, d} &= \{(d_i, 1), (d_i,0), (d,1), (d,0), (0,1), (0,0)\}.
\end{align*}

\sloppy
Consider $\mathcal{E}^{\text{two term}, z_i} = \{(0,z_i), (d_i, z_i)\}$.
Note that $\mathrm{span}\left(\{\Vec{v}_{0,z_i}, \Vec{v}_{d_i,z_i}\}\right) \cap \{\Vec{v}_{\Vec{e}}\}_{\Vec{e} \in \mathcal{E}} = \{\Vec{v}_{0,z_i}, \Vec{v}_{d_i,z_i}\}$ for $z_i \in \{0, 1\}$.
By Theorem \ref{mivlue_thm}, there exists weights $w(\Vec{e})$ under a given prior such that $\hat{\theta}^{\text{two term}, z_i} = \sum_{\Vec{e} \in \mathcal{E}^{\text{two term}, z_i}} w(\Vec{e}) Y(\Vec{e})$ is a MIV LUE with support $\mathcal{E}^{\text{two term}, z_i}$.
Specifically, examples of priors include the following, depending on whether $z_i = 0$ or $z_i = 1$.
First consider $z_i = 0$, i.e. $\mathcal{E}^{\text{two term}, 0} = \{(0,0), (d_i, 0)\}$. 
Let $\mathbf{\Sigma}^{\text{two term}, 0}$ be defined such that parameters $\V(\alpha) = \V(\theta_{1,d_i}) = 0$ and variances of all other parameters are positive.
Now consider $z_i = 1$, i.e. $\mathcal{E}^{\text{two term}, 1} = \{(0,1), (d_i, 1)\}$.
Let $\mathbf{\Sigma}^{\text{two term}, 1}$ be such that $\V(\theta_{1,d_i}) = 0$, $\V(\alpha) = \V(\theta_{2,1}) > 0$, $cov(\alpha, \theta_{2,1}) = -\V(\alpha)$ so that $\alpha = - \theta_{2,1}$, and the variances of all other parameters are positive while covariances are non-negative.
The prior variance matrices $\mathbf{\Sigma}^{\text{two term}, 0}$ and $\mathbf{\Sigma}^{\text{two term}, 1}$ inform the MIV LUE weights. 
In particular, the MIV LUE weights given by priors $\mathbf{\Sigma}^{\text{two term}, 0}$ and $\mathbf{\Sigma}^{\text{two term}, 1}$ are equal to the weights of the two-term ALUEs $\hat{\theta}_{1,d_i}^{\text{two term}, 0}$ and $\hat{\theta}_{1,d_i}^{\text{two term}, 1}$, respectively.
That is, the two-term ALUEs are also MIV LUEs for some prior.

We now consider $\mathcal{E}^{\text{four term}, d} = \{(0,0), (d,0), (d,1), (d_i,1)\}$, where $d \in \{1, \dotsc, d_i-1\}$. 
We consider the span of $\{\Vec{v}_{\Vec{e}{\,'}}\}_{\Vec{e}{\,'} \in \mathcal{E}^{\text{four term},d}}$.
In particular, the vector $\Vec{v}_{d_i,0} \in \mathrm{span}\left(\{\Vec{v}_{\Vec{e}{\,'}}\}_{\Vec{e}{\,'} \in \mathcal{E}^{\text{four term},d}}\right) \cap \{\Vec{v}_{\Vec{e}}\}_{\Vec{e} \in \mathcal{E}}$, where $\Vec{v}_{d_i,0} = \Vec{v}_{d_i,1} - \Vec{v}_{d,1} + \Vec{v}_{d,0}$.
However, $\Vec{v}_{d_i,0} \notin \{\Vec{v}_{\Vec{e}{\,'}}\}_{\Vec{e}{\,'} \in \mathcal{E}^{\text{four term},d}}$.
Then by Theorem \ref{mivlue_thm}, there do not exist MIV LUEs for any prior under our formulation with support $\mathcal{E}^{\text{four term},d}$, i.e. $\hat{\theta}_{1,d_i}^{\text{four term},d}$ is not a MIV LUE. 

The set of exposures $\mathcal{E}^{\text{six term}, d}$ is a support for a six-term MIV LUE. 
We focus on a generalized example of a six-term exposure set in the next section.
\end{ex}

\subsection{Example: Six-Term Exposure Set}
For notational simplicity, we focus on exposures with two exposure components, but the results generalize to cases with more than two exposure components where all other exposure components are the same for all six exposures.
Let $\mathcal{E}^{\text{six term},m} = \{(0,0), (0,j), (m,0), (m,j), (m_1,0), (m_1,j)\}$, where $j \in \{1, \dotsc, m_2\}$ and $m \in \{1, \dotsc, m_1-1\}$.
By Theorem \ref{mivlue_thm}, since $\{\Vec{v}_{\Vec{e}{\,'} \in \mathcal{E}^{\text{six term},m}}\} = \mathrm{span}\left(\{\Vec{v}_{\Vec{e}{\,'} \in \mathcal{E}^{\text{six term},m}}\}\right) \cap \{\Vec{v}_{\Vec{e} \in \mathcal{E}}\}$, there exists a MIV LUE $\hat{\theta}_{1,m_1}^{\text{six term},m}$ such that $\mathrm{supp}(\hat{\theta}_{1,m_1}^{\text{six term},m}) = \mathcal{E}^{\text{six term},m}$ for a given prior.
By Theorem \ref{affine_basis_thm}, since $\hat{\theta}_{1,m_1}^{\text{six term}}$ is an LUE, we can write
\begin{align}\label{six_term_ht_est}
    \hat{\theta}_{1,m_1}^{\text{six term}} &= \alpha_1 \left(HT_{(m_1,0)} - HT_{(0,0)}\right) + \alpha_2 \left(HT_{(m_1,j)} - HT_{(0,j)}\right) \nonumber \\
    &+ \alpha_3 \left(HT_{(m_1,j)} - HT_{(m,j)} + HT_{(m,0)} + HT_{(0,0)}\right),
\end{align}
where the three estimators are ALUEs and $\alpha_1 + \alpha_2 + \alpha_3 = 1$.
Furthermore, by a similar argument as in the proof of Theorem \ref{affine_basis_thm}, the set of ALUEs 
\begin{align*}
    \{\hat{\theta}_{1,m_1}^{\text{two term}, 0}, \hat{\theta}_{1,m_1}^{\text{two term}, j}, \hat{\theta}_{1,m_1}^{\text{four term},m}\} &=  \{HT_{(m_1,0)} - HT_{(0,0)}, HT_{(m_1,j)} - HT_{(0,j)},\\
    &HT_{(m_1,j)} - HT_{(m,j)} + HT_{(m,0)} + HT_{(0,0)}\}
\end{align*} 
forms a basis for estimators with exposure set $\mathcal{E}^{\text{six term},m}$.
Hence, we only need to focus on the three weights $\alpha_1, \alpha_2, \alpha_3$ as opposed to the six weights on the different exposures.

Recall that in the previous section, we showed that the two two-term estimators, $\hat{\theta}_{1,m_1}^{\text{two term}, 0}, \hat{\theta}_{1,m_1}^{\text{two term}, j}$, are also MIV LUEs for some prior, and so it is possible that $\alpha_1 = 1, \alpha_2 = 0, \alpha_3 = 0$ or $\alpha_1 = 0, \alpha_2 = 1, \alpha_3 = 0$.
However, since the four-term estimator $\hat{\theta}_{1,m_1}^{\text{four term}, m}$ is not a MIV LUE for any prior, then $\alpha_3 \neq 1$.
Although four-term ALUEs are not MIV LUEs, exposures in the supports of four-term ALUEs may still contribute to MIV LUEs.
Through the weights $\alpha_1, \alpha_2, \alpha_3$, we investigate how much emphasis might be put on exposures that are ``seemingly unrelated'' to the estimand of interest, such as exposures $(m,j)$ and $(m,0)$. 

Solving for the MIV LUE weights given by the MIV LUE problem in Equation~\eqref{eq:matrix_eq_mivlue} given exposure set $\mathcal{E}^{\text{six term}, m}$ and some prior $\mathbf{\Sigma}^{\text{six term}}$ (see Appendix \ref{appendix:six_term_exposure}), we determine that
\begin{align}\label{alpha_3_wt}
    \alpha_3 &= \frac{r(m,0)r(m,j)\bigg\{r(m_1,j) r(0,0) - r(m_1,0) r(0,j) \bigg\}}{D}
\end{align}
where 
$r(\Vec{e}) = \frac{p(\Vec{e})}{\V(Y(\Vec{e})}$ and 
\begin{align}
    D &=  r(m_1,0) \bigg[r(0,0) r(m,0) r(m_1,j) + r(0,0) r(m,j) r(m_1,j) \bigg] \nonumber \\
    &+ r(m_1,j) \bigg[r(m_1,0) r(m,0) r(0,j) + r(m_1,0) r(m,j) r(0,j) \bigg] \nonumber \\
    &+\bigg[r(m_1,0) + r(m_1,j)\bigg] 
    \bigg[r(0,0) r(m,0) r(0,j) + r(0,0) r(m,0) r(m,j) \nonumber \\
    &+r(0,0) r(m,j) r(0,j) + r(0,j) r(m,0) r(m,j) \bigg].
\end{align}
Hence, the weight $\alpha_3$ is determined by the prior variance-covariance matrix $\mathbf{\Sigma}^{\text{six term}}$ and design probabilities $p(\Vec{e}) \in (0,1)$ for $\Vec{e} \in \mathcal{E}$.
Since we assume that the design is fixed, we focus on how $\alpha_3$ changes as we vary the different prior variances. 

We first assume that the parameters are independent, i.e. covariances are zero. 
Rearrange Equation~\eqref{alpha_3_wt} such that $\V(\theta_{1,m})$ appears only in the denominator of $\alpha_3$.
Hence $\V(\theta_{1,m})$ is inversely related to $\alpha_3$, and the weight $\alpha_3$ is maximized as $\V(\theta_{1,m}) \to 0$.
This makes sense since exposures with $e_1 = m$ contributes the most in estimating $\theta_{1,m_1}$ when we are certain about $\theta_{1,m}$, and $\hat{\theta}_{1,m_1}^{\text{four term},m}$ is the only estimator in Equation~\eqref{six_term_ht_est} whose support contains exposures with $e_1 = m$.
If we are not as certain about $\theta_{1,m}$ relative to the other parameters, we put more weight on the two-term estimators.

\begin{figure}[tb]
    \center
  \begin{minipage}[t]{\linewidth}\centering
    \includegraphics[width=10cm]{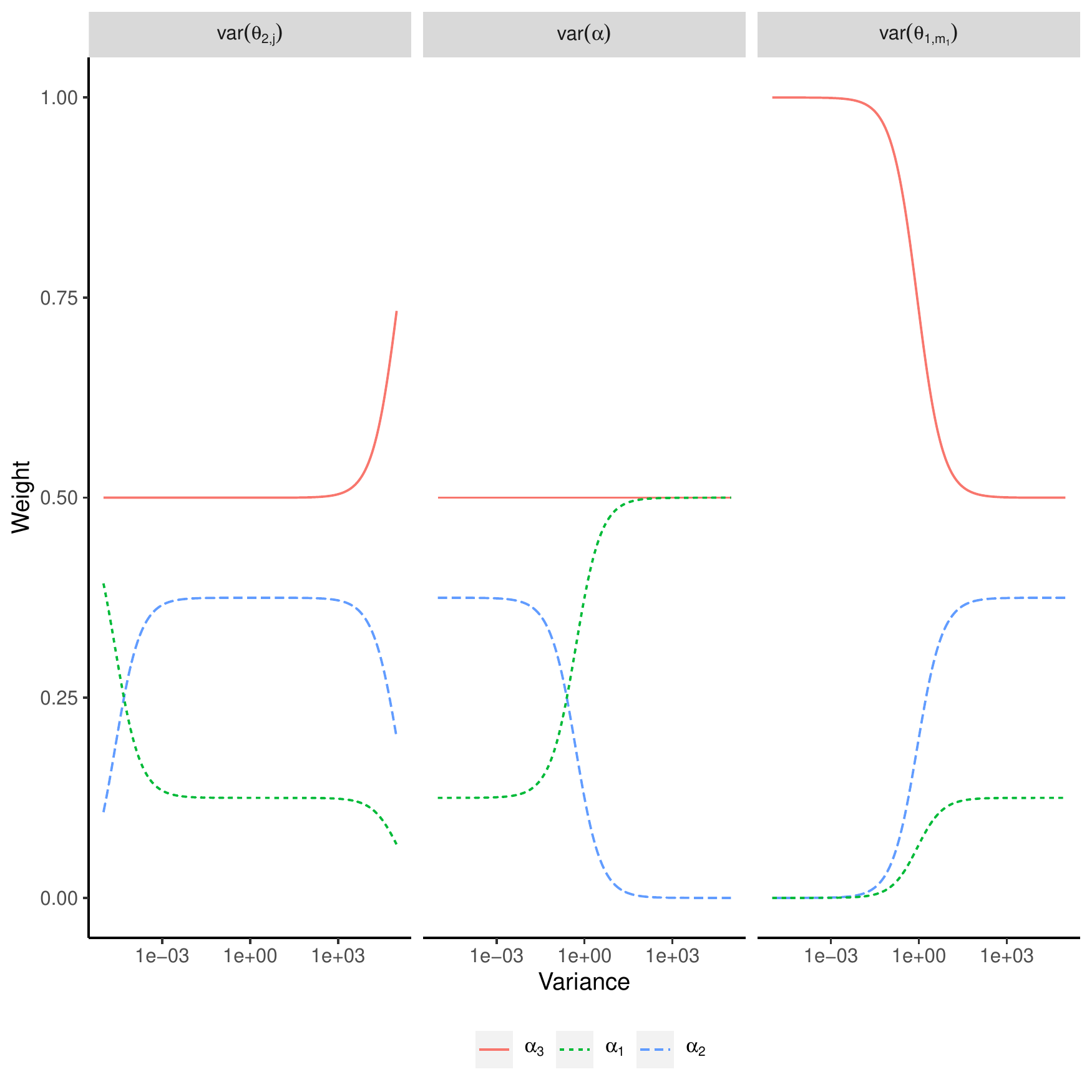}
  \end{minipage}\hfill
    \caption{The trends of weight $\alpha_1$, $\alpha_2$, and $\alpha_3$ (indicated by line type and color) for $\V(\theta_{1,m}) = 0.00001$ and different values $\V(\alpha)$, $\V(\theta_{1,m_1})$, and $\V(\theta_{2,j})$ as indicated by the panels. Variances that are not varying are set to values to maximize $\alpha_3$: $\V(\theta_{1,m}) = 0.00001$, $\V(\alpha) = 0.00001$, $\V(\theta_{1,m_1}) = 100,000$, and $\V(\theta_{2,j}) = 1$.}
    \label{fig:alpha_3_pattern}
\end{figure}
Figure \ref{fig:alpha_3_pattern} shows the trajectories of $\alpha_1$, $\alpha_2$, and $\alpha_3$ as $\V(\alpha)$, $\V(\theta_{1,m_1})$, and $\V(\theta_{2,j})$ vary when the probability of a unit being treated follows the Bernoulli distribution with probability $0.5$, $m_1 = 3$, and $m_2 = 1$. 
In each of the panels, the variances of parameters that are not varying are fixed to values aimed to maximize $\alpha_3$ (see Appendix \ref{appendix:six_term_exposure} for details).
That is, we set $\V(\theta_{1,m}) = 0.00001$, $\V(\alpha) = 0.00001$, $\V(\theta_{1,m_1}) = 100,000$, and $\V(\theta_{2,j}) = 1$.
In general, the weights depend on the fraction $\frac{\V(\theta_{2,j})}{\V(\theta_{1,m_1})}$.
As $\frac{\V(\theta_{2,j})}{\V(\theta_{1,m_1})}$ increases, the weight $\alpha_1$ is generally non-decreasing while $\alpha_2$ is generally non-increasing. 
This is because $\hat{\theta}_{1,m_1}^{\text{two term}, 0}$, which corresponds to $\alpha_1$, does not contain the parameter $\theta_{2,j}$, but $\hat{\theta}_{1,m_1}^{\text{two term}, j}$, which corresponds to $\alpha_2$, contains the parameter $\theta_{2,j}$.
When we are less certain about $\theta_{2,j}$ relative to $\theta_{1,m_1}$, i.e. when the variance of $\theta_{2,j}$ is relatively larger than the variance of $\theta_{1,m_1}$, the exposures of $\hat{\theta}_{1,m_1}^{\text{two term}, j}$ contributes less to the estimation of $\theta_{1,m_1}$.
When we are more certain about $\theta_{2,j}$ relative to $\theta_{1,m_1}$, i.e. when the variance of $\theta_{2,j}$ is relatively smaller than the variance of $\theta_{1,m_1}$, the exposures of $\hat{\theta}_{1,m_1}^{\text{two term}, 0}$ contributes less to the estimation of $\theta_{1,m_1}$.
Recall that the support of the four-term estimator contains exposures that are found in both of the supports of the two-term estimators.
Hence, the weight $\alpha_3$ is not necessarily monotonic as $\frac{\V(\theta_{2,j})}{\V(\theta_{1,m_1})}$ changes. 
As the $\frac{\V(\theta_{2,j})}{\V(\theta_{1,m_1})}$ approaches 0.0002, the weight $\alpha_3$ increases in general, but as $\frac{\V(\theta_{2,j})}{\V(\theta_{1,m_1})}$ falls outside of $(2 \times 10^{-6}, 2 \times 10^{-4})$, then $\alpha_3$ decreases. 
Hence, when the probability for a unit to be treated follows a Bernoulli distribution with probability $0.5$, $m_1 = 3$, and $m_2 = 1$, $\alpha_3$ is maximized if $\frac{\V(\theta_{2,j})}{\V(\theta_{1,m_1})}$ is inside the range $(2 \times 10^{-6}, 2 \times 10^{-4})$. 

The weights $\alpha_2$ and $\alpha_3$ also depend on $\V(\alpha)$, specifically the ratio $\frac{\V(\theta_{2,j})}{\V(\alpha)}$.
As $\frac{\V(\theta_{2,j})}{\V(\alpha)}$ increases, the weight $\alpha_3$ is non-decreasing, while the weight $\alpha_2$ is non-increasing.
This is possibly explained because $\hat{\theta}_{1,m_1}^{\text{four term},m}$, which corresponds to $\alpha_3$, also has exposures with $e_2 = 0$ in its support, which only depends on parameters $\alpha$ and either $\theta_{1,m}$ or $\theta_{1,m_1}$.
On the other hand, $\hat{\theta}_{1,m_1}^{\text{two-term}, j}$ only has exposures with $e_2 = j$ in its support.
When the variance of $\alpha$ is relatively higher than the variance of $\theta_{2,j}$, we prioritize the exposures in the support of $\hat{\theta}_{1,m_1}^{\text{two-term}, j}$.
When the variance of $\alpha$ is relatively lower than the variance of $\theta_{2,j}$, we prioritize the exposures in the support of the four-term estimator compared to exposures in the support of $\hat{\theta}_{1,m_1}^{\text{two-term}, j}$.
When the other variances are fixed to values to maximize $\alpha_3$, we see that the weight $\alpha_1$ does not depend on $\frac{\V(\theta_{2,j})}{\V(\alpha)}$.

The weight $\alpha_3$ approaches zero when we take the limit of the variances of $\alpha$ and $\theta_{1,m_1}$, specifically $\V(\alpha) \to \infty$ and $\V(\theta_{1,m_1}) \to 0$.
However, $\alpha_3$ remains non-zero even in the limits of $\V(\theta_{2,j})$, both towards zero and towards infinity.
Hence, even if one is highly uncertain about $\theta_{2,j}$, exposures in the support of the four-term ALUE may still contribute to the estimation of $\theta_{1,m_1}$.
Taking the limits of the variances of the parameters when the parameters are uncorrelated, the weight $\alpha_3$ is maximized at the following:
\begin{Cor}\label{cor:max_of_a3}
Consider a set of exposures \[ \mathcal{E}^{\text{six term},m} = \{(0,0), (0,j), (m,0), (m,j), (m_1,0), (m_1,j)\},\] where $j \in \{1, \dotsc, m_2\}$ and $m \in \{1, \dotsc, m_1-1\}$.
Under the assumption that all covariances are zero, the weight $\alpha_3$ is maximized when $\V(\alpha), \V(\theta_{1,m}) \to 0$, $\V(\theta_{1,m_1}) \to \infty$, and $\V(\theta_{2,j}) < \infty$.
Given a design $p(\Vec{e})$ for $\Vec{e} \in \mathcal{E}$, 
the maximum of $\alpha_3$ is constant and depends only on the design:
\begin{align}\label{max_a3}
   max_{\substack{\V(\alpha), \V(\theta_{1,m}) \\ \V(\theta_{1,m_1}), \V(\theta_{2,j})}}\{\alpha_3\} = \frac{p(m,j)p(m_1,j)}{\left[p(0,j) + p(m,j)\right]\left[p(m_1,0) + p(m_1,j)\right]}.
\end{align}
\end{Cor}
\noindent The maximum contribution of exposures with $e_1 = m$ depends on the design. 
The choice of the design is out of the scope for this paper, and future work may be done on this topic.
Under the conditions when $\alpha_3$ is maximized, we require $\V(\theta_{1,m_1}) \to \infty$, $\V(\theta_{1,m}) \to 0$, and $\V(\alpha) \to 0$.
Thus, estimators with an $\alpha_3$ weight equal to Equation~\eqref{max_a3} formally lie on the boundary of the set of MIV LUEs.
However, since we considered the set of MIV LUEs to be closed for convenience, estimators with a maximum $\alpha_3$ are MIV LUEs. 

\begin{figure}[tb]
    \center
  \begin{minipage}[t]{\linewidth}\centering
    \includegraphics[width=10cm]{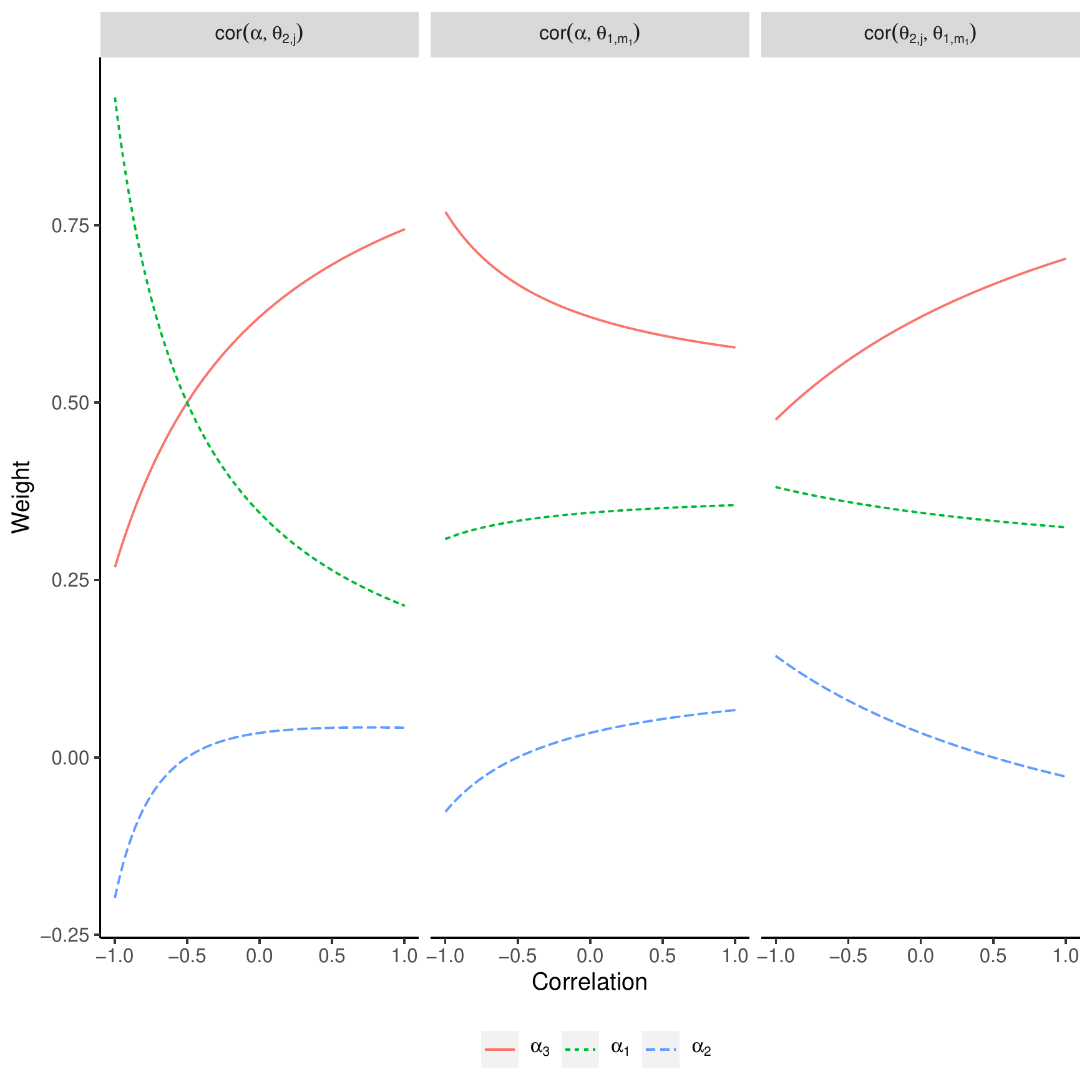}
  \end{minipage}\hfill
    \caption{The trends of weight $\alpha_1, \alpha_2, \alpha_3$ (indicated by line type and color) for $\V(\theta_{1,m}) = 0.00001$ and $\V(\theta_{1,m_1}) = \V(\alpha) = \V(\theta_{2,j}) = 1$ and different values of
    $\mathrm{cor}(\alpha, \theta_{2,j}), \mathrm{cor}(\alpha, \theta_{1,m_1}), \mathrm{cor}(\theta_{1,m_1}, \theta_{2,j})$ as indicated on the x-axis.}
    \label{fig:alpha_3_pattern_cor_alpha_v2j}
\end{figure}

\sloppy
When covariances between parameters are non-zero, similar deductions can be made---the weights $\alpha_1, \alpha_2, \alpha_3$ depend on the overall variances of the potential outcomes of the corresponding estimators.
Figure \ref{fig:alpha_3_pattern_cor_alpha_v2j} shows the trends of the weights $\alpha_1$, $\alpha_2$, and $\alpha_3$ as the correlation between pairs of parameters: $\mathrm{cor}(\alpha, \theta_{2,j})$, $\mathrm{cor}(\alpha, \theta_{1,m_1})$, and $\mathrm{cor}(\theta_{2,j}, \theta_{1,m_1})$ changes when the probability for a unit to be treated follows a Bernoulli distribution with probability $0.5$, $m_1 = 3$, and $m_2 = 1$.
Note that $\alpha_1, \alpha_2, \alpha_3$ does not depend on the covariances when the variances of parameters are taken to maximize $\alpha_3$.
Hence, we consider when $\V(\theta_{1,m}) = 0.00001$ and $\V(\theta_{1,m_1}) = \V(\alpha) = \V(\theta_{2,j}) = 1$.
Since the variances are equal to 1, the correlations here are equivalent to the covariances between the parameters.
The sign of $\mathrm{cor}(\alpha, \theta_{2,j})$ indicates whether the variance of potential outcomes with exposures where $e_2 = j$ increases or decreases---a negative correlation indicates a decrease in variance while a positive correlation indicates an increase in variance. 
Hence, as the $\mathrm{cor}(\alpha, \theta_{2,j})$ increases, the weight $\alpha_1$ increases while the weight $\alpha_2$ decreases---one is more certain about the parameters corresponding to exposures in $\hat{\theta}_{1,m_1}^{\text{two term}, 0}$ than the parameters corresponding to exposures in $\hat{\theta}_{1,m_1}^{\text{two term}, j}$.
The trajectory of weight $\alpha_3$ follows a similar pattern of the trajectory of $\alpha_1$ as $cor(\alpha, \theta_{2,j})$ varies, but at a smaller magnitude.
By a similar argument, as $\mathrm{cor}(\alpha, \theta_{1,m_1})$ increases, the weight $\alpha_1$ decreases as the weight $\alpha_2$ increases.
However, the trajectory of weight $\alpha_3$ now follows a similar pattern as the trajectory of $\alpha_2$.
Since $\alpha_2$ and $\alpha_3$ correspond to estimators whose support contains exposure $(m_1,j)$, we also see that $\alpha_2$ and $\alpha_3$ decreases as $\mathrm{cor}(\theta_{2,j}, \theta_{1,m_1})$ increases and $\alpha_1$ increases with $\mathrm{cor}(\theta_{2,j}, \theta_{1,m_1})$.
Overall, when covariances are non-zero, the weight $\alpha_3$ is generally smaller than when covariances are zero.
However, even when covariances are non-zero, we see that exposures with $e_1 = m$ may contribute to the estimation of $\theta_{1,m_1}$.

\section{Simulations} \label{sec:simulations}
In the previous section, we characterized MIV LUEs through their supports which vary with the prior distribution.
Here, we evaluate the performance of the MIV LUEs presented in Section \ref{sec:mivlue} through simulations to estimate network interference effects as described in Example \ref{network_interference_ex}.
Recall that we assume that the potential outcomes of a unit depend on the unit's treatment and the treatment of the unit's neighbors.
In particular, we assume a binary treatment and that the potential outcome of a unit depends on the number of treated neighbors and not necessarily which units are treated.
The set of exposures is given by $\mathcal{E} = \{\Vec{e}: \Vec{e} = (d_i^{\mathbf{z}}, z_i)\}$ where $d_i^{\mathbf{z}} \in \{0, \dotsc, d_i\}$ is the treated degree, or the number of treated neighbors.
Furthermore, $d_i$ is the degree of unit $i$, and $z_i \in \{0,1\}$ is the treatment assignment of unit $i$. 
Under additivity, the potential outcome of unit $i$, given exposure $\Vec{e}_i$, is given by:
\begin{align}
    Y_i(\Vec{e}_i) = \alpha^{(i)} + \theta^{(i)}_{2,1}z_i + \sum_{d=1}^{d_i} \theta^{(i)}_{1,d} \mathbb{I}\{d_i^{\mathbf{z}} = d\}.
\end{align}
Here, we include $i$ as a superscript and subscript to indicate different parameters, exposures, treatment assignments, and treated degrees for different units.
We focus on directed networks here, but results can be applied to undirected graphs.
In the case of directed networks, $d_i$ is the in-degree of unit $i$ or the number of edges pointing at unit $i$.

The parameter of interest is $\theta^{(i)}_{1,d_i}$, which is the interference effect when the treated degree is equal to the degree of the unit versus when the treated degree is zero.
More specifically, we are interested in the average interference effect when all of the neighbors of a unit versus none are treated, $\bar{\theta}_{1,d_i} = \frac{1}{n}\sum_{i=1}^n \theta^{(i)}_{1,d_i}$. 
We compare the performance of various linear estimators with inverse probability of exposure weighting:

\noindent \textbf{Two-term Horvitz-Thompson for untreated units:} 
Horvitz-Thompson inverse probability weighting estimators where $w_i(\Vec{e}_i) = \frac{\mathbb{I}\{\Vec{e}_i = (d_i, 0)\}}{n p(d_i, 0)} - \frac{\mathbb{I}\{\Vec{e}_i = (0, 0)\}}{np(0, 0)}$.
The two-term Horvitz-Thompson estimator for untreated units is unbiased even when additivity does not hold. 
We denote this estimator as $HT_0$.

\noindent \textbf{Two-term Horvitz-Thompson for treated units:}  
Horvitz-Thompson inverse probability weighting estimators where $w_i(\Vec{e}_i) = \frac{\mathbb{I}\{\Vec{e}_i = (d_i, 1)\}}{n p(d_i, 1)} - \frac{\mathbb{I}\{\Vec{e}_i = (0, 1)\}}{n p(0, 1)}$.
We denote this estimator as $HT_1$.

\noindent \textbf{Average Horvitz-Thompson:} Horvitz-Thompson inverse probability weighting estimator the average of the previous two estimators, denoted as $HT_{Avg} = \frac{1}{2}\left(HT_0 + HT_1 \right)$.

\noindent \textbf{MIV LUE with Independent Priors:}
LUE where weights are given by solving the MIV LUE problem with prior distributions: $\alpha^{(i)}, \theta^{(i)}_{2,1}, \theta^{(i)}_{1,1}, \theta^{(i)}_{1,d} \sim \mathcal{N}(0,1)$ for all $d \in \{1, \dotsc, d_i\}$.
We assume that priors are uncorrelated between units and independent between parameters. 
We denote this estimator as $M_{Ind}$.

\noindent \textbf{MIV LUE with Dilated Priors:} 
LUE where weights are given by solving the MIV LUE problem with prior distributions: $\alpha^{(i)} \sim \mathcal{N}(0,1)$, $\theta^{(i)}_{{2,1}} = \alpha^{(i)}$, and $\theta^{(i)}_{{1,d}} = \frac{d}{d_i} \eta_1 \times \alpha^{(i)}$ for all $d \in \{1, \dotsc, d_i\}$ for a fixed value of $\eta_1$.
Note that the prior variances are: $\V(\theta^{(i)}_{2,1}) = 1$ and $\V(\theta^{(i)}_{{1,d}}) = \left(\frac{d}{d_i} \eta_1 \right)^2$.
We assume that priors are uncorrelated between units.
However, there are covariances between the parameters. 
We let $\eta_1 = 1$ for our simulations, and we denote this estimator as $M_{Dil}$.

Recall that \citet{aronow2017estimating} proposed the linear unbiased estimator $HT_0$ to estimate the network interference effect.
Under additivity, the other estimators: $HT_1$, $HT_{Avg}$, $M_{Ind}$, and $M_{Dil}$ are also linear unbiased estimators.
However, the supports of $HT_0$ and $HT_1$ are of size two, whereas $HT_{Avg}$, $M_{Ind}$, and $M_{Dil}$ puts non-zero weights on more than two exposures.
Specifically, the support of $HT_{Avg}$ is equal to the union of the supports of $HT_0$ and $HT_1$, while the supports of $M_{Ind}$ and $M_{Dil}$ may be equal to the entire set of exposures.

We fix the design to be a Bernoulli design where the probability of being treated is 0.5. 
The probability of a given exposure is then given by:
\begin{align}
    \mathbb{P}(\Vec{e} = (d, z)) = \binom{d_i}{d} 0.5^{d_i + 1}.
\end{align}
For each simulation, we generated 1000 sets of parameters for each unit to generate the potential outcomes.
Unless otherwise specified, we generated the parameters for the potential outcomes as follows: $\alpha^{(i)}, \theta^{(i)}_{2,1}, \theta^{(i)}_{1,d} \sim \mathcal{N}(0,1)$ for all $d \in \{1, \dotsc, d_i\}$.
Hence, the sampling distribution for the parameters for the potential outcomes may be different from the prior distributions used for the estimators of interest.

We compare the estimators using the integrated mean squared error (IMSE), which is integrated over the parameters as in the integrated variance.
Under additivity, the estimators considered are unbiased, and so IMSE is largely driven by the integrated variance.
We evaluate the performance of the estimators under settings of varying the number of units and the number of edges, varying the level of additivity and interference, and varying the potential outcome distributions.

\subsection{Varying Number of Units and Number of Edges}\label{sec:vary_units_edges}
We first investigate how the IMSEs of the estimators change as we vary the size of the network. 
In particular, we vary the number of units and the number of edges in a network.
For each $n = 10, 20, \dotsc, 50$, we generated k-regular directed networks, where each unit has in-degree k, for $k = 2, 4, 6, 8$. 
For example, Figure \ref{fig:40_node_directed_graph} shows a 4-regular directed (as indicated by the arrows) network with 40 nodes.
We fixed the networks while we sampled different sets of potential outcome parameters and iterated through the treatment allocations.
Since each unit in a k-regular network has the same in-degree, each unit contributes equally to the estimation of the average interference effect. 
\begin{figure}[!tb]
    \center
  \begin{minipage}[t]{\linewidth}\centering
    \includegraphics[width=12cm]{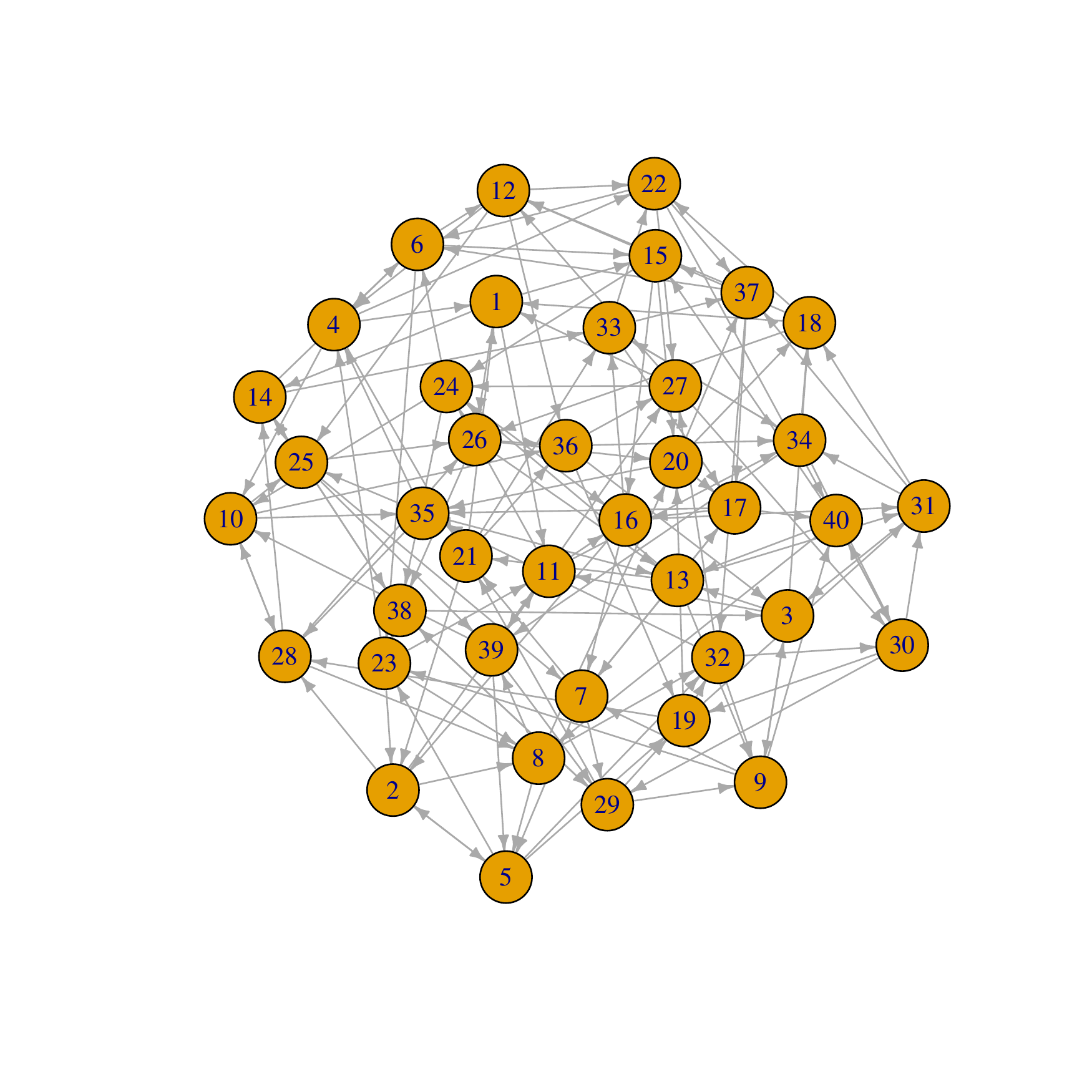}
  \end{minipage}\hfill
    \caption{Directed 4-regular network with forty nodes and each unit has in-degree of four. Arrows indicate directions of edges.}
    \label{fig:40_node_directed_graph}
\end{figure}
For networks with $n = 10$, for each sampled set of parameters, we computed the IMSE using all $2^{10}$ possible treatment allocations.
For networks with $n = 20, 30, 40, 50$, we computed the IMSE over a sample of 1500 treatment allocations.
Hence, for $n = 10$, we computed the exact integrated bias and variance whereas we estimated these for $n = 20, 30, 40, 50$.
Potential outcomes were simulated under additivity, i.e. there were no interaction effects.
Hence, in this simulation setting, we expect $HT_{Ind}$ to perform the best since the prior of $HT_{Ind}$ matches the distributions of the parameters for the potential outcomes.

\begin{figure}[!tb]
    {\center
    \includegraphics[width=\textwidth]{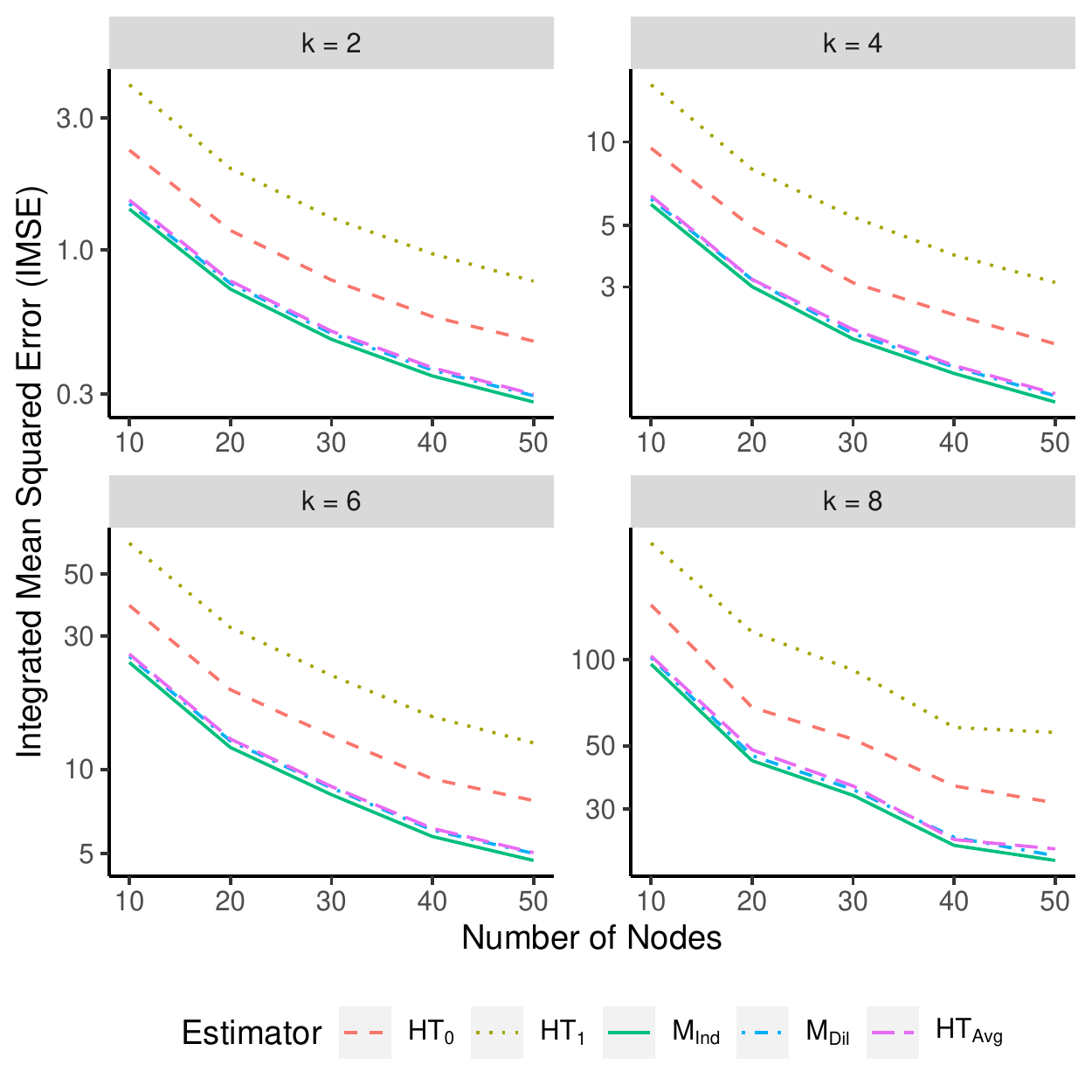}
    }
    \caption{IMSE for estimators (indicated by color and line type) when the number of units (indicated by x-axis) increases for a k-regular network for different values of $k$ (indicated by panel) under additivity and when mean interference is zero.}
    \label{fig:imse_numNodes_er_and_4_regular}
\end{figure}

Figure \ref{fig:imse_numNodes_er_and_4_regular} shows the IMSE for the estimators as the number of units (indicated by x-axis) increases for different values of $k$ (indicated by panels).
Overall, the IMSE decreases as the number of units increases.
Since all units have the same in-degree and hence the same exposure distribution, increasing the number of units leads to a decrease in the IMSE.
On the other hand, as the number of edges (or in-degree) $k$ increases, the IMSE increases for all estimators.
This is possibly explained by the fact that as the number of edges increases, the probability of a unit having treated degree zero or treated degree $d_i$ decreases.
Weights on exposures with treated degrees equal to zero or $d_i$ then increase with $k$ since weights are inversely related to the probabilities of exposures.
On the other hand, weights on other exposures are either zero (for two-term HT estimators and $HT_{Avg}$) or are relatively smaller (for $M_{Ind}$ and $M_{Dil}$) since the probability of exposures with $e_1 \in \{1, \dotsc, d_i-1\}$ increase with $k$, leading to a greater IMSE.

The red, dashed line indicates the IMSE for $HT_0$. 
Under additivity, $HT_1$ is also a linear unbiased estimator of $\bar{\theta}_{1,d_i}$, but $HT_1$ has higher IMSE than $HT_0$.
This is likely due to the extra variance introduced by the treated units.
However, there is a significant reduction in IMSE across the different panels when we average both $HT_0$ and $HT_1$.
Indeed, the IMSE of $HT_{Avg}$, given by the purple, long-dashed line, is lower than the IMSEs of $HT_0$ and $HT_1$.
There is an additional reduction in IMSE when we take the integrated variance into account and compute weights to minimize the integrated variance.
Although the IMSEs of $M_{Ind}$ and $M_{Dil}$ are only slightly lower than the IMSE of $HT_{Avg}$, we still see the benefit of using optimal weights. 
Furthermore, $M_{Ind}$ performs the best as expected.
The performances of the estimators suggest that there is an advantage in leveraging information from all data available as opposed to just using a subset of units.

\subsection{Varying Interference Effects and Deviations from Additivity}\label{sec:sims_additivity_interference}
Throughout this section, we derived linear unbiased estimators under the assumption that causal effects are additive. 
In this section, we examine the robustness of the MIV LUEs when the additivity assumption is violated.
We represent varying levels of additivity through an interaction effect between the direct effect and the interference effect.
Potential outcomes in this section were simulated according to the parameterization:
\begin{align}
    Y_i(\Vec{e}_i) = \alpha^{(i)} + \theta^{(i)}_{2,1}z_i + \sum_{d=1}^{d_i}\theta^{(i)}_{1,d}\mathbb{I}\{d_i^{\mathbf{z}} = d\} + \sum_{d=1}^{d_i} \Delta^{(i)}_{d}z_i \mathbb{I}\{d_i^{\mathbf{z}} = d\},
\end{align}
where $\alpha^{(i)}, \theta^{(i)}_{2,1} \sim \mathcal{N}(0,1)$, and $\theta^{(i)}_{1,d} \sim \mathcal{N}(\frac{d}{d_i} \mu_{1}, 1)$ and $\Delta^{(i)}_{d} \sim \mathcal{N}(\frac{d}{d_i} \delta_{1}, \mathbb{I}\{\delta_{1,d_i} > 0\})$ are the interference effects and interaction effects, respectively.
When $\Delta^{(i)}_{d} = 0$ for all $d \in \{1, \dotsc, d_i\}$, additivity holds.
We simulated potential outcomes under $\mu_{1} \in \{0, 10, 50\}$ and $\delta_{1} \in \{0, 2, 4, 6\}$.
Note when $\delta_{1} = 0$, we set $\V(\Delta_d^{(i)}) = 0$ so that $\Delta_d^{(i)} = 0$ to ensure that additivity holds.
Even though interference effects were not necessarily mean zero, we maintained zero-mean priors to evaluate the performance of our estimators when the priors do not match the potential outcome distributions.
In particular, we estimated average network interference effects on a 4-regular graph when $n = 40$.

\begin{figure}[!tb]
 
{\centering
    \includegraphics[width=\linewidth]{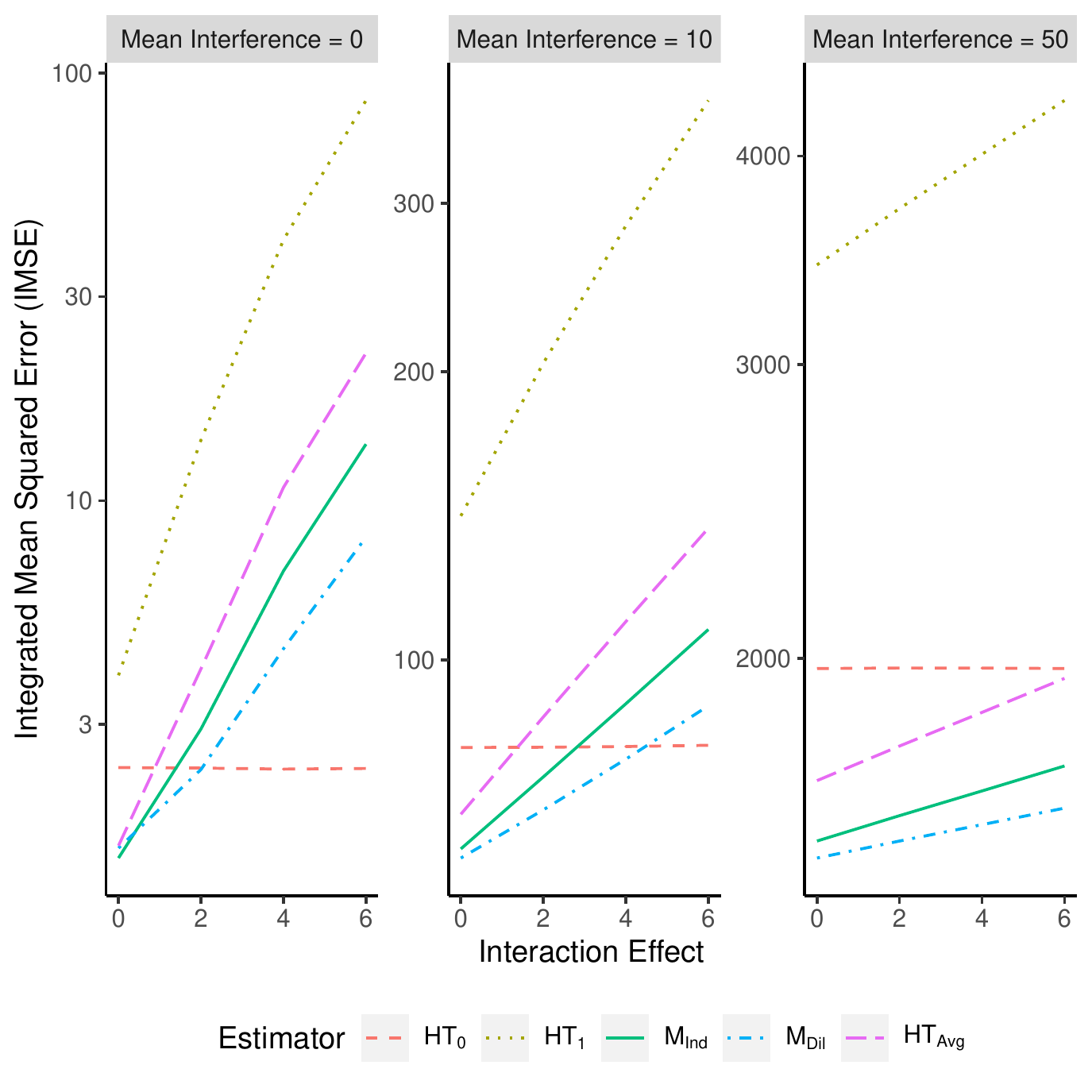}}
    \caption{IMSE for estimators (indicated by color and line type) when the interaction (indicated by x-axis) and interference effects (indicated by panel) vary for a 40 node 4-regular graph.}
    \label{fig:imse_additivity_interference}
\end{figure}
Figure \ref{fig:imse_additivity_interference} shows the IMSE of the estimators as the interaction effect increases when the mean interference effect is 0, 10, and 50 (indicated by the panels).
As the mean interference effect increases across the three panels, the IMSE increases for all estimators.
Since we used zero-mean priors to derive the MIV LUEs, it is reasonable that when the potential outcome distributions stray further away from the prior distribution, all estimators do not perform as well.
When the true mean interference effect is zero and additivity holds, $M_{Ind}$ outperforms the other estimators, as expected.
However, as the mean interference effect increases, $M_{Dil}$ actually outperforms $M_{Ind}$, despite the fact that the potential outcome parameters are independent.
Hence, there may be slight concerns when using an estimator with weights obtained from a prior distribution different from the potential outcome distribution.
Even though $M_{Ind}$ does not perform the best when the mean interference is non-zero, in general, the multi-term MIV LUEs outperform the other estimators.

As the interaction effect (indicated by the horizontal axis) increases, the IMSEs of estimators increase in general.
However, since $HT_{0}$ puts non-zero weight on untreated units, it is invariant to the interaction effect.
Furthermore, it is the only estimator considered that is unbiased even when additivity does not hold.
When the interaction effect is not zero, the other estimators are biased, which partially explains the increase in IMSE as the interaction effect increases.
In particular, $HT_{1}$ performs the worst as it only puts non-zero weights on treated exposures, and so the interaction effect is always present.
However, even when the interaction effect is non-zero, i.e. when additivity does not hold, we see that there are instances when $HT_{Avg}$, $M_{Ind}$, and $M_{Dil}$ outperform $HT_0$.
This is especially seen as the mean interference effect increases. 
Indeed, when the mean interference effect is equal to 50, the three estimators outperform $HT_{0}$ for all of the values of interaction effects considered.
This suggests that the estimators are fairly robust to violations of the additivity assumption, especially when the mean interference effect is large.
Furthermore, as the interaction effect increases, there is a bigger distinction between the IMSE of $HT_{Avg}$ and the IMSEs of $M_{Ind}$ and $M_{Dil}$, which was not seen in the previous section when additivity holds.
Hence, there is a benefit in using $M_{Ind}$ and $M_{Dil}$, over $HT_{Avg}$, especially when additivity does not hold.

\subsection{Varying Potential Outcome Distributions}
Lastly, we compare estimators in settings with different potential outcome parameter distributions.
In the previous sections, potential outcomes were sampled such that units and parameters were independent. 
When additivity holds and the true mean interference effect is zero, $M_{Ind}$ outperforms the other estimators.
In this section, in addition to the independent parameters, we also simulated potential outcome parameters under a dilated distribution where parameters are correlated.
That is, $\alpha^{(i)} \sim \mathcal{N}(0,1)$, $\theta^{(i)}_{2,1} = \alpha^{(i)}$, and $\theta^{(i)}_{1,d} = \frac{d}{d_i} \eta_{1} \alpha^{(i)}$ for $\eta_1 = 0, 1, 5, 10, 50$.
Under this setting, we expect $M_{Dil}$ to perform the best.
We compared results for a 4-regular graph with forty nodes. 

\begin{figure}[!tb]
    \center
  \begin{minipage}[t]{\linewidth}\centering
    \includegraphics[width=12cm]{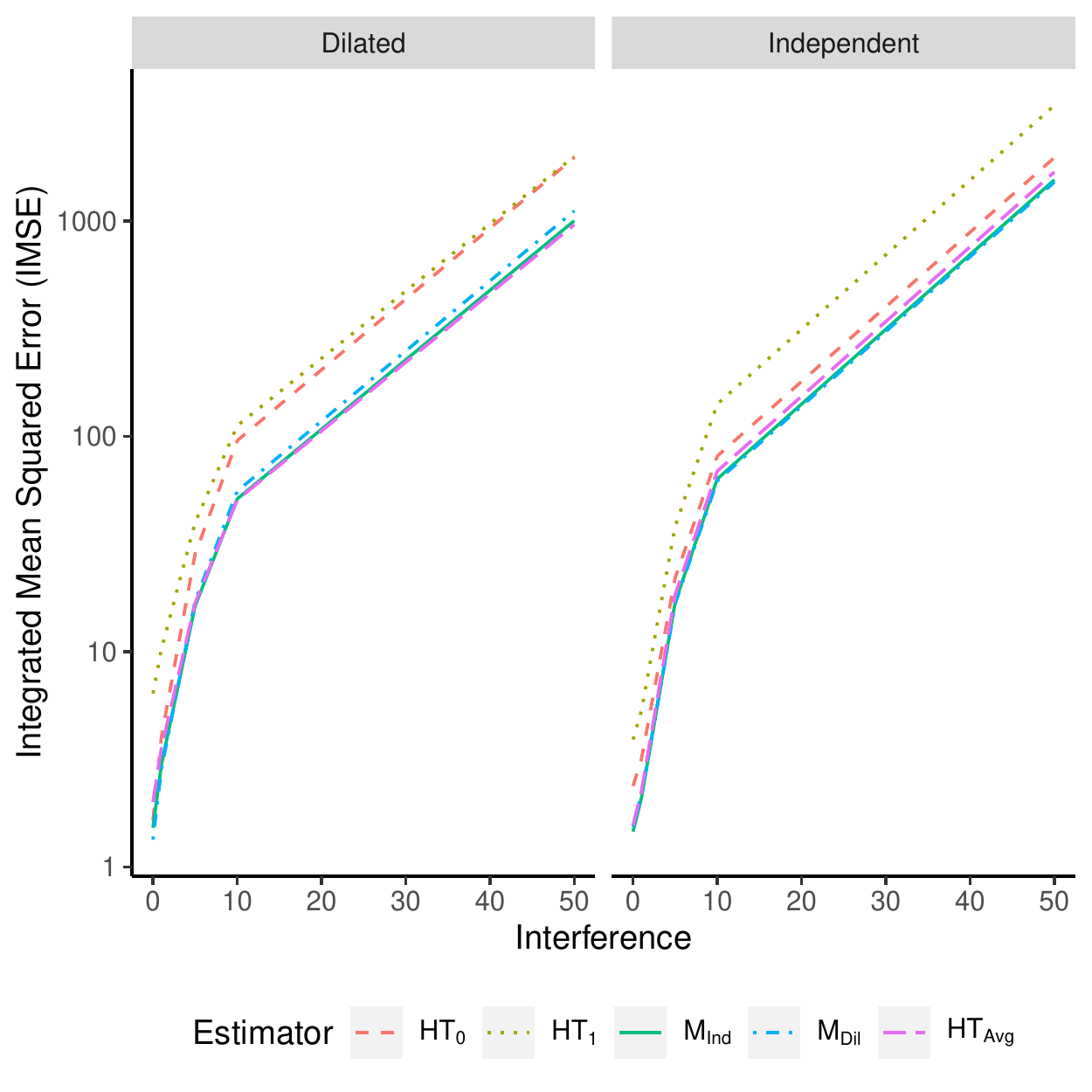}
  \end{minipage}\hfill
    \caption{IMSE for estimators (indicated by color and line type) under different potential outcome distributions (indicated by panel) as the interference effect varies (indicated by x-axis) under additivity for a 40-node 4-regular graph.}
    \label{fig:imse_dilated_v_independent_40_node_4_regular}
\end{figure}
Figure \ref{fig:imse_dilated_v_independent_40_node_4_regular} shows the IMSEs of the different estimators under the independent and dilated potential outcome distributions (indicated by panel) as we vary $\eta_1$ or $\mu_1$ (indicated by x-axis) for the dilated and independent distributions, respectively, and assuming that additivity holds.
The IMSEs for estimators under the two different potential outcome distributions are fairly similar, with estimators using potential outcome parameters sampled from independent Normal distributions having slightly higher IMSEs.
As seen in the results in Section \ref{sec:sims_additivity_interference}, the IMSEs of the estimators increase as the interference effect increases. 

The performance of the estimators under different potential outcome distributions reflected the results seen in Section \ref{sec:vary_units_edges}.
$HT_1$ performs worse than $HT_0$, but $HT_{Avg}$, $M_{Ind}$, and $M_{Dil}$ outperform $HT_0$, with $M_{Ind}$ and $M_{Dil}$ generally performing the best. 
The multi-term MIV LUE whose prior distribution matches the distribution of the potential outcomes performs the best when the true mean interference effect was low, as expected.
However, when $\mu_1$ and $\eta_1$ increase, the multi-term MIV LUE whose prior distribution matches the distribution of the potential outcomes does not perform as well.
Even so, the IMSEs of the two multi-term MIV LUEs are comparable.
Hence, even if we use a prior distribution that does not match that of the potential outcomes, there is benefit in the multi-term MIV LUEs since they outperform the other estimators.

\section{Discussion}
We proposed linear unbiased estimators for general causal effects as specified by exposure mappings under the assumption of additivity across exposure components.
Under this assumption, the space of linear unbiased estimators becomes much larger, and exposures that are ``seemingly unrelated'' to the estimand of interest can contribute to the estimation.
We can then leverage the information from units under other exposures that are not diectly related to the estimand of interest.
Given the set of exposures, we defined linear constraints for when these LUEs exist, and we introduced a class of atomic estimators which, when combined with some unbiased estimators for zero, forms an affine basis for the set of LUEs.
Additionally, we characterized an optimal subset of LUEs with minimum integrated variance.

In general, there is benefit to adding non-zero weight to more exposures.
Even if we just take the average of the two-term Horvitz-Thompson estimators for untreated and treated units (hence putting non-zero weight on four exposures), we saw a significant reduction in IMSE compared to the IMSEs of each of the two-term estimators separately.
If we further compute optimal weights for a LUE given a prior distribution, there is an additional reduction in the IMSE.
However, these multi-term estimators are only LUEs under additivity.
Under additivity, these multi-term estimators perform well in practice.
Although we require additivity for theoretical results, the multi-term estimators are fairly robust to violations of additivity in practice.
In fact, these multi-term estimators outperform two-term estimators for low levels of interaction effects and large interference effects.

Aside from additivity, we assumed that priors were uncorrelated between units which allowed for easier computation of the variances of estimators.
By assuming independent priors between units, we only had to account for the prior variances for unit $i$ when computing the LUEs for the unit-level effect.
Estimators may be derived to account for the covariance between units when computing the integrated variance.
However, when priors are correlated between units, we are likely not able to derive a closed-form solution.
Furthermore, we did not discuss estimators for the variance in this section.
To derive estimators for the variance, we may leverage the work done in \citet{aronow2017estimating}, who derived estimators for variances for two-term estimators. 
However, in our work, estimators may have more than two terms.
Since we have to account for covariances between the various Horvitz-Thompson terms, estimators for the variance could be quite complicated.

Although we focused on experimental settings, we would like to extend our multi-term MIV LUEs to observational studies as a next step.
In the context of observational studies, we would likely have to account for noise in the exposure mapping and noise in the probability of exposures.
In our current work, we did not make any assumptions about the treatment effects, but we did assume that the exposure mapping was known.
If the exposure mapping used is not the true underlying exposure mapping, which could happen in both experiments and observational studies, then the results may not be accurate.
\citet{aronow2017estimating} showed that in the case when an exposure mapping maps two treatment allocations to the same exposure, but the potential outcomes under the two treatment allocations are different, the two-term estimator $HT_0$ is unbiased for a weighted average of the potential outcomes under the different treatment allocations.
In our case, we could possibly account for the various types of noise in the exposure mapping.
Furthermore, we assumed that the probabilities of exposures were known, which is typically not true in observational studies.
In observational studies, we would have to estimate the probability of exposures using a model given covariate variables.
Therefore, we would like an estimator that is doubly robust \citep{robins1994estimation, li2021causal}.
However, unlike the typical doubly robust models, where one can make misspecifications in the outcome model or the treatment model, we would ideally want an estimator that is robust to misspecifications in the exposure mapping and/or the probability of exposure model.
Lastly, the inclusion of covariates was not discussed in this work, but one would likely benefit from including information from covariates when estimating treatment effects and can better quantify treatment effect heterogeneity.
Again, we could possibly leverage the work of \citet{aronow2017estimating} who proposed linear unbiased estimators for treatment effects using models that account for the covariates.

In summary, we characterized the set of linear unbiased estimators under the assumption of additive exposures. 
We further specified conditions of the supports of estimators that lead to MIV LUEs with non-zero weights on all exposures in the support.
Using these proposed MIV LUEs, we saw an added benefit of incorporating information from all units as opposed to two-term LUEs which only place non-zero weight on units with exposures in the estimand of interest.

\section*{Acknowledgements}
This work was supported in part by the Air Force Research Laboratory and DARPA under agreement numbers FA8750-18-2-0035 and FA8750-20-2-1001. The U.S. Government is authorized to reproduce and distribute reprints for Governmental purposes notwithstanding any copyright notation thereon.
Any opinions, findings, and conclusions or recommendations expressed in this material are those of the author(s) and do not necessarily reflect the views of the supporting institutions.

\bibliographystyle{plainnat}

\bibliography{ref.bib}

\appendix

\section{Linear Unbiased Constraints}\label{appendix:lue_constraint}

\begin{proof}[Proof of Proposition \ref{constraints}]
Let $\hat{\theta}_{1,m_1}$ be a linear estimator for the unit-level causal effect for a unit $i$ whose weights only depend on the unit's exposure, i.e. $\hat{\theta}_{1,m_1} = w(\Vec{e})Y(\Vec{e})$, where $\Vec{e}^{\,obs}_i = \Vec{e}$.
Without the loss of generality, we assume that the exposure component of interest is the first one.
The parameter of interest for the unit-level causal effect of the first exposure component being $m_1$ versus zero is denoted as $\theta_{1,m_1}$. 
The expected value of the estimator under additivity is given as follows:
\begin{align*}
    \E(\hat{\theta}_{1,m_1}) &= \sum_{\Vec{e} \in \mathcal{E}} p(\Vec{e})w(\Vec{e})Y(\Vec{e}) \\
    &= \sum_{\Vec{e} \in \mathcal{E}} p(\Vec{e})w(\Vec{e})\bigg[\alpha + \theta_{1,m_1}\mathbb{I}\{e_1 = m_1\} \\
    &\left.+ \sum_{m=1}^{m_1-1} \theta_{1,m}\mathbb{I}\{e_1 = m\} + \sum_{k=2}^K\sum_{j_k=1}^{m_k}\theta_{k,j_k}\mathbb{I}\{e_k = j_k\}\right].
\end{align*}
In order for $\E(\hat{\theta}_{1,m_1}) = \theta_{1,m_1}$, we need:
\begin{align*}
    &\sum_{\Vec{e} \in \mathcal{E}}p(\Vec{e})w(\Vec{e}) = 0\\
    &\sum_{\Vec{e} \in \mathcal{E}}p(\Vec{e})w(\Vec{e})\mathbb{I}\{e_1 = m_1\} = 1\\
    \forall m: m \in \{1, \dotsc, m_1-1\} & \sum_{\Vec{e} \in \mathcal{E}}p(\Vec{e})w(\Vec{e})\mathbb{I}\{e_1 = m\} = 0 \\
    \forall k,j_k: k \in \{2, \dotsc, K\}, j_k \in \{1, \dotsc, m_k\} & \sum_{\Vec{e} \in \mathcal{E}}p(\Vec{e})w(\Vec{e})\mathbb{I}\{e_k = j_k\} = 0.
\end{align*}
These give us the constraints needed for unbiasedness.
\end{proof}

\section{Affine Basis for Set of LUEs $\mathcal{U}$}\label{appendix:affine_basis_for_lue}
We first define notation for the weight of exposures in an estimator.
\begin{definition}[Weight of Exposure in Estimator in $\mathcal{M}$]
Let $\hat{\theta} \in \mathcal{M}$.
We denote the weight on the Horvitz Thompson term associated with exposure $\Vec{e} \in \mathcal{E}$ as $f_{\hat{\theta}}(\Vec{e}): \mathcal{E} \to \{-1, 0, 1\}$. 
\end{definition}

\subsection{Construction and Affine Independence of $\mathcal{M}$}\label{appendix:proof_of_affine_ind_of_m}

\begin{proof}[Proof of Lemma \ref{construction}]
Let $\mathcal{M}$ be the set of estimators described in Theorem \ref{construction}. 
Note that all estimators in $\mathcal{M}$ are MALUEs.
Let $\hat{\theta} = \sum_{\tilde{\theta} \in \mathcal{M}} g(\tilde{\theta})\tilde{\theta}$ where $\sum_{\tilde{\theta} \in \mathcal{M}}g(\tilde{\theta}) = 1$. 
Assume $\hat{\theta} \in \mathcal{M}$. 
If $\mathcal{M}$ is affine independent, then it implies that 
\begin{align*}
    g(\tilde{\theta}) = \begin{cases}1, & \text{if } \tilde{\theta} = \hat{\theta}\\
    0, & \text{otherwise} 
    \end{cases}.
\end{align*} 
Since $\tilde{\mathcal{E}}$ is ordered and each MALUE in $\mathcal{M}$ is uniquely identified by the exposures in $\tilde{\mathcal{E}}$, there is a natural ordering of the corresponding MALUEs.
We prove that $\mathcal{M}$ is affine independent using induction.
Consider each of the MALUEs $\tilde{\theta} \in M$.

\noindent \textbf{Base Case: $j = 1$} \newline
Consider the first MALUE $\tilde{\theta}^{(1)} \in \mathcal{M}$. 
In particular, we have the MALUE:
\begin{align*}
    \tilde{\theta}^{(1)} =&\phantom{+}  HT_{(m_1, m_2, \dotsc,m_K)} - HT_{(m_1-1, m_2, \dotsc, m_K)} \\
    &+ HT_{(m_1-1, 0, m_3, \dotsc, m_K)} - HT_{(0, 0, m_3, \dotsc, m_K)}
\end{align*}
Based on the construction of $\mathcal{M}$, $\tilde{\theta}$ is a unique MALUE for which $\Vec{e}^{\,(1)} \in \mathrm{supp}(\tilde{\theta})$ where $\Vec{e}^{\,(1)} = (m_1 - 1, m_2, \dotsc, m_K)$.
Thus, if $\Vec{e}^{\,(1)} \in \mathrm{supp}(\hat{\theta})$, then $\tilde{\theta} = \hat{\theta}$, i.e. $g(\tilde{\theta}) = 1$. 
Otherwise, $g(\tilde{\theta}) = 0$.

\noindent \textbf{Induction Hypothesis:} Now assume that for $j \in \{2, \dotsc, u\}$, the weight $g$ for the $j$th MALUE in $\mathcal{M}$ is given by Equation~\eqref{def_of_g}.

\noindent \textbf{Case:} $j = u+1$ \newline
Now consider the $(u+1)$th estimator $\tilde{\theta}^{(u+1)} \in \mathcal{M}$. 
From the definition of the estimators in $\mathcal{M}$, $\tilde{\theta}^{(u+1)}$ is uniquely identified by an exposure $\Vec{e}^{\,(u+1)} \in \tilde{\mathcal{E}}$ where either $e_1^{\,(u+1)} \in \{1, \dotsc, m_1-1\}$ or $e_1^{\,(u+1)}  = m_1$, depending on whether $\tilde{\theta}^{(u+1)}$ is a four-term or two-term MALUE, respectively. 
If $\Vec{e}^{\,(u+1)} \notin \mathrm{supp}(\hat{\theta})$, then $g(\tilde{\theta}^{(u+1)}) = 0$. 
Otherwise, we consider the different cases.

If for all $j'$th MALUEs $\tilde{\theta}^{(j')} \in \mathcal{M}$ for $j' < j$ we have $g(\tilde{\theta}^{(j')}) = 0$ and $\Vec{e}^{\,(u+1)} \in \mathrm{supp}(\hat{\theta})$, then $g(\tilde{\theta}^{(u+1)}) = 1$. 
This is because $\mathcal{M}$ is ordered and since $\tilde{\theta}^{(u+1)}$ is a MALUE, the exposure components in the exposures corresponding to each of the HT terms are simultaneously decreasing.
Hence, $\tilde{\theta}^{(u+1)}$ is the last MALUE in $\mathcal{M}$ for which $\Vec{e}$ is in its support. 
Since all previous MALUEs $\tilde{\theta}^{(j')}$ have weight 0, then $\tilde{\theta}^{(u+1)}$ must be equal to $\hat{\theta}$, i.e. $g(\tilde{\theta}^{(u+1)}) = 1$.  

If there exists a MALUE $\tilde{\theta}^{(j')} \in \mathcal{M}$ for $j' < j$ such that $g(\tilde{\theta}^{(j')}) = 1$, it means that $\tilde{\theta}^{(j')} = \hat{\theta}$. 
Because of the induction hypothesis, there can be at most one estimator before the $(u+1)$th estimator that has non-zero weight (namely the estimator equal to $\hat{\theta}$). 
If $\Vec{e}^{\,(u+1)} \in \mathrm{supp}(\hat{\theta})$, then $\Vec{e}^{\,(u+1)} \in \mathrm{supp}(\tilde{\theta}^{(j')})$.
Furthermore, $f_{\hat{\theta}}(\Vec{e}^{\,(u+1)}) = f_{\tilde{\theta}^{(j')}}(\Vec{e}^{\,(u+1)})$.
If $g(\tilde{\theta}^{(u+1)}) \neq 0$, the weight $f_{\tilde{\theta}^{(j')} + g(\tilde{\theta}^{(u+1)}) \tilde{\theta}^{(u+1)}}(\Vec{e}^{\,(u+1)})$ is given by  $f_{\tilde{\theta}^{(j')}}(\Vec{e}^{\,(u+1)}) + g(\tilde{\theta}^{(u+1)})f_{\tilde{\theta}^{(u+1)}}(\Vec{e}^{\,(u+1)})$. 
Since $\tilde{\theta}^{(u+1)}$ is the last MALUE in $\mathcal{M}$ with exposure $\Vec{e}^{\,(u+1)}$ in its support, there are no other MALUEs for which we cancel the extra weight on the Horvitz Thompson term with exposure $\Vec{e}^{\,(u+1)}$. 
Thus, if $g(\tilde{\theta}^{(u+1)}) \neq 0$, we have $f_{\tilde{\theta}^{(j')}}(\Vec{e}^{\,(u+1)}) + g(\tilde{\theta}^{(u+1)})f_{\tilde{\theta}^{(u+1)}}(\Vec{e}^{\,(u+1)}) \neq f_{\tilde{\theta}^{(j')}}(\Vec{e}^{\,(u+1)}) = f_{\hat{\theta}}(\Vec{e}^{\,(u+1)})$. 
This leads to biasedness.
Hence, $g(\tilde{\theta}^{(u+1)})$ must be 0. 

For all estimators in $\tilde{\theta} \in \mathcal{M}$, the weights $g(\tilde{\theta})$ are defined as in Equation~\eqref{def_of_g}. 
Thus, if $\hat{\theta} \in \mathcal{M}$, then $\hat{\theta}$ cannot be written as an affine combination of estimators $\tilde{\theta}$ in $\mathcal{M}$, i.e. the set of estimators $\mathcal{M}$ are affine independent. 

\end{proof}

\subsection{Affine Basis for $\mathcal{U}$}\label{appendix:proof_affine_basis_for_lue}

\begin{proof}[Proof of Theorem \ref{affine_basis_thm}]
First, we want to show that $\hat{\Theta}$ is affine independent. 
Consider the support of $\hat{\Theta}$:
\begin{align}
    \mathrm{supp}(\hat{\Theta}) &= \{\Vec{e}: \Vec{e} \in \mathcal{E}, e_1 = m_1\} \nonumber\\
    &\cup \{\Vec{e}: \Vec{e} \in \mathcal{E}, e_1 \in \{1, \dotsc, m_1-1\}, \exists k \in \{2, \dotsc, K\} \text{ s.t. } e_k \neq 0\} \nonumber \\
    &\cup \{\Vec{e}: \Vec{e} \in \mathcal{E}, e_1 = 0, \exists k,k' \in \{2, \dotsc, K\} \text{ s.t. } e_k \neq 0, e_{k'} \neq 0\}.
\end{align}
Note that each estimator in $\hat{\Theta}$ can be uniquely identified by an exposure.
In particular, each two-term estimator in $\mathcal{M}$ can be uniquely identified by an exposure where $e_1 = m_1$. 
Each four-term estimator in $\mathcal{M}$ can be uniquely identified by an exposure where $e_1 \in \{1, \dotsc, m_1-1\}$ and there is at least another non-zero exposure.
Each estimator in $\mathcal{Z}$ can be uniquely identified by an exposure where $e_1 = 0$ and there are at least two other non-zero exposure components. 

Order the set of exposures such that the set of exposures with $e_1 \in \{1, \dotsc, m_1-1\}$ are first, then the exposures with $e_1 = m_1$, and finally $e_1 = 0$. 
Within each subset of exposures in $\mathrm{supp}(\hat{\Theta})$, order the exposures in a reverse reflected lexicographic manner by uniquely identifying exposure. 
Recall that each estimator in $\mathcal{M}$ is a MALUE, and so the exposures, and their corresponding Horvitz-Thompson terms, can be arranged such that the exposure components are simultaneously non-increasing.
Furthermore, exposures, and their corresponding Horvitz-Thompson terms, in the support of estimators in $\mathcal{Z}$ can also be ordered.
In particular, exposures in the support of an estimator in $\mathcal{Z}$ can be arranged, in increasing order, according to the reverse reflected lexicographic order.
Hence, the zero estimators are also monotonic.
Because of the ordering of $\mathrm{supp}(\hat{\Theta})$ and the monotonicity of the estimators, each estimator $\tilde{\theta} \in \hat{\Theta}$ is the last estimator in $\hat{\Theta}$ for which the uniquely identifying exposure is in its support. 
For example, the two term estimator $HT_{(m_1, e_2, \dotsc, e_K)} - HT_{(0, e_2, \dotsc, e_K)}$, where $e_k \in \{0, \dotsc, m_k\}$, is the last estimator for which $(m_1, e_2, \dotsc, e_K)$ is in its support.

Now, let $\hat{\theta} = \sum_{\tilde{\theta} \in \hat{\Theta}} g(\tilde{\theta})\tilde{\theta}$, where $\sum_{\tilde{\theta} \in \hat{\Theta}} g(\tilde{\theta}) = 1$.
Suppose that $\hat{\theta} \in \hat{\Theta}$.
Based on the ordering of the estimators, we can extend the proof of Theorem \ref{construction} to here. 
By doing so, we argue that $g$ is given by:
\begin{align*}
    g(\tilde{\theta}) = \begin{cases}1, & \text{if } \tilde{\theta} = \hat{\theta}\\
    0, & \text{otherwise} 
    \end{cases}.
\end{align*} 
Hence, $\hat{\Theta}$ is affine independent.

Next, we show that $\hat{\Theta}$ spans the set of LUE. 
To do so, we determine the dimension of $\hat{\Theta}$ and the dimension of $\mathcal{U}$.
The number of estimators in $\tilde{\theta} \in \hat{\Theta}$ is equal to the number of uniquely identifying exposures. 
That is,
\begin{align}
    |\hat{\Theta}| &= \left(\underbrace{\prod_{k=2}^K (m_k + 1)}_{\text{two term estimators}}\right) + \left(\underbrace{(m_1 - 1)\left[\prod_{k=2}^K(m_k + 1) - 1\right]}_{\text{four term estimators}}\right) \nonumber \\
    &+ \left(\underbrace{\prod_{k=2}^K(m_k + 1) - 1 - \sum_{k=2}^K m_k}_{\text{zero estimators}}\right) \nonumber \\
    &= \prod_{k=1}^K(m_k + 1)  - \sum_{k=1}^K m_k.
\end{align}
However, since the weights of the estimators must sum to 1, there is one less free dimension.
Thus, the dimension of the affine space is 
\begin{align*}
    \prod_{k=1}^K(m_k + 1)  - \sum_{k=1}^K m_k - 1.
\end{align*}
The dimension of the LUEs is determined by the number of exposures minus the number of constraints.
Hence, we have,
\begin{align}
    |\mathcal{U}| &= \left(\underbrace{\prod_{k=1}^K(m_k + 1)}_{\text{number of exposures}}\right) - \left(\underbrace{\left[\sum_{k=1}^K m_k\right] + 1}_{\text{number of constraints}}\right) \nonumber \\
    &= \prod_{k=1}^K(m_k + 1) - \sum_{k=1}^K m_k - 1.
\end{align}
So $\hat{\Theta}$ is an affine independent set with dimension equal to the dimension of the set of LUEs, i.e. $\hat{\Theta}$ spans the set of LUEs.
Hence, the set $\hat{\Theta}$ is an affine basis for the set of LUEs.
\end{proof}

\section{MIV LUEs}
\subsection{Optimization Problem} \label{appendix:mivlue_matrix_problem}
We solve the following optimization problem to find weights for exposures for an LUE that has minimum integrated variance. 
First, let $\alpha \in \Theta$ be the parameter corresponding to the baseline (i.e. when all exposure component values are 0) and $\theta_{k,j_k} \in \Theta$ be parameters corresponding to the $k$th exposure component, where $k \in \{1, \dotsc, K\}$, when it is equal to $j_k$ versus zero, where $j_k \in \{1, \dotsc, m_k\}$. 
Suppose the parameter of interest is $\theta_{1, m_1} \in \Theta$.
To find a MIV LUE for a given prior, we find weights such that the integrated variance is minimized with respect to the linear unbiased constraints. 
Therefore, the optimization problem becomes:

\begin{align*}
    \mathcal{L} =&\phantom{+} \frac{1}{2}\int_{\Theta}\sum_{\Vec{e}}p(\Vec{e})\left(w(\Vec{e})Y(\Vec{e}) - \theta_{1,m_1}\right)^2\pi(\theta')d\theta' \\
    &+ \lambda_1\left(1 - \sum_{\Vec{e}}p(\Vec{e})w(\Vec{e})\mathbbm{I}\{e_1 = m_1\}\right)\\
    &- \sum_{m = 1}^{m_1 - 1}\lambda_{2,m}\left(\sum_{\Vec{e}}p(\Vec{e})w(\Vec{e})\mathbbm{I}\{e_1 = m\}\right) - \lambda_3\sum_{\Vec{e}}p(\Vec{e})w(\Vec{e}) \\
    &- \sum_{k = 2}^K\sum_{j_k = 1}^{m_k} \lambda_{4, k, j_k}\sum_{\Vec{e}}p(\Vec{e})w(\Vec{e})\mathbbm{I}\{e_k = j_k\}\\
    &= \frac{1}{2}\left[\sum_{\Vec{e}}p(\Vec{e})w(\Vec{e})^2\int_{\Theta}Y(\Vec{e})^2\pi(\theta')d\theta' + \sum_{\Vec{e}}p(\Vec{e})\int_{\Theta}{\theta^2_{1,m_1}}\pi(\theta')d\theta'\right. \\
    &- \left.2\int_{\Theta}{\theta_{1,m_1}}\sum_{\Vec{e}}p(\Vec{e})w(\Vec{e})Y(\Vec{e})\pi(\theta')d\theta' \right] \\
    &+ \lambda_1\left(1 - \sum_{\Vec{e}}p(\Vec{e})w(\Vec{e})\mathbbm{I}\{e_1 = m_1\}\right) \\
    &- \sum_{m = 1}^{m_1 - 1}\lambda_{2,m}\left(\sum_{\Vec{e}}p(\Vec{e})w(\Vec{e})\mathbbm{I}\{e_1 = m\}\right)- \lambda_3\sum_{\Vec{e}}p(\Vec{e})w(\Vec{e})\\
    &- \sum_{k = 2}^K\sum_{j_k = 1}^{m_k} \lambda_{4, k, j_k}\sum_{\Vec{e}}p(\Vec{e})w(\Vec{e})\mathbbm{I}\{e_k = j_k\}\\
    &= \frac{1}{2}\left[\sum_{\Vec{e}}p(\Vec{e})w(\Vec{e})^2 \V(Y(\Vec{e})) - \int_{\Theta}{\theta^2_{1,m_1}}\pi(\theta')d\theta' \right]\\
    &+ \lambda_1\left(1 - \sum_{\Vec{e}}p(\Vec{e})w(\Vec{e})\mathbbm{I}\{e_1 = m_1\}\right)\\
    & - \sum_{m = 1}^{m_1 - 1}\lambda_{2,m}\left(\sum_{\Vec{e}}p(\Vec{e})w(\Vec{e})\mathbbm{I}\{e_1 = m\}\right)\\
    &- \lambda_3\sum_{\Vec{e}}p(\Vec{e})w(\Vec{e}) - \sum_{k = 2}^K\sum_{j_k = 1}^{m_k} \lambda_{4, k, j_k}\sum_{\Vec{e}}p(\Vec{e})w(\Vec{e})\mathbbm{I}\{e_k = j_k\}
\end{align*}

Taking the derivatives with respect to exposure $\Vec{e} \in \mathcal{E}$ and setting the derivative to zero yields

\begin{align*}
    0 &= %
    p(\Vec{e})w(\Vec{e})\V(Y(\Vec{e})) - \lambda_1 p(\Vec{e})\mathbbm{I}\{e_1 = m_1\} \\
    &- \sum_{m = 1}^{m_1 - 1}\lambda_{2,m} p(\Vec{e})\mathbbm{I}\{e_1 = m\}  - \lambda_3 p(\Vec{e}) - \sum_{k = 2}^K\sum_{j_k = 1}^{m_k} \lambda_{4, k, j_k}p(\Vec{e})\mathbbm{I}\{e_k = j\}\\
    \Rightarrow &p(\Vec{e})w(\Vec{e})\V(Y(\Vec{e})) = \lambda_1 p(\Vec{e})\mathbbm{I}\{e_1 = m_1\} + \sum_{m = 1}^{m_1 - 1}\lambda_{2,m} p(\Vec{e})\mathbbm{I}\{e_1 = m\} \\
    &+ \lambda_3 p(\Vec{e}) + \sum_{k = 2}^K\sum_{j_k = 1}^{m_k} \lambda_{4, k, j_k}p(\Vec{e})\mathbbm{I}\{e_k = j_k\},\\
\end{align*}
and hence
\begin{align}\label{mivlue_wt_eq}
    &w(\Vec{e}) = \\
    &\frac{\lambda_1 \mathbbm{I}\{e_1 = m_1\} + \sum_{m = 1}^{m_1 - 1}\lambda_{2,m} \mathbbm{I}\{e_1 = m\} + \lambda_3 + \sum_{k = 2}^K\sum_{j_k = 1}^{m_k} \lambda_{4, k, j_k}\mathbbm{I}\{e_k = j_k\}}{\V(Y(\Vec{e}))} \notag
\end{align}

The MIV LUE problem is equivalent to matrix problem described below.
Let matrix $\mathbf{P} =\begin{pmatrix}
    \mathbf{W} & \mathbf{C}^T \\
    \mathbf{C} & \mathbf{0}
   \end{pmatrix}$ 
be a block matrix of dimension $\left(|\mathcal{E}| + \sum_{k=1}^K m_k + 1 \right) \times \left(|\mathcal{E}| + \sum_{k=1}^K m_k + 1 \right)$, where $\mathbf{W}$ is a $|\mathcal{E}| \times |\mathcal{E}|$ diagonal matrix and the $j$th diagonal entry corresponding to $\Vec{e}_j$ for $j \in \{1, \dotsc, |\mathcal{E}|\}$ is 
\begin{align*}
\mathbf{W}_{j,j} = p(\Vec{e}_j)\V(Y(\Vec{e}_j))
\end{align*}
and $\mathbf{C}$ is a $|\Theta| \times |\mathcal{E}|$ matrix of constraints given by Proposition \ref{constraints} where the rows correspond to the parameters in $\Theta$ and the columns correspond to the exposures.
The entries of $\mathbf{C}$ are of the form $p(\Vec{e}_j)\mathbbm{I}\{\theta_{k,l}\in \Vec{e}_j\}$, where $j$ indexes in the column and $k,l$ indexes in the rows.
The notation $\theta_{k,j_k} \in \Vec{e}$ means that the parameter $\theta_{k,j_k}$ is a summand in the summation of parameters that equals to the potential outcome given exposure $\Vec{e}$.

Letting $\mathbf{b}$ be a $|\mathcal{E}| + \sum_{k=1}^K m_k + 1$-dimensional vector that is zero except for the $|\mathcal{E}|+1$ entry being one, the MIVLUE problem is equivalent to
\begin{align}\label{eq:full_matrix_eq}
    \begin{pmatrix}
        \mathbf{W} & \mathbf{C}^T \\
        \mathbf{C} & \mathbf{0}
    \end{pmatrix} \begin{pmatrix}
        \mathbf{w} \\ \lambda
    \end{pmatrix}  = \mathbf{b}.
\end{align}
Hence, the solution to the MIVLUE probem is $\mathbf{w} = \left(\mathbf{P}^{-1}\mathbf{b}\right)$.

\subsection{Characterizing MIV LUEs}\label{appendix:characterizing_mivlue}

\begin{proof}[Proof of Lemma \ref{var_cov_matrix_lemma}]
Let $\eta \in \mathbb{R}$ be such that $\eta \to \infty$ and $B \in \mathbb{R}^{|{\Theta}| \times |{\Theta}|}$ be a positive semi-definite matrix where elements $0<b_{k,j}<\infty$ are small. 
Now consider the variance-covariance matrix $\tilde{\mathbf{\Sigma} }_{\eta} = \eta \mathbf{\Sigma}  + B$. 
Then:
\begin{align}
    \lim_{\eta \to \infty} v_1^{T}\tilde{\mathbf{\Sigma} }_{\eta} v_1 &= \lim_{\eta \to \infty} v_1^{T}(\eta \mathbf{\Sigma}  + B)v_1 \nonumber \\
    &= \lim_{\eta \to \infty} \eta \underbrace{v_1^T \mathbf{\Sigma} v_1}_{=0} + v_1^TBv_1 \nonumber \\
    &= v_1^TBv_1 < \infty
\end{align}

\begin{align}
    \lim_{\eta \to \infty} v_2^T\tilde{\mathbf{\Sigma} }_{\eta} v_2 &= \lim_{\eta \to \infty} v_2^T(\eta \mathbf{\Sigma}  + B)v_2 \nonumber \\
    &= \lim_{\eta \to \infty} \eta \underbrace{v_2^T \mathbf{\Sigma}  v_2}_{=a} + v_2^TBv_2 \nonumber \\
    &= \lim_{\eta \to \infty} \underbrace{\eta}_{\to \infty}a + \underbrace{v_2^T B v_2}_{<\infty} \to \infty.
\end{align}
\end{proof}

\begin{Lemma}\label{lemma:p_matrix_full_rank}
Let the design $p$ be such that $p(\Vec{e}) > 0$ for all $\Vec{e} \in \mathcal{E}$ and let $\V(Y(\Vec{e})) > 0$ for all $\Vec{e} \in \mathcal{E}$. 
The matrix 
    $\mathbf{P} =\begin{pmatrix}
        \mathbf{W} & \mathbf{C}^T \\
        \mathbf{C} & \mathbf{0}
       \end{pmatrix}$
is full rank.
\end{Lemma}
\begin{proof}
First $\mathbf{W}$ is full rank since it is diagonal with positive diagonal entries.
Furthermore, $\mathbf{C}$ has full row rank because otherwise, the linear unbiased constraints given by Proposition \ref{constraints} are redundant. 
If the constraints are redundant, we can remove a constraint, but the constraints for unbiasedness are minimal.
That is, if we removed a constraint, there are linear estimators that satisfy the remaining constraints but are not unbiased. 
Hence, $\mathbf{C}$ must have full row rank in order to preserve unbiasedness. 

\noindent Suppose that $c = \begin{pmatrix} c_{\mathcal{E}} \\ c_{\Theta} \end{pmatrix}\in \mathbb{R}^{|\mathcal{E}| + |\Theta|}$ satisfies $c^T\mathbf{P}=0$ where $c_{\mathcal{E}}$ is a vector of length $|\mathcal{E}|$ and $c_{\Theta}$ is a vector of length $|\Theta|$.
Then:
\begin{align*}
    c^T\mathbf{P} = c^T \begin{pmatrix}
        \mathbf{W} & \mathbf{C}^T \\
        \mathbf{C} & \mathbf{0}
       \end{pmatrix} = \begin{pmatrix}
        c_{\mathcal{E}}^T \mathbf{W} + c_{\Theta}^T \mathbf{C} \\
        c_{\mathcal{E}}^T \mathbf{C}^T
       \end{pmatrix} = \begin{pmatrix}
           \mathbf{0}_{|\mathcal{E}| \times 1} \\
           \mathbf{0}_{|\Theta| \times 1}
       \end{pmatrix}.
\end{align*}
We then solve for $c_{\mathcal{E}}$ and $c_{\Theta}$:
\begin{align*}
    \mathbf{0}_{|\mathcal{E}| \times 1} &= c_{\mathcal{E}}^T \mathbf{W} + c_{\Theta}^T \mathbf{C} \implies c_{\mathcal{E}}^T = -c_{\Theta}^T \mathbf{C} \mathbf{W}^{-1}, \\
    \mathbf{0}_{|\Theta| \times 1} &=  c_{\mathcal{E}}^T \mathbf{C}^T =  -c_{\Theta}^T \mathbf{C} \mathbf{W}^{-1} \mathbf{C}^T \implies c_{\Theta}^T = \mathbf{0}_{|\Theta| \times 1},
\end{align*}
where in the first line, we can take the inverse of $\mathbf{W}$ since it is full rank and has positive diagonal entries, and in the second line, we multiply both sides by $(\mathbf{C} \mathbf{W}^{-1} \mathbf{C}^T)^{-1}$ where $\mathbf{C} \mathbf{W}^{-1} \mathbf{C}^T$ is full rank since $\mathbf{C}$ has full row rank and $\mathbf{W}$ is full rank and has positive diagonal entries.
Since $c_{\Theta}^T = \mathbf{0}_{|\Theta| \times 1}$, then $c_{\mathcal{E}}^T = \mathbf{0}_{|\mathcal{E}| \times 1}$.
Hence, all the rows in $\mathbf{P}$ are linearly independent, and since $\mathbf{P}$ is a square matrix, $\mathbf{P}$ is full rank. 
\end{proof}

\begin{proof}[Proof of Theorem \ref{mivlue_thm}]
The proof will proceed as follows.
We will first partition the set of parameters and exposures into five sets.
We will then show that if a specific one of these sets is empty, then $\mathrm{supp}(\hat{\theta}) \subseteq \mathcal{E}'$.
We will also show that $\mathrm{supp}(\hat{\theta}) \subseteq \mathcal{E}'$ holds if that set is non-empty.
Finally, we will argue that $\mathrm{supp}(\hat{\theta}) = \mathcal{E}'$ under the conditions specified.

First, recall the definition of $\mathbf{P}$ from Eq.~\eqref{eq:full_matrix_eq}.
Given $\mathcal{E'} \subseteq \mathcal{E}$, we arrange the exposures in $\mathbf{P}$, and the corresponding rows in $\mathbf{b}$ and $\mathbf{w}$, such that exposures $\Vec{e}_1, \dotsc, \Vec{e}_{|\mathcal{E}'|} \in \mathcal{E}' \subseteq \mathcal{E}$ and 
$\Vec{e}_{|\mathcal{E}'| + 1}, \dotsc, \Vec{e}_{|\mathcal{E}|} \in \mathcal{E} \setminus \mathcal{E'}$. 
For each $\Vec{e}\in \mathcal{E}$, let $\Vec{v}_{\Vec{e}}\in\{0,1\}^{|\Theta|}$ be such that $\Vec{v}_{\Vec{e}}^T \Vec{v} = Y(\Vec{e})$, where $\Vec{v}$ is the vector of parameters.
Suppose that $\mathrm{span}\left(\{\Vec{v}_{\Vec{e}{\,'}}\}_{\Vec{e}{\,'} \in \mathcal{E'}}\right) \cap \{\Vec{v}_{\Vec{e}}\}_{\Vec{e} \in \mathcal{E}} = \{\Vec{v}_{\Vec{e}{\,'}}\}_{\Vec{e}{\,'} \in \mathcal{E'}}$.
Note that $\mathrm{span}\left(\{\Vec{v}_{\Vec{e}^{\,'}}\}_{\Vec{e}^{\,'} \in \mathcal{E'}}\right)$ is a linear subspace of $\mathbb{R}^{|\Theta|}$.
Then there exists a positive semi-definite matrix $\mathbf{\Sigma}$ such that $\mathrm{span}\left(\{\Vec{v}_{\Vec{e}^{\,'}}\}_{\Vec{e}^{\,'} \in \mathcal{E'}}\right) = \mathrm{Null}(\mathbf{\Sigma})$.
In particular, $\mathbf{\Sigma} = I - XX^T$, where the columns of $X$ are vectors that form an orthonormal basis for $\mathrm{span}\left(\{\Vec{v}_{\Vec{e}^{\,'}}\}_{\Vec{e}^{\,'} \in \mathcal{E'}}\right)$.
By Lemma \ref{var_cov_matrix_lemma}, there then exists a sequence of variance-covariance matrices $\tilde{\mathbf{\Sigma}}_{\eta}$ such that $\lim_{\eta\to\infty}\Vec{v}_{\Vec{e}_j}^T \tilde{\mathbf{\Sigma}}_{\eta} \Vec{v}_{\Vec{e}_j} < \infty$ for $j \in \{1, \dotsc, |\mathcal{E}'|\}$ and $\lim_{\eta\to\infty}\Vec{v}_{\Vec{e}_{j'}}^T \tilde{\mathbf{\Sigma}}_{\eta} \Vec{v}_{\Vec{e}_{j'}} = \infty$ for $j' \in \{|\mathcal{E}'| + 1, \dotsc, \mathcal{E}\}$.

Let the matrix $\mathbf{P}_{\eta}$ be the same as $\mathbf{P}$, but with rows rearranged as described below and where the variances are given by variance-covariance matrix $\tilde{\mathbf{\Sigma}}_{\eta}$.
The MIV LUE problem then becomes $\mathbf{P}_{\eta}^{-1} \mathbf{b} =  \mathbf{w}_{\eta}$.
Let $\lim_{\eta \to \infty} \mathbf{P}_{\eta}^{-1} \mathbf{b} = \mathbf{w}^*$.
We want to show that $\mathrm{supp}(\mathbf{w}^*) = \mathcal{E}'$, and we first argue that $\mathrm{supp}(\mathbf{w}^*) \subseteq \mathcal{E}'$.
Note that $\mathrm{supp}(\mathbf{w}^*) = \mathrm{supp}(\hat{\theta})$ since $\hat{\theta}$ is given by $\mathbf{w}^*$.

We define submatrices of $\mathbf{P}_{\eta}$ as follows. 
First, the diagonal matrix $\mathbf{W}$ can be decomposed as follows.
Denote $\mathbf{F}_{\eta}$ as the diagonal matrix of probabilities and variances for the potential outcomes that have a finite limiting variance:
\begin{align}
    \mathbf{F}_{\eta} = 
    \begin{bmatrix}
        p(\Vec{e}_{1}) \V(Y(\Vec{e}_{1}))_{\eta}  & \dotsb & 0\\
        \dotsb &  \dotsb  & \dotsb\\
        0 & \dotsb & p(\Vec{e}_{|\mathcal{E}'|}) \V(Y(\Vec{e}_{|\mathcal{E}'|}))_{\eta} 
    \end{bmatrix}%
\end{align}
Similarly, denote $\mathbf{N}_{\eta}$ as the analogous matrix for potential outcomes that have a non-finite limiting variance:
\begin{align}
    \mathbf{N}_{\eta} = \begin{bmatrix}
    p(\Vec{e}_{|\mathcal{E}'| + 1}) \V(Y(\Vec{e}_{|\mathcal{E}'| + 1}))_{\eta}  & \dotsb & 0\\
    \dotsb &  \dotsb  & \dotsb\\
    0 & \dotsb & p(\Vec{e}_{|\mathcal{E}|}) \V(Y(\Vec{e}_{|\mathcal{E}|}))_{\eta} 
    \end{bmatrix}%
\end{align}
Next, denote the following subsets of parameters.
Let $\Theta^N \subseteq \Theta$ denote the set of parameters where $\theta^N \in \Theta^N$ are such that $\theta^N \notin \Vec{e}^{\,'}$ for all $\Vec{e}^{\,'} \in \mathcal{E}'$.
Here, we write $\theta \notin \Vec{e}$ to mean that $\theta$ is not in the sum of parameters in $Y(\Vec{e})$. 
Note that for $\theta^{N} \in \Theta^{N}$, we have $\lim_{\eta \to \infty} \V(\theta^{N}) = \infty$ by Lemma \ref{var_cov_matrix_lemma}.
In addition, we will further divide the parameters in $\Theta^F = \Theta \setminus \Theta^{N}$ as $\Theta^F = \Theta^R \cup \Theta^{NR}$.
Specifically, $\Theta^{NR}$ will be a maximal subset of $\Theta^F$ such that the submatrix of $\bvar{C}$ with rows given by $\Theta^{NR}$ and columns given by $\mathcal{E}'$ is linearly independent.

The matrix $\bvar{C} \in \mathbb{R}^{|\Theta| \times |\mathcal{E}|}$ contains submatrices corresponding to the linear unbiased constraints.
We denote the constraint matrices as $\bvar{C}_{e}^{p}$, where the subscript corresponds to the set of exposures $e$ and the superscript corresponds to the set of parameters $p$.
For each $e \in \{N, F\}$ and $p \in \{N, NR, R\}$, we define $\bvar{C}_{e}^{p}$ to be the submatrix of $\bvar{C}$ containing linear unbiased constraints in which rows correspond to parameters in $\Theta^p$ and columns correspond to the exposures in $\mathcal{E}^e$.
Here, $\mathcal{E}^F = \mathcal{E'}$ and $\mathcal{E}^N = \mathcal{E} \setminus \mathcal{E'}$.

\noindent For example ${\mathbf{C}_N^{N}}$ is defined as follows:
\begin{align}
    {\mathbf{C}_N^{N}} = \begin{bmatrix}
    p(\Vec{e}_{|\mathcal{E}'| + 1})\mathbb{I}\{\theta^{N}_1 \in \Vec{e}_{|\mathcal{E}'| + 1}\}  & \dotsb &  p(\Vec{e}_{|\mathcal{E}|})\mathbb{I}\{\theta^{N}_1 \in \Vec{e}_{|\mathcal{E}|}\}  \\
    \dotsb &  \dotsb  & \dotsb\\
     p(\Vec{e}_{|\mathcal{E}'| + 1})\mathbb{I}\{\theta^{N}_{|\Theta^{N}|} \in \Vec{e}_{|\mathcal{E}'| + 1}\} & \dotsb & p(\Vec{e}_{|\mathcal{E}|})\mathbb{I}\{\theta^{N}_{|\Theta^{N}|} \in \Vec{e}_{|\mathcal{E}|}\}
    \end{bmatrix}_{|\Theta^{N}| \times (|\mathcal{E}| - |\mathcal{E}'|)}.
\end{align}
Altogether, this results in the equation $\mathbf{P}_{\eta}^{-1}\mathbf{b} = \mathbf{w}_{\eta}$, equivalently
\begin{align}\label{pw_b_eq_w_cfr}
    \begin{pmatrix}
    \mathbf{N}_{\eta} & {\mathbf{C}_N^{N}}^T  & \mathbf{0} &  {\mathbf{C}_N^{NR}}^T & {\mathbf{C}_N^{R}}^T\\
    {\mathbf{C}_N^{N}} & \mathbf{0} &   {\mathbf{C}_F^{N}} & \mathbf{0} & \mathbf{0}\\
    \mathbf{0} & {\mathbf{C}_F^{N}}^T  & \mathbf{F}_{\eta} & {\mathbf{C}_F^{NR}}^T & {\mathbf{C}_F^{R}}^T \\
    {\mathbf{C}_N^{NR}} &  \mathbf{0} & {\mathbf{C}_F^{NR}} & \mathbf{0} & \mathbf{0} \\
    \mathbf{C}_N^R & \mathbf{0} & \mathbf{C}_F^R & \mathbf{0} & \mathbf{0}
    \end{pmatrix}^{-1}
    \begin{pmatrix}
         \mathbf{0}_{|\mathcal{E}| - |\mathcal{E}'| \times 1}  \\
         \mathbf{0}_{|\Theta^{N}| \times 1} \\
         \mathbf{0}_{|\mathcal{E}'| \times 1}\\
         1 \\
         \mathbf{0}_{|\Theta^{NR}| - 1 \times 1} \\
         \mathbf{0}_{|\Theta^{R}| \times 1} 
    \end{pmatrix}
    = 
    \begin{pmatrix}
         \mathbf{w}(\Vec{e}^{N})_{\eta}  \\
         {\lambda}^{N}_{\eta} \\
         \mathbf{w}(\Vec{e}^{F})_{\eta}\\
         {\lambda_{1}}_{\eta} \\
         {\lambda}^{NR}_{\eta} \setminus {\lambda_1}_{\eta} \\
         {\lambda^R}_{\eta}
    \end{pmatrix},
\end{align}
where the 1 in $\mathbf{b}$ and ${\lambda_{1}}_{\eta}$ in $\mathbf{w}_{\eta}$ corresponds to the constraint for $\theta_{1,m_1} \in \Theta^{NR}$.

\begin{case}[Assume $\mathbf{\Theta^R = \emptyset}$:]
We first consider the case when $\Theta^{R} = \emptyset$, so we can consider the linear equation
\begin{align}\label{pw_b_eq_wo_cfr}
    \left(
    \begin{array}{cccc}
    \mathbf{N}_{\eta} & {\mathbf{C}_N^{N}}^T  & \mathbf{0} &  {\mathbf{C}_N^{NR}}^T \\
    {\mathbf{C}_N^{N}} & \mathbf{0} &   {\mathbf{C}_F^{N}} & \mathbf{0}\\
    \mathbf{0} & {\mathbf{C}_F^{N}}^T  & \mathbf{F}_{\eta} & {\mathbf{C}_F^{NR}}^T \\
    {\mathbf{C}_N^{NR}} &  \mathbf{0} & {\mathbf{C}_F^{NR}} & \mathbf{0}
    \end{array}
    \right)^{-1}  
    \left(
    \begin{array}{c}
         \mathbf{0}_{|\mathcal{E}| - |\mathcal{E}'| \times 1}  \\
         \mathbf{0}_{|\Theta^{N}| \times 1} \\
         \mathbf{0}_{|\mathcal{E}'| \times 1}\\
         1 \\
         \mathbf{0}_{|\Theta^{NR}| - 1 \times 1}
    \end{array}
    \right)
     = 
    \left(
    \begin{array}{c}
         \mathbf{w}(\Vec{e}^{N})_{\eta}  \\
         {\lambda}^{N}_{\eta} \\
         \mathbf{w}(\Vec{e}^{F})_{\eta}\\
         {\lambda_{1}}_{\eta} \\
         {\lambda}^{NR}_{\eta} \setminus {\lambda_1}_{\eta}
    \end{array}
    \right).
\end{align}
Note that since we can rearrange rows of $\mathbf{P}_{\eta}$ such that $\mathbf{P}_{\eta}$ has the form $\begin{pmatrix}
    \mathbf{W}_{\eta} & \mathbf{C}^T \\
    \mathbf{C} & \mathbf{0}
\end{pmatrix}$, 
where $\mathbf{W}_{\eta} = \begin{pmatrix} 
\mathbf{N}_{\eta} & \mathbf{0} \\
\mathbf{0} & \mathbf{F}_{\eta}
\end{pmatrix}$ and $\mathbf{C} = \begin{pmatrix}
    \mathbf{C}_{N}^N & \mathbf{C}_F^N \\
    \mathbf{C}_N^{NR} & \mathbf{C}_F^{NR}
\end{pmatrix}$,
and since variances are given by $\tilde{\mathbf{\Sigma}}_{\eta}$, we have $\V(Y(\Vec{e})) > 0$ for all $\Vec{e} \in \mathcal{E}$.
Then, by Lemma \ref{lemma:p_matrix_full_rank}, $\mathbf{P}_{\eta}$ is full rank, and so $\mathbf{P}_{\eta}$ is invertible with solution $\mathbf{w}_{\eta} = \mathbf{P}^{-1}_{\eta}\mathbf{b}$.
Since $\mathbf{P}_{\eta} = \begin{bmatrix}
\mathbf{A}_{\eta} & \mathbf{B} \\
\mathbf{B}^T & \mathbf{D}_{\eta}
\end{bmatrix}$ 
is a block matrix, 
where 
\begin{align*}
    \mathbf{A}_{\eta} = \begin{pmatrix}
        \mathbf{N}_{\eta} & {\mathbf{C}_N^{N}}^T \\
        \mathbf{C}_N^N & \mathbf{0}
    \end{pmatrix},
    \mathbf{B} = \begin{pmatrix}
        \mathbf{0} &  {\mathbf{C}_N^{NR}}^T \\
        \mathbf{C}_F^N & \mathbf{0}
    \end{pmatrix}, \text{ and }
    \mathbf{D}_{\eta} = \begin{pmatrix}
        \mathbf{F}_{\eta} & {\mathbf{C}_F^{NR}}^T \\
        {\mathbf{C}_F^{NR}} & \mathbf{0}
    \end{pmatrix},
\end{align*}
we have $\mathbf{P}^{-1}_{\eta}=$
\begin{align}\label{schur_eq}
    \begin{bmatrix}
    \mathbf{A}^{-1}_{\eta} + \mathbf{A}^{-1}_{\eta}\mathbf{B}(\mathbf{D}_{\eta} - \mathbf{B}^T\mathbf{A}^{-1}_{\eta}\mathbf{B})^{-1}\mathbf{B}^T\mathbf{A}_{\eta}^{-1} & -\mathbf{A}^{-1}_{\eta}\mathbf{B}(\mathbf{D}_{\eta} - \mathbf{B}^T\mathbf{A}_{\eta}^{-1}\mathbf{B})^{-1} \\
    -(\mathbf{D}_{\eta} - \mathbf{B}^T\mathbf{A}_{\eta}^{-1}\mathbf{B})^{-1}\mathbf{B}^T\mathbf{A}_{\eta}^{-1} & (\mathbf{D}_{\eta} - \mathbf{B}^T\mathbf{A}_{\eta}^{-1}\mathbf{B})^{-1}
    \end{bmatrix}.
\end{align}
First, we want to determine the vector of weights $\mathbf{w}^*(\Vec{e}^{N})$, where $\Vec{e}^{N} \in \mathcal{E}^{N}$. 
We focus on the first $|\mathcal{E}| - |\mathcal{E}'|$ rows in $\mathbf{P}_{\eta}^{-1}$, i.e.
\begin{align}
    &\mathbf{w}^*(\Vec{e}^{N}) =\notag\\
    & \lim_{\eta \to \infty} \left[\mathbf{A}_{\eta}^{-1} + \mathbf{A}_{\eta}^{-1}\mathbf{B}(\mathbf{D}_{\eta} - \mathbf{B}^T\mathbf{A}_{\eta}^{-1}\mathbf{B})^{-1}\mathbf{B}^T\mathbf{A}_{\eta}^{-1}\right]_{\text{first $|\mathcal{E}| - |\mathcal{E}'|$ rows}}
    \begin{bmatrix}
    \mathbf{0}_{|\mathcal{E}| - |\mathcal{E}'| \times 1} \\
    \mathbf{0}_{|\Theta^{N}| \times 1}  
    \end{bmatrix} \nonumber \\
    &+ \lim_{\eta \to \infty} \left[- \mathbf{A}_{\eta}^{-1}\mathbf{B}(\mathbf{D}_{\eta} - \mathbf{B}^T\mathbf{A}_{\eta}^{-1}\mathbf{B})^{-1}\right]_{\text{first $|\mathcal{E}| - |\mathcal{E}'|$ rows}}
    \begin{bmatrix}
    \mathbf{0}_{|\mathcal{E}'| \times 1} \\
    1 \\
    \mathbf{0}_{|\Theta^{NR}| - 1\times 1} 
    \end{bmatrix} 
\end{align}
Since the first summand is multiplied by $\mathbf{0}_{|\mathcal{E}| - |\mathcal{E}'| + |\Theta^{N}| \times 1}$, we focus on the first $|\mathcal{E}| - |\mathcal{E}'|$ rows of the limit of $- \mathbf{A}_{\eta}^{-1}\mathbf{B}(\mathbf{D}_{\eta} - \mathbf{B}^T\mathbf{A}_{\eta}^{-1}\mathbf{B})^{-1}$.
Note that $\mathbf{A}_{\eta}$ is also a block matrix, so $\mathbf{A}_{\eta}^{-1}\begin{bmatrix} a & b \\ c & d \end{bmatrix}$ where
\begin{align*}
    a &= \mathbf{N}^{-1}_{\eta} + \mathbf{N}^{-1}_{\eta} {\mathbf{C}_N^{N}}^T(-  {\mathbf{C}_N^{N}}  \mathbf{N}^{-1}_{\eta} {\mathbf{C}_N^{N}}^T)^{-1} {\mathbf{C}_N^{N}} 
    \mathbf{N}_{\eta}^{-1} ,\\
    b &=
     -\mathbf{N}^{-1}_{\eta} {\mathbf{C}_N^{N}}^T  \left(- {\mathbf{C}_N^{N}} \mathbf{N}^{-1}_{\eta} {\mathbf{C}_N^{N}}^T\right)^{-1} ,\\
    c &=
     -\left(-{\mathbf{C}_N^{N}} \mathbf{N}^{-1}_{\eta} {\mathbf{C}_N^{N}}^T\right)^{-1} {\mathbf{C}_N^{N}} \mathbf{N}^{-1}_{\eta} ,\\
    d &=  
      \left(-{\mathbf{C}_N^{N}} \mathbf{N}^{-1}_{\eta} {\mathbf{C}_N^{N}}^T\right)^{-1}.
\end{align*}
Then, we have 
\begin{align*}
    &-\mathbf{A}_{\eta}^{-1}\mathbf{B}(\mathbf{D}_{\eta} - \mathbf{B}^T\mathbf{A}_{\eta}^{-1}\mathbf{B})^{-1} \\
    =& -\begin{bmatrix} a & b \\ c & d \end{bmatrix} 
    \begin{bmatrix} 
    \mathbf{0} & {\mathbf{C}_N^{NR}}^T \\
    {\mathbf{C}_{F}^{N}} & \mathbf{0}
    \end{bmatrix} \nonumber \\
    &\scalemath{.95}{\left(
    \begin{bmatrix}
    \mathbf{F}_{\eta} & {\mathbf{C}_F^{NR}}^T \\
    {\mathbf{C}_F^{NR}} & \mathbf{0} 
    \end{bmatrix} - 
    \begin{bmatrix}
    \mathbf{0} & {\mathbf{C}_F^{N}}^T \\
    {\mathbf{C}_{N}^{NR}} & \mathbf{0}
    \end{bmatrix}
    \begin{bmatrix}
    a & b \\
    c & d
    \end{bmatrix}
    \begin{bmatrix}
    \mathbf{0} & {\mathbf{C}_N^{NR}}^T \\
    {\mathbf{C}_{F}^{N}} & \mathbf{0}
    \end{bmatrix}
    \right)^{-1}} \nonumber \\
    =& -\begin{bmatrix} b {\mathbf{C}_{F}^{N}} & a {\mathbf{C}_N^{NR}}^T \\ 
    d {\mathbf{C}_{F}^{N}} & c {\mathbf{C}_N^{NR}}^T \end{bmatrix} \nonumber \\
    &\left(
    \begin{bmatrix}
    \mathbf{F}_{\eta} & {\mathbf{C}_F^{NR}}^T \\
    {\mathbf{C}_F^{NR}} & \mathbf{0} 
    \end{bmatrix} - 
    \begin{bmatrix}
    {\mathbf{C}_F^{N}}^T c & {\mathbf{C}_F^{N}}^T d\\
    {\mathbf{C}_{N}^{NR}} a & {\mathbf{C}_{N}^{NR}} b
    \end{bmatrix}
    \begin{bmatrix}
    \mathbf{0} & {\mathbf{C}_N^{NR}}^T \\
    {\mathbf{C}_{F}^{N}} & \mathbf{0}
    \end{bmatrix}
    \right)^{-1} \nonumber \\
    =& -\begin{bmatrix} b {\mathbf{C}_{F}^{N}} & a {\mathbf{C}_N^{NR}}^T \\ 
    d {\mathbf{C}_{F}^{N}} & c {\mathbf{C}_N^{NR}}^T \end{bmatrix} \nonumber \\
    &\times \begin{bmatrix}
    \mathbf{F}_{\eta} - {\mathbf{C}_F^{N}}^T d {\mathbf{C}_{F}^{N}} & 
    {\mathbf{C}_F^{NR}}^T  - {\mathbf{C}_F^{N}}^T c {\mathbf{C}_N^{NR}}^T\\
    {\mathbf{C}_F^{NR}} - {\mathbf{C}_{N}^{NR}} b {\mathbf{C}_F^{N}} & 
     - {\mathbf{C}_{N}^{NR}} a {\mathbf{C}_N^{NR}}^T
    \end{bmatrix}^{-1} \nonumber \\
    =& -\begin{bmatrix} \mathbf{0} & a {\mathbf{C}_N^{NR}}^T \\
    \mathbf{0} & c {\mathbf{C}_N^{NR}}^T \end{bmatrix} 
    \underbrace{
    \begin{bmatrix}
    \mathbf{F}_{\eta}  & 
    {\mathbf{C}_F^{NR}}^T\\
    {\mathbf{C}_F^{NR}} & 
     - {\mathbf{C}_{N}^{NR}} a {\mathbf{C}_N^{NR}}^T
    \end{bmatrix}^{-1}}_{ = \begin{bmatrix}
    a' & b' \\
    c' & d'
    \end{bmatrix}} \nonumber \\
    =& -\begin{bmatrix} a {\mathbf{C}_N^{NR}}^T c' & a {\mathbf{C}_N^{NR}}^T d'\\
    c {\mathbf{C}_N^{NR}}^T c' & c {\mathbf{C}_N^{NR}}^T d' \end{bmatrix}, 
\end{align*}
where $\mathbf{C}_F^{N} = \mathbf{0}$ since by definition of $\Theta^{N}$, for all $\theta^{N} \in \Theta^{N}$, we have $\theta^{N} \not \in \Vec{e}$ for $\Vec{e} \in \mathcal{E'}$.
We are interested in the first $|\mathcal{E}| - |\mathcal{E'}|$ rows, but since the first $|\mathcal{E}'|$ columns are multiplied by 0, we focus on the last $|\Theta^{NR}|$ columns:
\begin{align}\label{wN_equation}
    &-a {\mathbf{C}_N^{NR}}^T d'  \\
    =&
     \mathbf{N}^{-1}_{\eta} {\mathbf{C}_N^{NR}}^T \left({\mathbf{C}_{N}^{NR}} \mathbf{N}^{-1}_{\eta} {\mathbf{C}_N^{NR}}^T \right. \nonumber \\
    &- \left. {\mathbf{C}_{N}^{NR}} \mathbf{N}^{-1}_{\eta} {\mathbf{C}_N^{N}}^T( {\mathbf{C}_N^{N}}  \mathbf{N}^{-1}_{\eta} {\mathbf{C}_N^{N}}^T)^{-1} {\mathbf{C}_N^{N}}
    \mathbf{N}_{\eta}^{-1} {\mathbf{C}_N^{NR}}^T +  {\mathbf{C}_F^{NR}} \mathbf{F}_{\eta}^{-1} {\mathbf{C}_F^{NR}}^T \right)^{-1} \nonumber \\
    &- \mathbf{N}^{-1}_{\eta} {\mathbf{C}_N^{N}}^T(  {\mathbf{C}_N^{N}}  \mathbf{N}^{-1}_{\eta} {\mathbf{C}_N^{N}}^T)^{-1} {\mathbf{C}_N^{N}}
    \mathbf{N}_{\eta}^{-1} {\mathbf{C}_N^{NR}}^T
    \left({\mathbf{C}_{N}^{NR}} \mathbf{N}^{-1}_{\eta} {\mathbf{C}_N^{NR}}^T \right. \nonumber \\
    &- \left. {\mathbf{C}_{N}^{NR}} \mathbf{N}^{-1}_{\eta} {\mathbf{C}_N^{N}}^T( {\mathbf{C}_N^{N}}  \mathbf{N}^{-1}_{\eta} {\mathbf{C}_N^{N}}^T)^{-1} {\mathbf{C}_N^{N}} 
    \mathbf{N}_{\eta}^{-1} {\mathbf{C}_N^{NR}}^T +  {\mathbf{C}_F^{NR}} \mathbf{F}_{\eta}^{-1} {\mathbf{C}_F^{NR}}^T \right)^{-1} 
\end{align}

Note that $\mathbf{F}_{\eta}$ is full rank and by definition, $\mathbf{C}_F^{NR}$ is also full row rank. 
The rank of the product ${\mathbf{C}_F^{NR}} \mathbf{F}_{\eta}^{-1} {\mathbf{C}_F^{NR}}^T$ is equal to $\min\left(\mathrm{rank}\left(\mathbf{C}_F^{NR}\right), \mathrm{rank}\left(\mathbf{F}_{\eta}\right)\right)$.
The rank of $\mathbf{F}_{\eta}$ is $|\mathcal{E'}|$ and the rank of $\mathbf{C}_F^{NR}$ is $|\Theta^{NR}|$.
If $|\mathcal{E'}| < |\Theta^{NR}|$, then there are parameters such that they only appear in the same exposures, leading to linearly dependent constraints in $\mathbf{C}_F^{NR}$.
This contradicts the definition of $\mathbf{C}_F^{NR}$, so $|\mathcal{E'}| \geq |\Theta^{NR}|$. 
Hence, $\mathrm{rank}\left({\mathbf{C}_F^{NR}} \mathbf{F}_{\eta}^{-1} {\mathbf{C}_F^{NR}}^T\right) = |\Theta^{NR}|$.
So ${\mathbf{C}_F^{NR}} \mathbf{F}_{\eta}^{-1} {\mathbf{C}_F^{NR}}^T$ is full rank.

By the continuity of matrix inverse at full-rank matrices, we can exchange the limit and the inverse in Equation~\eqref{wN_equation}.
Note that we can write the $j$th diagonal entry of $\mathbf{N}_{\eta}$ as $p(\Vec{e}_{j}) \left(\eta a_{\Vec{e}_j} + \Vec{v}_{\Vec{e}_j}^T B \Vec{v}_{\Vec{e}_j} \right)$.
Let $\tilde{\mathbf{N}}_{\eta}$ be the matrix with diagonal entries $p(\Vec{e}_j) \left(a_{\Vec{e}_j} + \frac{1}{\eta} \Vec{v}_{\Vec{e}_j}^T B \Vec{v}_{\Vec{e}_j} \right)$ so that $\mathbf{N}_{\eta} = \eta \tilde{\mathbf{N}}_{\eta}$, i.e. $\mathbf{N}_{\eta}^{-1} = \frac{1}{\eta} \tilde{\mathbf{N}}_{\eta}^{-1}$.
As $\lim_{\eta \to \infty} \left(\tilde{\mathbf{N}}^{-1}_{\eta}\right)_{{j,j}} = p(\Vec{e}_{j})^{-1} a_{\Vec{e}_j}^{-1} < \infty$, 
\begin{align} \label{inverse_terms_lim_0}
    & \lim_{\eta \to \infty} \left({\mathbf{C}_{N}^{NR}}  \mathbf{N}^{-1}_{\eta} {\mathbf{C}_{N}^{NR}}^T  - {\mathbf{C}_{N}^{NR}} \mathbf{N}^{-1}_{\eta} {\mathbf{C}_N^{N}}^T ({\mathbf{C}_N^{N}}  \mathbf{N}^{-1}_{\eta} {\mathbf{C}_N^{N}}^T)^{-1} {\mathbf{C}_N^{N}}
    \mathbf{N}_{\eta}^{-1} {\mathbf{C}_{N}^{NR}}^T \right)\nonumber\\
    &= \lim_{\eta \to \infty} {\mathbf{C}_{N}^{NR}}  \frac{1}{\eta}\tilde{\mathbf{N}}^{-1}_{\eta} {\mathbf{C}_{N}^{NR}}^T  \nonumber \\
    &- \lim_{\eta \to \infty} {\mathbf{C}_{N}^{NR}} \frac{1}{\eta}\tilde{\mathbf{N}}^{-1}_{\eta} {\mathbf{C}_N^{N}}^T ({\mathbf{C}_N^{N}}  \frac{1}{\eta} \tilde{\mathbf{N}}^{-1}_{\eta} {\mathbf{C}_N^{N}}^T)^{-1} {\mathbf{C}_N^{N}}
    \frac{1}{\eta} \tilde{\mathbf{N}}_{\eta}^{-1} {\mathbf{C}_{N}^{NR}}^T \nonumber\\
    &= \mathbf{0}_{|\Theta^{NR}| \times |\Theta^{NR}|}.
\end{align}
Hence, 
\begin{align}
   &\lim_{\eta \to \infty} a {\mathbf{C}_N^{NR}}^T d' \\
   =& \lim_{\eta \to \infty} -\frac{1}{\eta} \tilde{\mathbf{N}}_{\eta}^{-1} {\mathbf{C}_N^{NR}}^T \left({\mathbf{C}_F^{NR}} \mathbf{F}_{\eta}^{-1} {\mathbf{C}_F^{NR}}^T\right)^{-1} \nonumber \\
   &+ \lim_{\eta \to \infty} \frac{1}{\eta}\tilde{\mathbf{N}}^{-1}_{\eta} {\mathbf{C}_N^{N}}^T ({\mathbf{C}_N^{N}}  \tilde{\mathbf{N}}^{-1}_{\eta} {\mathbf{C}_N^{N}}^T)^{-1} {\mathbf{C}_N^{N}}
    \tilde{\mathbf{N}}_{\eta}^{-1} {\mathbf{C}_N^{NR}}^T \left({\mathbf{C}_F^{NR}} \mathbf{F}_{\eta}^{-1} {\mathbf{C}_F^{NR}}^T\right)^{-1} \label{eq:lim_nonfin}.
\end{align}
Since each matrix above is bounded in $\eta$, 
$\lim_{\eta \to \infty} a {\mathbf{C}_{N}^{NR}}^T d' = \mathbf{0}_{|\mathcal{E}| - |\mathcal{E'}| \times |\Theta^{NR}|}$.
Hence, $\mathbf{w}^*(\Vec{e}^{N}) = \mathbf{0}_{|\mathcal{E}| - |\mathcal{E'}| \times 1}$.
This establishes that $\mathrm{supp}(\mathbf{w}^*)\subset \mathcal{E}'$ if $\Theta^R = \emptyset$.
\end{case}

\begin{case}[Assume $\mathbf{\Theta^R \neq \emptyset}$:]
Now we consider the case when $\Theta^R \neq \emptyset$, i.e. $\mathbf{C}_F^{R} \neq \mathbf{0}$.
The matrix equation becomes
\begin{align}
     \begin{pmatrix}
    \mathbf{P}_{\eta} & 
     \begin{pmatrix}
     {\mathbf{C}_N^R}^T \\ 
     \mathbf{0}_{|\Theta^N| \times |\Theta^R|} \\
     {\mathbf{C}_F^R}^T \\
     \mathbf{0}_{|\Theta^{NR}| \times |\Theta^R|}
     \end{pmatrix} \\
     \begin{pmatrix}
     \mathbf{C}_N^R & 
     \mathbf{0}_{|\Theta^R| \times |\Theta^N|} & 
     \mathbf{C}_F^R & 
     \mathbf{0}_{|\Theta^R| \times |\Theta^{NR}|}
     \end{pmatrix} & 
     \mathbf{0}_{|\Theta^R| \times |\Theta^R|}
    \end{pmatrix} &
    \begin{pmatrix}
         \mathbf{w}_{\eta} \\
         {\lambda^R}_{\eta}
    \end{pmatrix}\notag \\
    =&
    \begin{pmatrix}
         \mathbf{b} \\
         \mathbf{0}_{|\Theta^{R}| \times 1} 
    \end{pmatrix},\label{eq:full-with-r}
\end{align}
where $\mathbf{P}_{\eta}, \mathbf{w}_{\eta}$, and $\mathbf{b}$ are matrices and vectors from Equation~\eqref{pw_b_eq_wo_cfr}.
Denote Equation~\eqref{eq:full-with-r} as $\tilde{\mathbf{P}}_{\eta} \tilde{\mathbf{w}}_{\eta} = \tilde{\mathbf{b}}$.
Since we already showed that $\mathbf{w}^*$ is the solution to the matrix equation in the limit when $\Theta^R = \emptyset$, we have the following:
\begin{align} 
      \tilde{\mathbf{P}}^*
    \begin{pmatrix}
        \mathbf{w}^* \\
        \mathbf{0}_{|\Theta^R| \times 1} 
    \end{pmatrix} = 
    \begin{pmatrix}
        \mathbf{0}_{|\mathcal{E} \setminus \mathcal{E'}| \times 1}  \\
        \mathbf{0}_{|\Theta^N| \times 1} \\
        \mathbf{0}_{|\mathcal{E'}| \times 1}  \\
         1 \\
         \mathbf{0}_{|\Theta^{NR}|-1 \times 1} \\
         \mathbf{C}_N^R \mathbf{w}^*(\Vec{e}^{\,N}) + \mathbf{C}_F^R \mathbf{w}^*(\Vec{e}^{\,F})
    \end{pmatrix},
\end{align}
where $\mathbf{w}^*(\Vec{e}^{\,N})$ and $\mathbf{w}^*(\Vec{e}^{\,F})$ are the weights of the exposures in $\mathcal{E} \setminus \mathcal{E'}$ and $\mathcal{E'}$, respectively as given by $\mathbf{w}^*$ and $\tilde{\mathbf{P}}^* = \lim_{\eta \to \infty} \tilde{\mathbf{P}}_{\eta}$.
Recall that in the previous case, we showed that $\mathbf{w}^*(\Vec{e}^{\,N}) = \mathbf{0}_{|\mathcal{E}| - \mathcal{E'}| \times 1}$.
Furthermore, recall that by construction of $\mathbf{C}_F^{NR}$, $\mathbf{C}_F^R = \mathbf{T} \mathbf{C}_F^{NR}$, where the first column of $\mathbf{T}$ only contains zeros since $\theta_{1,m_1}$ cannot be linearly dependent with another parameter.
Otherwise, unbiasedness does not hold.
Since $\mathbf{w}^*(\Vec{e}^{\,F})$ solves the matrix equation given by Equation~\eqref{pw_b_eq_wo_cfr}, then $\mathbf{C}_F^R \mathbf{w}^*(\Vec{e}^{\,F}) = \mathbf{0}_{|\Theta^R| \times 1}$.
Then, $\tilde{\mathbf{P}}^* \tilde{\mathbf{w}}^* = \tilde{\mathbf{b}}$, where $\tilde{\mathbf{w}}^* = \begin{pmatrix}
\mathbf{w}^* \\ \mathbf{0}_{|\Theta^R| \times 1}
\end{pmatrix}$.

Let $\mathbf{w}$ be the true solution in the limit to the problem $\tilde{\mathbf{P}}_{\eta} \tilde{\mathbf{w}}_{\eta} = \tilde{\mathbf{b}}$ as given by Equation~\eqref{eq:full-with-r}.
Then: 
\begin{align}
    \tilde{\mathbf{P}}^* {\mathbf{w}} - \tilde{\mathbf{P}}^* \tilde{\mathbf{w}}^* = \tilde{\mathbf{P}}^* \left({\mathbf{w}} - \tilde{\mathbf{w}}^*\right)  = \mathbf{0}_{|\mathcal{E}| + |\Theta|}. 
\end{align}
By Lemma \ref{lemma:p_matrix_full_rank}, $\tilde{\mathbf{P}}^*$ is full rank, where now $\mathbf{C} = \begin{pmatrix}
    \mathbf{C}_{N}^N & \mathbf{C}_F^N \\
    \mathbf{C}_N^R & \mathbf{C}_F^R \\
    \mathbf{C}_N^{NR} & \mathbf{C}_F^{NR}
\end{pmatrix}$. 
We can then multiply both sides by ${\tilde{\mathbf{P}}}^{*^{-1}}$, and since all elements in ${\tilde{\mathbf{P}}}^{*^{-1}}$ are finite, we have:
\begin{align}
    {\mathbf{w}} - \tilde{\mathbf{w}}^* = \mathbf{0}_{|\mathcal{E}| + |\Theta|}.
\end{align}
Thus, in the limit, the solutions $\tilde{\mathbf{w}}^*$ and ${\mathbf{w}}$ are the same, and we see that ${\tilde{\mathbf{w}}}^*(\Vec{e}^{\,N}) = \mathbf{0}_{|\mathcal{E}| - |\mathcal{E'}|}$ and ${{\tilde{\mathbf{w}}}}^{*}(\Vec{e}^{\,F})$ depends on $\mathbf{w}^*(\Vec{e}^{\,F})$.

Hence, if there exists $\mathcal{E'} \subseteq \mathcal{E}$ such that $\mathrm{span}\left(\{\Vec{v}_{\Vec{e}^{\,'}}\}_{\Vec{e}^{\,'} \in \mathcal{E'}}\right) \cap \{\Vec{v}_{\Vec{e}}\}_{\Vec{e} \in \mathcal{E}} = \{\Vec{v}_{\Vec{e}^{\,'}}\}_{\Vec{e}^{\,'} \in \mathcal{E'}}$, then there exists a $\hat{\theta}$ with $\mathrm{supp}(\hat{\theta}) \subseteq \mathcal{E'}$ and $\hat{\theta}$ is a limit of MIV LUEs. 
\end{case} 

\noindent \underline{\textbf{Showing $\mathrm{supp}(\mathbf{w}^*) = \mathcal{E'}$:}} \newline
Finally, we want to determine the vector of weights $\mathbf{w}^*(\Vec{e}^{\,F})$, where $\Vec{e}^{\,F} \in \mathcal{E}'$, which is given by:

\begin{align}
    \mathbf{w}^*(\Vec{e}^{\,F}) &= \lim_{\eta \to \infty} \left[\left(\mathbf{D}_{\eta} - \mathbf{B}^T\mathbf{A}_{\eta}^{-1}\mathbf{B}\right)^{-1}\mathbf{B}^T\mathbf{A}_{\eta}^{-1}\right]_{\text{first $|\mathcal{E}'|$ rows}}\begin{bmatrix}
    \mathbf{0}_{|\mathcal{E}| - |\mathcal{E}'| \times 1} \\
    \mathbf{0}_{|\Theta^{N}| \times 1} \\
    \end{bmatrix} \nonumber \\
    &+ \lim_{\eta \to \infty} \left(\mathbf{D}_{\eta} - \mathbf{B}^T\mathbf{A}_{\eta}^{-1}\mathbf{B}\right)^{-1}_{\text{first $|\mathcal{E}'|$ rows}}
    \begin{bmatrix}
    \mathbf{0}_{|\mathcal{E}'| \times 1} \\
    1 \\
    \mathbf{0}_{|\Theta^{NR}| - 1\times 1} 
    \end{bmatrix}.
\end{align}
We focus on when $\Theta^R = \emptyset$ since we have shown that the weights for $\Vec{e}^{\,F}$ when $\Theta^R \neq \emptyset$ are the same in the limit as the weights when $\Theta^R = \emptyset$.
Here, we focus on the first $|\mathcal{E}'|$ rows and last $|\Theta^{NR}|$ columns of $\lim_{\eta \to \infty} \left(\mathbf{D}_{\eta} - \mathbf{B}^T\mathbf{A}_{\eta}^{-1}\mathbf{B}\right)^{-1}$.
Note that $\mathbf{D}_{\eta} - \mathbf{B}^T\mathbf{A}_{\eta}^{-1}\mathbf{B}$ is a block matrix, so we are interested in the upper right block of the inverse.
Again, we denote ${\mathbf{A}^{-1}_{\eta}} = \begin{bmatrix} a & b \\ c & d \end{bmatrix}$.
Then, using the right hand side of Equation (\ref{schur_eq}), the upper right block of $\left(\mathbf{D}_{\eta} - \mathbf{B}^T\mathbf{A}_{\eta}^{-1}\mathbf{B}\right)^{-1}$ is:
\begin{align}\label{w_F_weights}
    &\left(\mathbf{D}_{\eta} - \mathbf{B}^T\mathbf{A}_{\eta}^{-1} \mathbf{B}\right)^{-1}_{\text{upper right block}} \\
    =& 
    \mathbf{F}_{\eta}^{-1}
    {\mathbf{C}_{F}^{NR}}^T
    \left( {\mathbf{C}_{N}^{NR}} a {\mathbf{C}_{N}^{NR}}^T + 
    {\mathbf{C}_{F}^{NR}} \mathbf{F}_{\eta}^{-1}
    {\mathbf{C}_F^{NR}}^T
    \right)^{-1} \nonumber \\
    =& 
    \mathbf{F}_{\eta}^{-1}
    {\mathbf{C}_{F}^{NR}}^T
    \left( {\mathbf{C}_{N}^{NR}} \mathbf{N}^{-1}_{\eta} 
    {\mathbf{C}_{N}^{NR}}^T \right. \nonumber\\
    &- \left. {\mathbf{C}_{N}^{NR}} \mathbf{N}^{-1}_{\eta} {\mathbf{C}_N^{N}}^T ({\mathbf{C}_N^{N}} \mathbf{N}^{-1}_{\eta} {\mathbf{C}_N^{N}}^T)^{-1} {\mathbf{C}_N^{N}} 
    \mathbf{N}_{\eta}^{-1}  {\mathbf{C}_{N}^{NR}}^T \right. \nonumber \\
    &+ \left.
    {\mathbf{C}_{F}^{NR}} \mathbf{F}_{\eta}^{-1}
    {\mathbf{C}_F^{NR}}^T
    \right)^{-1}.
\end{align}
Since ${\mathbf{C}_{F}^{NR}} \mathbf{F}_{\eta}^{-1} {\mathbf{C}_F^{NR}}^T$ is full rank as shown previously, we can take the limit inside the inverse. 
From Equation (\ref{inverse_terms_lim_0}), we have 
\begin{align*}
    \lim_{\eta \to \infty} \left({\mathbf{C}_{N}^{NR}}  \mathbf{N}^{-1}_{\eta} {\mathbf{C}_{N}^{NR}}^T  - {\mathbf{C}_{N}^{NR}} \mathbf{N}^{-1}_{\eta} {\mathbf{C}_N^{N}}^T ({\mathbf{C}_N^{N}}  \mathbf{N}^{-1}_{\eta} {\mathbf{C}_N^{N}}^T)^{-1} {\mathbf{C}_N^{N}}
    \mathbf{N}_{\eta}^{-1} {\mathbf{C}_{N}^{NR}}^T\right) 
\end{align*}
is zero.
Hence,
\begin{align}
    \mathbf{w}^*(\Vec{e}^{F}) = \lim_{\eta \to \infty} \mathbf{F}_{\eta}^{-1} {\mathbf{C}_F^{NR}}^T\left({\mathbf{C}_F^{NR}} \mathbf{F}_{\eta}^{-1} {\mathbf{C}_F^{NR}}^T \right)^{-1}
    \begin{bmatrix}
    1 \\
    \mathbf{0}_{|\Theta^{NR}| - 1\times 1} 
    \end{bmatrix}.
\end{align}

Note that the $(k,l)$th entry of ${\mathbf{C}_F^{NR}} \mathbf{F}_{\eta}^{-1} {\mathbf{C}_F^{NR}}^T$ is given by:
\begin{align}
    \left({\mathbf{C}_F^{NR}} \mathbf{F}_{\eta}^{-1} {\mathbf{C}_F^{NR}}^T\right)_{k,l} = \sum_{j = 1}^{|\mathcal{E}'|} \mathbb{I}\{\theta_{k}, \theta_{l} \in \Vec{e}_j\}\frac{p(\Vec{e}_j)}{\V(Y(\Vec{e}_j))},
\end{align}
where the $k,l \in \{1, \dotsc, |\Theta^{NR}|\}$ indexes the different parameters in $\Theta^{NR}$.
If for every exposure $\Vec{e}_j \in \mathcal{E'}$, we have
\begin{align}
    \lim_{\eta \to \infty} \sum_{k = 1}^{|\Theta^{NR}|} Adj\left({\mathbf{C}_F^{NR}} \mathbf{F}_{\eta}^{-1} {\mathbf{C}_F^{NR}}^T\right)_{k,1} \mathbb{I}\{\theta_{k} \in \Vec{e}_j\} \neq 0,
\end{align} 
where $Adj\left({\mathbf{C}_F^{NR}} \mathbf{F}_{\eta}^{-1} {\mathbf{C}_F^{NR}}^T\right)_{k,1}$ is the $(k,1)th$ entry of the adjugate matrix of \newline 
${\mathbf{C}_F^{NR}} \mathbf{F}_{\eta}^{-1} {\mathbf{C}_F^{NR}}^T$ corresponding of $\theta_{1,m_1}$, then $\mathbf{w}^*(\Vec{e}) \neq 0$.
Hence, $\mathrm{supp}(\mathbf{w}^*) = \mathcal{E'}$.
\end{proof}

\subsection{Example: Derivation of Weights for Six-Term Exposure}\label{appendix:six_term_ex}

We show that in general $\mathbf{w}^*(\Vec{e}^{F}) \neq 0$ through an example.
Consider 
\begin{align*}
    \mathcal{E}^{\text{six term}, m} = \{(0,0), (0,j), (m,0), (m,j), (m_1, 0), (m_1,j)\},
\end{align*}
where $m \in \{1, \dotsc, m_1-1\}, j \in \{1, \dotsc, m_2\}$ and consider a prior covariance-matrix $\mathbf{\Sigma}$, where all prior variances of parameters are finite.
Denote entries of the inverse of ${\mathbf{C}_F^{NR}} \mathbf{F}_{\eta}^{-1} {\mathbf{C}_F^{NR}}^T$ as $\tilde{a}_{k,j} =  \frac{1}{ det\left({\mathbf{C}_F^{NR}} \mathbf{F}_{\eta}^{-1} {\mathbf{C}_F^{NR}}^T\right)} Adj\left({\mathbf{C}_F^{NR}} \mathbf{F}_{\eta}^{-1} {\mathbf{C}_F^{NR}}^T\right)_{k,j}$ where $Adj$ is the adjugate.
Then, 
\begin{align*}
    \mathbf{w}^*(\Vec{e}^{\,F}) &= %
    \begin{bmatrix}
    \frac{\tilde{a}_{2,1}}{\V(Y(0,0))} & \frac{\tilde{a}_{2,2}}{\V(Y(0,0))} & \frac{\tilde{a}_{2,3}}{\V(Y(0,0))} & \frac{\tilde{a}_{2,4}}{\V(Y(0,0))} \\
    \frac{\tilde{a}_{2,1} + \tilde{a}_{4,1}}{\V(Y(0,j))}  & \frac{\tilde{a}_{2,2} + \tilde{a}_{4,2}}{\V(Y(0,j))} & \frac{\tilde{a}_{2,3} + \tilde{a}_{4,3}}{\V(Y(0,j))} & \frac{\tilde{a}_{2,4} + \tilde{a}_{4,4}}{\V(Y(0,j))} \\
    \frac{\tilde{a}_{2,1} + \tilde{a}_{3,1}}{\V(Y(m,0))} & \frac{\tilde{a}_{2,2} + \tilde{a}_{3,2}}{\V(Y(m,0))} & 
    \frac{\tilde{a}_{2,3} + \tilde{a}_{3,3}}{\V(Y(m,0))} & \frac{\tilde{a}_{2,4} + \tilde{a}_{3,4}}{\V(Y(m,0))} \\
    \frac{\tilde{a}_{2,1} + \tilde{a}_{3,1} + \tilde{a}_{4,1}}{\V(Y(m,j))} & \frac{\tilde{a}_{2,2} + \tilde{a}_{3,2} + \tilde{a}_{4,2}}{\V(Y(m,j))} & 
    \frac{\tilde{a}_{2,3} + \tilde{a}_{3,3} + \tilde{a}_{4,3}}{\V(Y(m,j))} & \frac{\tilde{a}_{2,4} + \tilde{a}_{3,4} + \tilde{a}_{4,4}}{\V(Y(m,j))} \\
    \frac{\tilde{a}_{1,1} + \tilde{a}_{2,1}}{\V(Y(m_1,0))} & \frac{\tilde{a}_{1,2} + \tilde{a}_{2,2} }{\V(Y(m_1,0))} & 
    \frac{\tilde{a}_{1,3} + \tilde{a}_{2,3} }{\V(Y(m_1,0))} & \frac{\tilde{a}_{1,4} + \tilde{a}_{2,4} }{\V(Y(m_1,0))} \\
    \frac{\tilde{a}_{1,1} + \tilde{a}_{2,1} + \tilde{a}_{4,1}}{\V(Y(m_1,j))} & \frac{\tilde{a}_{1,2} + \tilde{a}_{2,2} + \tilde{a}_{4,2}}{\V(Y(m_1,j))} & 
    \frac{\tilde{a}_{1,3} + \tilde{a}_{2,3} + \tilde{a}_{4,3}}{\V(Y(m_1,j))} & \frac{\tilde{a}_{1,4} + \tilde{a}_{2,4} + \tilde{a}_{4,4}}{\V(Y(m_1,j))} \\
    \end{bmatrix} .
\end{align*}
The weights $\mathbf{w}^*(\Vec{e}^{\,F})$ are given by the entries in the first column, and so $\mathbf{w}^*(\Vec{e}^{\,F})$ is non-zero if the corresponding entries of the inverse of ${\mathbf{C}_F^{NR}} \mathbf{F}_{\eta}^{-1} {\mathbf{C}_F^{NR}}^T$ are non-zero.
We focus on $\mathbf{w}^*(0,0) = \frac{\tilde{a}_{2,1}}{\V(Y(0,0))}$.
Here $\tilde{a}_{2,1} = \frac{1}{ det\left({\mathbf{C}_F^{NR}} \mathbf{F}_{\eta}^{-1} {\mathbf{C}_F^{NR}}^T\right)} Adj\left({\mathbf{C}_F^{NR}} \mathbf{F}_{\eta}^{-1} {\mathbf{C}_F^{NR}}^T\right)_{2,1}$.
Since ${\mathbf{C}_F^{NR}} \mathbf{F}_{\eta}^{-1} {\mathbf{C}_F^{NR}}^T$ is full rank, the determinant is non-zero, and so we focus on the adjugate term in terms of the minor, denoted by $\mathbf{M}_{i,j}$:
\begin{align}\label{six_term_estimator_ex}
    Adj&\left({\mathbf{C}_F^{NR}} \mathbf{F}_{\eta}^{-1} {\mathbf{C}_F^{NR}}^T\right)_{2,1} = -\mathbf{M}_{2,1} \nonumber \\
    =& \left(\frac{p(m_1,0) p(m,0) p(0,j)}{\V(Y(m_1,0)) \V(Y(m,0)) \V(Y(0,j))} \right. \\
    &+ \frac{p(m_1,0) p(m,0) p(m,j)}{\V(Y(m_1,0)) \V(Y(m,0)) \V(Y(m,j))}  \nonumber \\
    &+ \frac{p(m_1,0) p(m,0) p(m_1,j)}{\V(Y(m_1,0)) \V(Y(m,0)) \V(Y(m_1,j))}\\
    &+ \frac{p(m_1,0) p(m,j) p(0,j)}{\V(Y(m_1,0)) \V(Y(m,j)) \V(Y(0,j))}\\
    &+ \frac{p(m_1,0) p(m,j) p(m_1,j)}{\V(Y(m_1,0)) \V(Y(m,j)) \V(Y(m_1,j))} \\
    &\left. + \frac{p(m_1,j) p(m,0) p(m,j)}{\V(Y(m_1,j)) \V(Y(m,0)) \V(Y(m,j))} \right).
\end{align}
Thus, we would need to set at least two probabilities of exposures to be zero in order for $\mathbf{w}^*(0,0) = 0$.
This holds similarly for other parameters.
Hence, for typical choices of the design probabilities and for priors where all variances are finite, $w(\Vec{e}^{\,F}) \neq 0$, i.e. $\mathrm{supp}(\hat{\theta}) = \mathcal{E}^{\text{six term}, m}$.

\section{Example: Six-Term Exposure Set}\label{appendix:six_term_exposure}

\begin{proof}[Proof of Corollary \ref{cor:max_of_a3}]
Consider the exposure set 
\begin{align*}
    \mathcal{E}^{\text{six term}, m} = \{(0,0), (0,j), (m_1,0), (m_1,j), (m,0), (m,j)\},
\end{align*} 
where $j \in \{1, \dotsc, m_2\}$ and $m \in \{1, \dotsc, m-1\}$.
Note that $\{\Vec{v}_{\Vec{e}^{\,'}}\}_{\Vec{e}^{\,'} \in \mathcal{E}^{\text{six term}, m}} = \mathrm{span}\left(\{\Vec{v}_{\Vec{e}^{\,'}}\}_{\Vec{e}^{\,'} \in \mathcal{E}^{\text{six term}, m}}\right) \cap \{\Vec{v}_{\Vec{e}}\}_{\Vec{e} \in \mathcal{E}}$.
By Theorem \ref{mivlue_thm}, there exists an estimator $\hat{\theta}$ such that it is a MIV LUE, for a given prior variance-covariance matrix, and $\mathrm{supp}(\hat{\theta}) = \mathcal{E}^{\text{six term}, m}$. 
Since $\hat{\theta}$ is a LUE, there are weights $\alpha_1, \alpha_2, \alpha_3 \in \mathbb{R}$ such that 
\begin{align}
    \hat{\theta} &= \alpha_1 \left(HT_{(m_1,0)} - HT_{(0,0)}\right) + \alpha_2 \left(HT_{(m_1,j)} - HT_{(0,j)}\right) \nonumber \\
    &+ \alpha_3 \left(HT_{(m_1,j)} - HT_{(m,j)} + HT_{(m,0)} + HT_{(0,0)}\right), 
\end{align}
where the three ALUEs form a basis for six-term estimators.
We know that $\alpha_1$ and $\alpha_2$ can equal 1 since the two two-term estimators are also MIV LUEs, but $\alpha_3 \neq 0$ because the four-term estimator is not a MIV LUE. 
However, exposures in the support for the four-term estimator can still contribute to MIV LUEs.
We investigate this contribution by finding the maximum of the weight $\alpha_3$.

First, we want to solve for the weights of the exposures in $\mathcal{E}^{\text{six-term}, m}$. 
From the proof of Theorem \ref{mivlue_thm}, we know that the weights are given by 
\begin{align}
    w(0,0) &= \frac{\tilde{a}_{2,1}}{\V(Y(0,0))} \\ 
    w(0,j) &= \frac{\tilde{a}_{2,1} + \tilde{a}_{4,1}}{\V(Y(0,j))}  \\
    w(m,0) &= \frac{\tilde{a}_{2,1} + \tilde{a}_{3,1}}{\V(Y(m,0))} \\
    w(m,j) &= \frac{\tilde{a}_{2,1} + \tilde{a}_{3,1} + \tilde{a}_{4,1}}{\V(Y(m,j))} \\
    w(m_1,0) &= \frac{\tilde{a}_{1,1} + \tilde{a}_{2,1}}{\V(Y(m_1,0))}  \\
    w(m_1,j) &= \frac{\tilde{a}_{1,1} + \tilde{a}_{2,1} + \tilde{a}_{4,1}}{\V(Y(m_1,j))},
\end{align}
where the terms $\tilde{a}_{i,j}$ are the limit of terms in the adjugate matrix divided by the determinant of $\mathbf{C}_F^{NR} \mathbf{F}_{\eta}^{-1} {\mathbf{C}_F^{NR}}^T$ and the potential outcome variances are given by the prior variance-covariance matrix.
Suppose the variances of the potential outcomes are all finite.
We first compute the determinant. 
We write $r(\Vec{e}) = \frac{p(\Vec{e})}{\V(Y(\Vec{e}))}$:
\begin{align*}
    &det\left(\mathbf{C}_F^{NR} \mathbf{F}_{\eta}^{-1} {\mathbf{C}_F^{NR}}^T\right)\\
    =& \left[r(m_1,0) + r(m_1,j)\right] \\
    &\times \bigg\{\bigg[r(0,0) + r(0,j) + r(m_1,0) + r(m_1,j) + r(m,0) + r(m,j)\bigg]  \\
    & \times  \bigg[r(m,0) + r(m,j)\bigg]
    \bigg[r(0,j) + r(m,j) + r(m_1,j)\bigg] \\
    &- \bigg[r(0,j) + r(m,j) + r(m_1,j)\bigg]^2 \bigg[r(m,0) + r(m,j)\bigg] \\
    &+ \bigg[r(m,0) + r(m,j)\bigg] r(m,j) \bigg[r(0,j) + r(m,j) + r(m_1,j)\bigg]\\
    &- \bigg[r(m,0) + r(m,j)\bigg]^2 \bigg[r(0,j) + r(m,j) + r(m_1,j)\bigg]\\
    &+  \bigg[r(0,j) + r(m,j) + r(m_1,j)\bigg] \bigg[r(m,0) + r(m,j)\bigg] r(m,j)  \\
    &-  \bigg[r(0,0) + r(0,j) + r(m_1,0) \\
    &+ r(m_1,j) + r(m,0) + r(m,j)\bigg] r(m,j)^2 \bigg\} \\
    &- \bigg\{\bigg[r(m_1,0) + r(m_1,j)\bigg] \bigg[r(m,0) + r(m,j)\bigg]\\
    &\times \bigg[r(0,j) + r(m,j) + r(m_1,j)\bigg] \\
    &- r(m_1,j) \bigg[r(m,0) + r(m,j)\bigg] \bigg[r(0,j) + r(m_1,j) + r(m,j)\bigg] \\
    &+ r(m_1,j) r(m,j) \bigg[r(m,0) + r(m,j)\bigg] \\
    &- r(m,j)^2 \bigg[r(m_1,0) + r(m_1,j)\bigg] \bigg\} \\
    &- r(m_1,j) \bigg\{ \bigg[r(m_1,0) + r(m_1,j)\bigg] \bigg[r(m,0) + r(m,j)\bigg] r(m,j) \\
    &- r(m_1,j) \bigg[r(m,0) + r(m,j) \bigg]^2 \\
    &+ r(m_1,j) \bigg[r(0,0) + r(0,j) + r(m_1,0) \\
    &+ r(m_1,j) + r(m,0) + r(m,j)\bigg] \\
    &\times \bigg[r(m,0) + r(m,j) \bigg] \\
    &- \bigg[r(m_1,0) + r(m_1,j)\bigg] \bigg[r(m,0) + r(m,j)\bigg] \\
    &\times \bigg[r(0,j) + r(m,j) + r(m_1,j) \bigg]\bigg\} \\
    &= r(m_1,0) \bigg[r(0,0) r(m,0) r(m_1,j) + r(0,0) r(m,j) r(m_1,j) \bigg] \\
    &+ r(m_1,j) \bigg[r(m_1,0) r(m,0) r(0,j) + r(m_1,0) r(m,j) r(0,j) \bigg] \\
    &\times \bigg[r(m_1,0) + r(m_1,j)\bigg] 
    \bigg[r(0,0) r(m,0) r(0,j) \\
    &+ r(0,0) r(m,0) r(m,j)  + r(0,0) r(m,j) r(0,j) \\
    &+ r(0,j) r(m,0) r(m,j) \bigg].
\end{align*}

The different entries of the adjugate matrix that are needed to compute the exposure weights are as follows:
\begin{align*}
    &\tilde{A}_{1,1} \\ 
    =& r(0,0) r(m,0) r(0,j) + r(0,0) r(m,0) r(m_1,j) + r(0,0) r(m,0) r(m,j) \\
    &+ r(0,0) r(m,j) r(0,j) + r(0,0) r(m,j) r(m_1,j) + r(m_1,0) r(m,0) r(0,j) \\
    &+ r(m_1,0) r(m,0) r(m_1,j) + r(m_1,0) r(m,0) r(m,j) + r(m_1,0) r(m,j) r(0,j) \\
    &+ r(m_1,0) r(m,j) r(m_1,j) + r(0,j) r(m,0) r(m,j) + r(m_1,j) r(m,0) r(m,j) \\
    &\tilde{A}_{2,1} \\
    =& - \bigg[ r(m_1,0) r(m,0) r(0,j) + r(m_1,0) r(m,0) r(m,j) + r(m_1,0) r(m,0) r(m_1,j) \\
    &+ r(m_1,0) r(m,j) r(0,j) + r(m_1,0) r(m,j) r(m_1,j) + r(m_1,j) r(m,0) r(m,j) \bigg] \\
    &\tilde{A}_{3,1} \\
    =& r(m_1,0) r(m,0) r(0,j) + r(m_1,0) r(m,0) r(m,j) + r(m_1,0) r(m,0) r(m_1,j)  \\
    &+ r(m_1,j) r(m,j) r(0,0) + r(m_1,j) r(m,j) r(m_1,0) + r(m_1,j) r(m,0) r(m,j) \\
    &\tilde{A}_{4,1} \\
    =& \bigg[r(m,0) + r(m,j)\bigg] \bigg[r(m_1,0) r(0,j) - r(m_1,j) r(0,0) \bigg].
\end{align*}

Using the adjugate entries and the determinant, the weights $\alpha_1, \alpha_2, \alpha_3$ are then as follows:
\begin{align*}
    \alpha_1 &= r(m_1,0)\bigg\{r(0,0) r(m,0) r(0,j) + r(0,0) r(m,0) r(m_1,j)\\
    &+ r(0,0) r(m,0) r(m,j) + r(0,0) r(m,j) r(0,j)\\
    &+ r(0,0) r(m,j) r(m_1,j) + r(0,j) r(m,0) r(m,j) \bigg\}\\
    & \times \bigg\{r(m_1,0) \bigg[r(0,0) r(m,0) r(m_1,j) + r(0,0) r(m,j) r(m_1,j) \bigg] \\
    &+ r(m_1,j) \bigg[r(m_1,0) r(m,0) r(0,j) + r(m_1,0) r(m,j) r(0,j) \bigg] \\
    &\bigg[r(m_1,0) + r(m_1,j)\bigg] 
    \bigg[r(0,0) r(m,0) r(0,j) + r(0,0) r(m,0) r(m,j)  \\
    &+r(0,0) r(m,j) r(0,j) + r(0,j) r(m,0) r(m,j) \bigg]\bigg\}^{-1} \numberthis \\
    \alpha_2 &= r(0,j)\bigg\{r(m,0) r(m_1,j) r(0,0) + r(m,j) r(m_1,j) r(0,0) \\
    &+ r(m_1,0) r(m,0) r(m,j)  + r(m_1,0) r(m,0) r(m_1,j) \\
    &+ r(m_1,0) r(m,j) r(m_1,j) + r(m_1,j) r(m,0) r(m,j)\bigg\}\\
    & \times \bigg\{r(m_1,0) \bigg[r(0,0) r(m,0) r(m_1,j) + r(0,0) r(m,j) r(m_1,j) \bigg] \\
    &+ r(m_1,j) \bigg[r(m_1,0) r(m,0) r(0,j) + r(m_1,0) r(m,j) r(0,j) \bigg] \\
    &\bigg[r(m_1,0) + r(m_1,j)\bigg] 
    \bigg[r(0,0) r(m,0) r(0,j) + r(0,0) r(m,0) r(m,j)  \\
    &+r(0,0) r(m,j) r(0,j) + r(0,j) r(m,0) r(m,j) \bigg]\bigg\}^{-1} \numberthis \\
    \alpha_3 &= 1 - \alpha_1 - \alpha_2 \\
    &= r(m,0)r(m,j)\bigg\{r(m_1,j) r(0,0) - r(m_1,0) r(0,j) \bigg\}\\
    & \times \bigg\{r(m_1,0) \bigg[r(0,0) r(m,0) r(m_1,j) + r(0,0) r(m,j) r(m_1,j) \bigg] \\
    &+ r(m_1,j) \bigg[r(m_1,0) r(m,0) r(0,j) + r(m_1,0) r(m,j) r(0,j) \bigg] \\
    &\bigg[r(m_1,0) + r(m_1,j)\bigg] 
    \bigg[r(0,0) r(m,0) r(0,j) + r(0,0) r(m,0) r(m,j)  \\
    &+r(0,0) r(m,j) r(0,j) + r(0,j) r(m,0) r(m,j) \bigg]\bigg\}^{-1} \numberthis.
\end{align*}
We focus on the $\alpha_3$ weight, and we want to find the maximum of this weight. 
Based on an informal analysis of the partial derivatives, the $\alpha_3$ weight is maximized when $\V(\alpha), \V(\theta_{1,m}) \to 0$, $\V(\theta_{1,m_1}) \to \infty$, and $\V(\theta_{2,j}) < \infty$.
We now compute the limit of $\alpha_3$ when $\V(\alpha), \V(\theta_{1,m}) \to 0$ and $\V(\theta_{1,m_1}) \to \infty$.

\noindent We consider the case when $r(m_1,j) r(0,0) - r(m_1,0) r(0,j) > 0$.
We first take $\lim \V(\theta_{1,m}) \to 0$:
\begin{align*}
    &\lim_{\V(\theta_{1,m}) \to 0} \alpha_3\\
    =& \lim_{\V(\theta_{1,m}) \to 0}  \bigg\{r(m_1,j) r(0,0) - r(m_1,0) r(0,j) \bigg\}\\
    & \scalemath{1}{\times \bigg\{r(m_1,0) \bigg[r(0,0) \frac{\V(Y(m,j))}{p(m,j)} r(m_1,j) + r(0,0) \frac{\V(Y(m,0))}{p(m,0)} r(m_1,j) \bigg]} \\
    &+ \scalemath{1}{r(m_1,j) \bigg[r(m_1,0) \frac{\V(Y(m,j))}{p(m,j)} r(0,j) + r(m_1,0) \frac{\V(Y(m,0))}{p(m,0)} r(0,j) \bigg]} \\
    &\bigg[r(m_1,0) + r(m_1,j)\bigg] 
    \bigg[r(0,0) \frac{\V(Y(m,j))}{p(m,j)} r(0,j) + r(0,0)  \\
    &+r(0,0) \frac{\V(Y(m,0))}{p(m,0)} r(0,j) + r(0,j)  \bigg]\bigg\}^{-1} \\
    =& \lim_{\V(\theta_{1,m}) \to 0}  \bigg\{r(m_1,j) r(0,0) - r(m_1,0) r(0,j) \bigg\}\\
    & \times \bigg\{r(m_1,0) \bigg[r(0,0) \frac{\V(\alpha) + \V(\theta_{2,j})}{p(m,j)} r(m_1,j) + r(0,0) \frac{\V(\alpha)}{p(m,0)} r(m_1,j) \bigg] \\
    &+ r(m_1,j) \bigg[r(m_1,0) \frac{\V(\alpha) \V(\theta_{2,j})}{p(m,j)} r(0,j) + r(m_1,0) \frac{\V(\alpha)}{p(m,0)} r(0,j) \bigg] \\
    &\bigg[r(m_1,0) + r(m_1,j)\bigg] 
    \bigg[r(0,0) \frac{\V(\alpha) + \V(\theta_{2,j})}{p(m,j)} r(0,j) + r(0,0)  \\
    &+r(0,0) \frac{\V(\alpha)}{p(m,0)} r(0,j) + r(0,j)  \bigg]\bigg\}^{-1} \\
    =& \bigg\{p(m,j) p(m,0) \bigg[p(m_1,j) p(0,0) \V(Y(m_1,0)) \V(Y(0,j)) \\
    &- p(m_1,0) p(0,j) \V(Y(m_1,j)) \V(Y(0,0)) \bigg] \bigg\}\\
    &\times \bigg\{\V(Y(0,j)) \bigg[p(m_1,0) p(0,0) \V(Y(0,j)) p(m_1,j) p(m,0) \\
    &+ p(0,0) p(m_1,j) \V(Y(0,0)) p(m,j) p(m_1,0)\bigg] \\
    &+ \V(Y(0,0)) \bigg[p(m_1,j) p(m_1,0) p(0,j) p(m_1,0) \V(Y(0,j)) \\
    &+ p(m_1,0) \V(Y(0,0)) p(0,j) p(m_1,j) p(m,j) \bigg] \\
    &+ \V(Y(m_1,j)) \bigg[p(m_1,0) p(0,0) p(0,j) p(m,0) \V(Y(0,j)) \\
    &+ p(m_1,0) p(0,0) p(m,j) p(m,0) \V(Y(0,j)) \\
    &+ p(0,0) p(0,j) p(m_1,0) \V(Y(0,0)) p(m,j) \\
    &+ p(0,j) p(m_1,0) \V(Y(0,0)) p(m,j) p(m,0)\bigg] \\
    &+ \V(Y(m_1,0)) \bigg[p(m_1,j) p(0,0) p(0,j) p(m,0) \V(Y(0,j)) \\
    &+ p(m_1,j) p(0,0) p(m,j) p(m,0) \V(Y(0,j)) \\
    &+ p(0,0) p(0,j) p(m_1,j) \V(Y(0,0)) p(m,j) \\
    &+ p(0,j) p(m_1,j) \V(Y(0,0)) p(m,j) p(m,0)\bigg]\bigg\}^{-1}.
\end{align*}

Then, we take the limit of the term as $\V(\theta_{1,m_1}) \to \infty$. 
However, since $\V(\theta_{1,m_1})$ appears in both the numerator and denominator, the limit will lead to $\frac{\infty}{\infty}$.
Thus, we use L'Hopital's rule and take the limit of the partial derivative of the numerator and denominator with respect to $\V(\theta_{1,m_1})$ as $\V(\theta_{1,m_1}) \to \infty$:
\begin{align*}
    & \lim_{\substack{\V(\theta_{1,m}) \to 0 \\\V(\theta_{1,m_1}) \to \infty}} a_3 \\
    =& \bigg\{p(m,j) p(m,0) \bigg[p(m_1,j) p(0,0) \V(Y(0,j)) \\
    &- p(m_1,0) p(0,j) \V(Y(0,0)) \bigg] \bigg\}\\
    &\times \bigg\{\V(Y(0,j)) \bigg[p(m_1,0) p(0,0) p(0,j) p(m,0)  \\
    &+ p(m_1,0) p(0,0) p(m,j) p(m,0) \\
    &+ p(m_1,j) p(0,0) p(0,j) p(m,0) + p(m_1,j) p(0,0) p(m,j) p(m,0)\bigg]\\
    &+ \V(Y(0,0)) \bigg[p(0,0) p(0,j) p(m_1,0)  p(m,j) \\
    &+ p(0,j) p(m_1,0) p(m,j) p(m,0) \\
    &+ p(0,0) p(0,j) p(m_1,j) p(m,j) + p(0,j) p(m_1,j) p(m,j) p(m,0) \bigg] \bigg\}^{-1}.
\end{align*}
To maximize the limit of the term, we can set $\V(\alpha) \to 0$ so that we are not subtracting any terms. 
Note that in addition, we would need $\V(\theta_{2,j}) < \infty$.
Then, taking the limit as $\V(\alpha) \to 0$, we get:
\begin{align*}
    &\lim_{\substack{\V(\theta_{1,m}) \to 0 \\ \V(\theta_{1,m_1}) \to \infty \\ \V(\alpha) \to 0}} a_3 \\
    =& \bigg\{p(m,j) p(m,0) p(m_1,j) p(0,0) \V(\theta_{2,j}) \bigg\} \times \\
    & \bigg\{\V(\theta_{2,j}) \bigg[p(m_1,0) p(0,0) p(0,j) p(m,0) + p(m_1,0) p(0,0) p(m,j) p(m,0) \\
    &+ p(m_1,j) p(0,0) p(0,j) p(m,0) + p(m_1,j) p(0,0) p(m,j) p(m,0)\bigg]\bigg \}^{-1}.
\end{align*}
Since there is a $\V(\theta_{2,j})$ in both the denominator and numerator, we get the following:
\begin{align*}
    &\lim_{\substack{\V(\theta_{1,m}) \to 0 \\ \V(\theta_{1,m_1}) \to \infty \\\V(\alpha) \to 0}} a_3 \\
    =& \frac{p(m,j) p(m_1,j)}{p(m_1,0) p(0,j)+ p(m_1,0)  p(m,j) + p(m_1,j)  p(0,j) + p(m_1,j)  p(m,j)}.
\end{align*}
\end{proof}

\section{Simulations from an Erd\"{o}s-r\'{e}nyi Network}
We also sampled networks from an Erd\"{o}s-r\'{e}nyi distribution where the probability of an edge is 0.25 (denoted as ER(0.25)).
In particular, we sampled an ER(0.25) directed network of sizes $n = 10, 20, \dotsc, 50$.
Figure \ref{fig:er_quarter_network} shows a directed network with 40 nodes.
Note that in an ER(0.25) graph, units may have different degrees, with an expected degree being $(n-1) \times 0.25$.
Hence, an ER(0.25) graph is generally denser than a $k$-regular graph.
Since units have different degrees, each unit is affected differently by other units, and so unlike in a $k$-regular graph, each unit may contribute to the estimate of the average interference effect differently in an ER(0.25) network. 

\begin{figure}[htb]
    \center
  \begin{minipage}[t]{\linewidth}\centering
    \includegraphics[width=12cm]{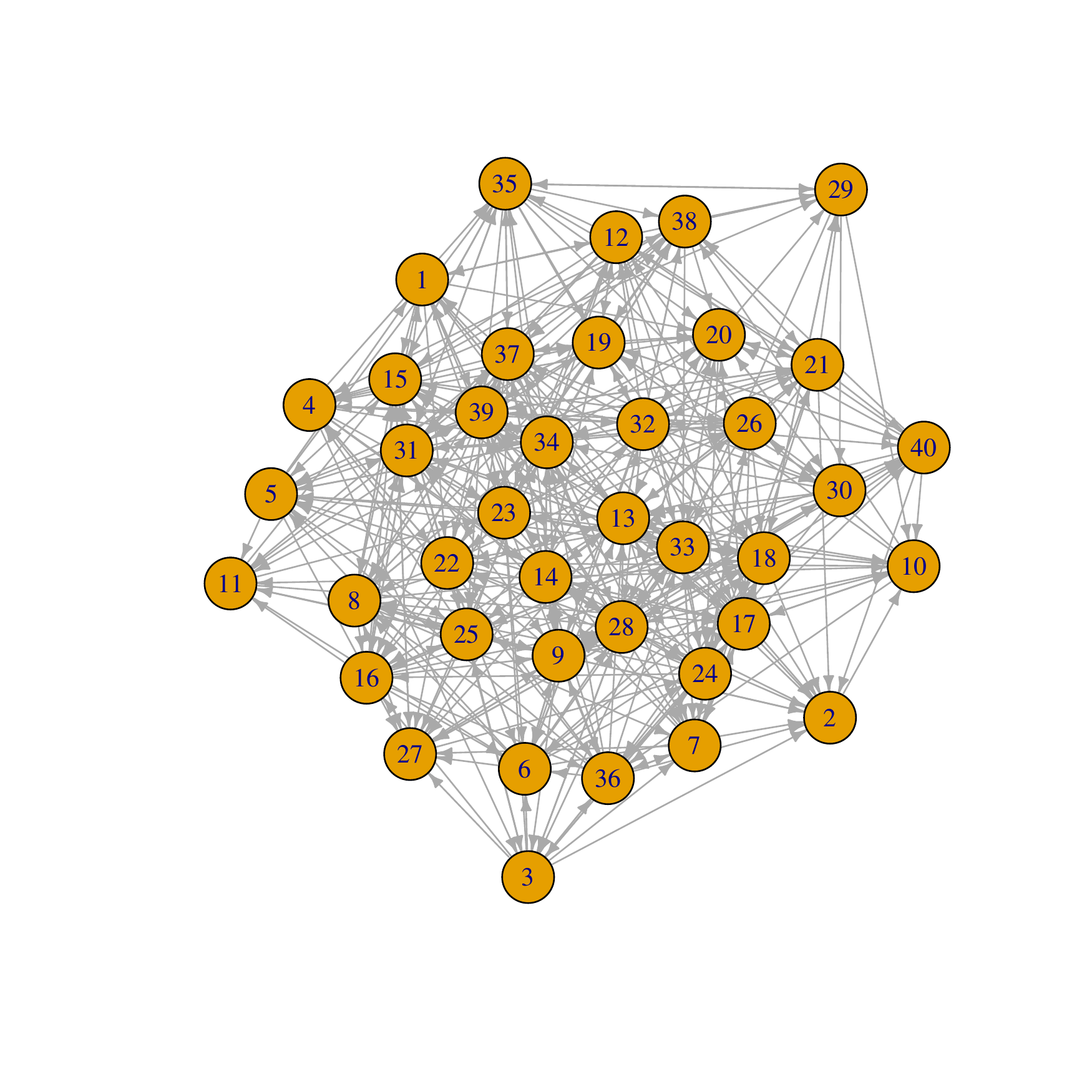}
  \end{minipage}\hfill
    \caption{Directed Erd\"{o}s-r\'{e}nyi network with 40 nodes and probability of an edge is 0.25.}
    \label{fig:er_quarter_network}
\end{figure}

Figure \ref{fig:imse_num_nodes_er_quarter} shows the IMSEs for the different estimators as the number of units increases when the true mean interference effect is zero and additivity holds.
Note that as the number of units increases, the number of edges also increases in an Erd\"{o}s-r\'{e}nyi network.
Hence, the IMSEs increase with the number of units, unlike in the $k$-regular graph. 
Instead, the increases in IMSEs are similar to the case of the $k$-regular graphs when the graph becomes denser.
Furthermore, the IMSEs of the estimators in the ER(0.25) network are higher than the IMSEs in the $k$-regular graphs.
However, in general, $M_{Ind}$, $M_{Dil}$, and $HT_{Avg}$ still outperform the two-term estimators, with the IMSE of $M_{Ind}$, $M_{Dil}$, and $HT_{Avg}$ being very close as in the case of the $k$-regular network.

\begin{figure}[htb]
    \center
  \begin{minipage}[t]{\linewidth}\centering
    \includegraphics[width=12cm]{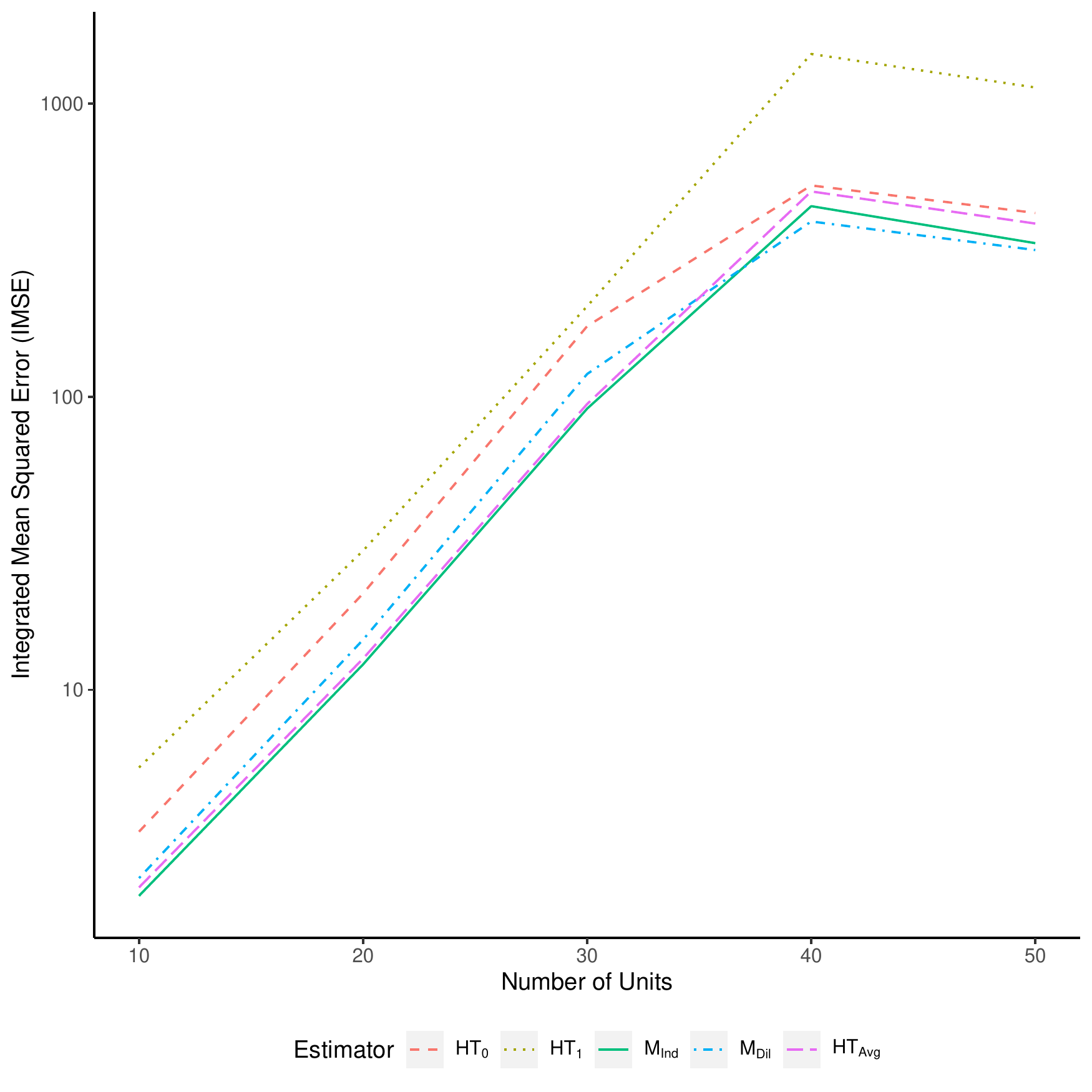}
  \end{minipage}\hfill
    \caption{IMSE for estimators (indicated by color and line type) as the number of units vary (indicated by x-axis) when additivity holds and mean interference effect is zero for ER(0.25) network.}
    \label{fig:imse_num_nodes_er_quarter}
\end{figure}

Figure \ref{fig:imse_mean_interference_and_interaction_er} shows the IMSEs for the estimators for different interference and interaction effect sizes for a 40-node ER(0.25) network.
Again, the IMSEs are generally higher than the IMSEs in the $k$-regular graphs.
As in the $k$-regular network, the IMSEs of all estimators increase as the mean interference increases since we assumed a zero-mean prior for the parameters.
There are some instances when the multi-term MIV LUEs outperform $HT_{0}$, such as when the interference and interaction effect is low.
However, unlike in the $k$-regular network, as the mean interference effect increases, the multi-term MIV LUEs have higher IMSEs than $HT_0$ besides $M_{Dil}$.
This suggests that in the presence of heterogeneity in the degree distributions of the nodes, the multi-term MIV LUEs are not as robust to additivity as in the case when the degree distributions are more homogenous.
Despite this, the multi-term MIV LUEs still outperform $HT_{Avg}$ and $HT_1$, suggesting that there might still be some benefit in using the multi-term MIV LUEs.

\begin{figure}[htb]
    \center
  \begin{minipage}[t]{\linewidth}\centering
    \includegraphics[width=12cm]{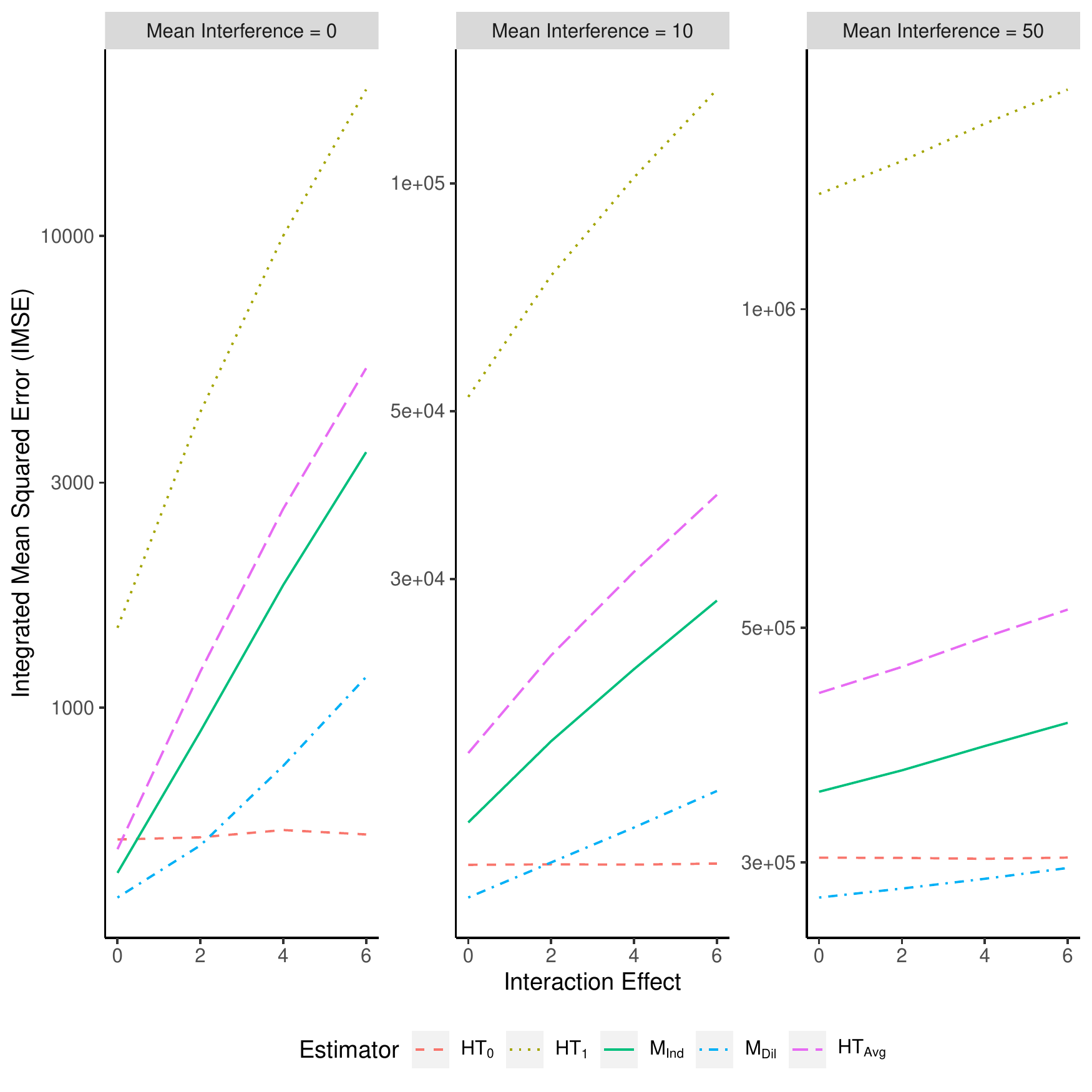}
  \end{minipage}\hfill
    \caption{IMSE for estimators (indicated by color and line type) as the interaction effect varies (indicated by x-axis) for different mean interference effects (indicated by the panels) for ER(0.25) network with 40 nodes.}
    \label{fig:imse_mean_interference_and_interaction_er}
\end{figure}

\end{document}